\documentclass{article}
\usepackage{tikz}
\usepackage{subcaption}
\usepackage{amsfonts}
\usepackage{amssymb}
\usepackage{graphicx}
\usepackage{amsmath}
\usepackage{amsthm}
\usepackage{pifont}
\usepackage{xcolor}
\usepackage{enumitem}

\usepackage[normalem]{ulem}

\usepackage[colorlinks=true,linkcolor=blue,citecolor= red,bookmarksopen=true,urlcolor=blue]{hyperref}
\nonstopmode
\newtheorem{theorem}{Theorem}[section]
\newtheorem{proposition}[theorem]{Proposition}
\newtheorem{lemma}[theorem]{Lemma}
\newtheorem*{lemma*}{Lemma}

\newtheorem{corollary}[theorem]{Corollary}

\newtheorem{definition}[theorem]{Definition}
\newtheorem{notation}[theorem]{Notation}
\newtheorem*{standingassumptionI*}{Standing Assumption I}
\newtheorem*{standingassumptionII*}{Standing Assumption II}
\newtheorem{assumption}[theorem]{Assumption}

\theoremstyle{remark}
\newtheorem{example}[theorem]{Example}
\newtheorem{remark}[theorem]{Remark}

\newcommand{\R}{{\mathbb R}}
\newcommand{\N}{{\mathbb N}}

\newcommand\abs[1]{\left|#1\right|}

\newcommand{\fit}{{\mathcal{F}^{i}_t}}
\newcommand{\fiT}{{\mathcal{F}^{i}_T}}
\newcommand{\fib}{{\mathbb{F}^{i}}}
\newcommand{\hi}{{\mathcal{H}^{i}}}
\newcommand{\ki}{{K_{i}}}

\newcommand{\mi}{{M_{i}}}
\newcommand{\pii}{\rho_{i,+}}

\newcommand{\hatt}{\widehat{t}}
\newcommand{\somma}{\mathcal{S}}

\renewcommand{\emph}[1]{\textit{#1}}
\newcommand{\red}[1]{\textcolor{red}{#1}}

\newcommand{\blue}[1]{{\color{blue}#1}}

\newcommand{\mcY}{\mathcal {Y}}
\newcommand{\mcF}{\mathcal {F}}

\newcommand{\mcR}{\rho^{\mathcal Y}_{+}}
\newcommand{\ro}{\rho^{N}_{+}}
\newcommand{\cmark}{{\color{green}\ding{51}}}
\newcommand{\xmark}{{\color{red}\ding{55}}}

\newcommand{\bfone}{\mathbf{1}}
\newcommand{\hatk}{\widehat{k}}
\newcommand{\hatY}{\widehat{Y}}

\begin{document}

\title{Collective Arbitrage and the Value of Cooperation}
\date{\today}

\author{
Francesca Biagini\thanks{Department of Mathematics, University of Munich, Theresienstraße 39, 80333 Munich, Germany.
\emph{francesca.biagini@math.lmu.de}.},
Alessandro Doldi\thanks{Dipartimento di Scienze per l'Economia e l'Impresa, Università degli Studi di Firenze, Via delle Pandette 32, 50127 Firenze, Italy,
\emph{alessandro.doldi@unifi.it}. },
Jean-Pierre Fouque\thanks{Department of Statistics \& Applied Probability, University of California, Santa Barbara, CA 93106-
3110, \emph{fouque@pstat.ucsb.edu}. }, \\ Marco Frittelli\thanks{Dipartimento di Matematica, Università degli Studi di Milano, Via Saldini 50, 20133 Milano, Italy,
\emph{marco.frittelli@unimi.it}.  }, 
Thilo Meyer-Brandis\thanks{Department of Mathematics, University of Munich, Theresienstraße 39, 80333 Munich, Germany.
\emph{meyerbr@math.lmu.de}.}
}

\maketitle

\begin{abstract}
\noindent We introduce the notions of \emph{Collective Arbitrage} and of \emph{Collective Super-replication} in a discrete-time setting where agents are investing in their markets and are allowed to cooperate through exchanges.  We accordingly establish versions of the fundamental theorem of asset pricing and of the pricing-hedging duality. A reduction of the price interval of the contingent claims can be obtained by applying the collective super-replication price.

\end{abstract}


\section{Introduction}

The theory developed in this paper aims at expanding the classical Arbitrage Pricing Theory to a setting where $N$ agents are investing in stochastic security markets and  are allowed to cooperate through suitable exchanges. More precisely, we suppose that each agent 
is allowed to invest in a subset of the available assets $(X^1,\dots,X^J)$, for a given $J\in\mathbb{N}$, 
and in a common riskless asset. 
Note that we do not exclude that such subset coincides with the full set $(X^1,\dots,X^J)$.
The novel notions of Collective Arbitrage and Collective Super-replication, are  based on the possibility that the $N$ agents may additionally enter in a zero-sum risk exchange mechanism, where no money is injected or taken out of the overall system. Cooperation and the multi-dimensional aspect are the key features  of Collective Arbitrage and Collective Super-replication. 
In this setting agents not only may invest in their respective market but  may additionally cooperate to improve their positions by taking advantage of the risk exchanges. 
In the case of one single agent, the theory reduces to the classical one.
As we mention in the Section \ref{LR}, there is an extensive literature in recent years on variations around the concept of one-agent No Arbitrage or No Free Lunch.

Departing from this stream of literature, the main aim of this paper is to understand the consequences of the cooperation between several agents in relation to Arbitrage and Super-replication.
\newline
Before moving into the details of our new setup, we briefly summarize the classical one-agent situation.
Let a filtered probability space $(\Omega, \mathcal{F}, \mathbb F ,P)$, with  $\mathbb F =\{\mathcal{F}_t\}_{t \in  \mathcal T}$, $\mathcal T=\{1,\dots,T\} $,
be given, and
denote by $X=(X^1,\dots,X^J)$ the $J$ adapted stochastic processes representing the prices of $J$ securities. 
The set of admissible trading strategies consists of predictable processes (see page \pageref{admisstrat} for a precise definition) and is denoted by $\mathcal H$ and we let $K$ be the set of  time-$T$ stochastic integral of $H \in \mathcal H$ with respect to $X$. The set $K$ represents all the possible terminal time-$T$ payoffs available in the market from admissible trading strategies and having zero initial cost. 

An arbitrage opportunity  is an admissible trading strategy $H \in \mathcal H$, having zero initial cost and producing a non negative final payoff $k \in K$, being strictly positive with positive probability.  
Equivalently, we have no arbitrage in this setting if the only non negative element in $K$ is $P$-a.s. equal to $0$, or more formally $ K \cap L^{0 }_+(\Omega,\mcF,P)=\{0\}$.\newline
We now present  the main concepts that we are introducing in this paper, without giving full details. The precise setup, as well as rigorous assumptions and formulations of the results, are stated in the subsequent sections.  

\subsection{Collective Arbitrage}\label{CA}

Since each agent $i=1,\dots,N$ is allowed to invest only in a subset of the available assets $(X^1,\dots,X^J)$,  
we define, similarly to the notion of the set $K$, the market $K_i$ of agent $i$, that is the space of all the possible time-$T$ payoffs that agent $i$ can obtain by using admissible trading strategies in his/her allowed investments and having zero initial cost. \\
We consider the set of all zero-sum risk exchanges 
\begin{equation}\label{2345}
\mathcal{Y}_0=\left\{Y \in  (L^{0 }(\Omega, \mathcal{F},P))^N \mid \sum_{i=1}^N Y^i = 0 \,\, P\text {-a.s.}\right\},
\end{equation} and the set $\mcY$ of possible\slash allowed exchanges $$\mcY \subseteq \mathcal{Y}_0  \, \text{ such that } \, 0 \in \mcY.$$ We stress  that even if the overall sum is $P$-a.s.~equal to $0$,  each components $Y^i$ of the vector $Y \in \mcY$ is in general a random variable. If $Y^i$ is positive on some event, agent $i$ is receiving, on that event, from the collection of the other agents the amount $Y^i$ of capital. If $Y^i$ is negative on some event, agent $i$ is providing, on that event, to the collection of the other agents the amount $Y^i$ of capital. So $Y \in \mcY$ represents the capital that the agents may exchange among themselves with the requirement that the overall amount distributed is equal to zero. To implement a risk exchange $Y \in \mcY$, the agents agree at the initial time by a binding agreement to exchange at terminal time $T$ the capital amounts implied by $Y$.  
\newline
A \emph{Collective Arbitrage} is a vector $(k^1, \dots , k^N)$, where $k^i \in K_i $ for each $i$, and a vector $Y =(Y^1, \dots, Y^N) \in \mcY$ such that
 \begin{align*}
k^i+Y^i&\geq 0, \quad P\text{-a.s.}\text{ for all }i\in\{1,\dots,N\},
\end{align*}
and  $$P(k^j+Y^j>0)>0 \, \text{ for at least one } j \in \{1,\dots,N\} .$$\\
One may immediately notice that if $N=1$, then $Y \in \mcY$ must be equal to $0$ and thus a Collective Arbitrage reduces to a Classical Arbitrage. \newline
However, for $N\geq 2$, in a Collective Arbitrage, agents are entangled by the vector of exchanges $Y \in \mcY$:
this additional possible cooperation may create a Collective Arbitrage even if there is No Arbitrage for each single agent. 
Observe that even if some agents - in some states of the world - will lend some amounts to other agents, in a Collective Arbitrage the condition  $k^i + Y^i \geq 0 $ holds for all $i$ (for all agents), namely: no agent incurs any loss.
\newline
We study the implications of the assumption of No Collective Arbitrage with respect to the set $\mcY$, which we  denote in short  by $\mathbf{NCA(\mathcal{Y})}$, or by $\mathbf{NCA}$ if $\mathcal{Y}$ is not essential. We also write $\mathbf{NA}_{i} $ for the No Arbitrage condition (in the classical sense) for agent $i$ in market $K_i$ and $\mathbf{NA}$ for the No Arbitrage condition (in the classical sense) in the global market $K$, where it is possible to invest in all the assets. 
 In Item 6 of Section \ref{secCA}  we show that, when considering the possibility that one agent invests in the global market, then \textbf{NA} and \textbf{NCA} are equivalent notions. \label{refperdopo}Thus  the difference between \textbf{NA} and \textbf{NCA} can be fully appreciated considering segmented markets.

It is easy to verify (see Section \ref{secCA}) that under very reasonable conditions on $\mcY$ the following implications  hold
\begin{equation*}\label{Implications}
\mathbf{NA} \Rightarrow  \mathbf{NCA(\mathcal{Y})} \Rightarrow \mathbf{NA}_{i} \text{ } \forall i\in\{1,\dots,N\},
\end{equation*}
but none of the reverse implications holds true in general. We
 show that the strongest condition $\mathbf{NA}$ is equivalent to $\mathbf{NCA(\mathcal{Y})}$ for the \textquotedblleft largest\textquotedblright choice $\mcY=\mcY_0$, while 
the weakest condition, $\mathbf{NA}_{i} \text{ } \forall i$,  is equivalent to $\mathbf{NCA(\mathcal{Y})}$ for the \textquotedblleft smallest\textquotedblright choice $\mcY=\mcY_0\cap (L^{0 }(\Omega, \mathcal{F}_0,P))^N$. The latter space actually consists of zero-sum deterministic vectors, when $\mathcal{F}_0$ is the trivial sigma algebra. However, for general sets $\mcY$ the notions of $\mathbf{NCA(\mathcal{Y})}$ give rise to new concepts.

\subsubsection{Economic Motivation of the Assumption of $\mathbf{NCA(\mathcal{Y})}$ }
\label{economicmotiv}

Various conceivable situations are covered by the Collective Arbitrage framework: the decision to cooperate and engage in an exchange might be based, for example, on rational behavior from the individual agents' points of view (as explained below), or on the policies of a mother institution forcing subunits to cooperate (e.g. several trading desks of a financial company), or on the regulatory policies of a supervising authority forcing banks in a network to cross-subsidize, or on taxation.

We stress that, similar to the case of classical arbitrage, efficient markets with risk exchanges should not allow for the possibility of collective arbitrage, and one of the main objectives of our paper is to study and characterize markets with $\mathbf{NCA(\mathcal{Y})}$.\\
We can distinguish two cases of collective arbitrage, the second being the most relevant one,
\begin{enumerate}
    \item There exists at least one agent who has an individual (classical) arbitrage opportunity in her market. Together with zero strategies of the other agents this would also be a collective arbitrage, which should obviously be ruled out in an efficient market. 
   \item More interestingly, no agent has an individual arbitrage opportunity but there exists a ``true" collective arbitrage, where agents can achieve arbitrages by using the exchange opportunities. In this situation it is not possible that one agent gives away capital to the other agents in \textit{every} scenario, but every agent truly profits from the cooperation in the sense that she receives some positive amount from the other agents in at least one scenario. The example below addresses this case. In any case, also these collective arbitrages would be exploited by the agents, leading to the assumption that efficient markets should not allow for collective arbitrage possibilities, just as in the classical arbitrage case.
\end{enumerate}

\begin{example}

In the framework of case 2. above, we show how a Collective Arbitrage can be implemented as rational choice for each agent.  In the simple setting described in Example  \ref{toyex}, two agents can obtain in their respective (complete) markets the following terminal time payoffs at \textbf{zero initial cost}
$$
k^1= 
\begin{cases}
    -5, &  \text{ on } A \\
    +5, &  \text{ on } A^c
\end{cases}
\quad 
k^2= 
\begin{cases}
    10, &  \text{ on } A \\
    -2, &  \text{ on } A^c
\end{cases}
$$
where $A$ is a measurable event with $0< P(A) < 1$. Notice that $E_{Q^1}[k^1]=E_{Q^2}[k^2]=0 $ for the unique equivalent martingale measure $Q^i$ in the market $K_i$. Clearly there are monotone preferences (and values of $P(A)$) for which $k^1$ is preferable to $k^2$ and others for which the opposite happens. Now consider the exchanges $(Y^1,Y^2)$, satisfying $Y^1+Y^2=0$, given by

$$
Y^1= 
\begin{cases}
    +5, &  \text{ on } A \\
    -2, &  \text{ on } A^c
\end{cases}
\quad 
Y^2=
\begin{cases}
    -5, &  \text{ on } A \\
    2, &  \text{ on } A^c
\end{cases}
$$
which generate the collective arbitrage
$$
k^1+Y^1=
\begin{cases}
    0, &  \text{ on } A \\
    3, &  \text{ on } A^c
\end{cases}
\quad 
k^2+Y^2= 
\begin{cases}
    5, &  \text{ on } A \\
    0, &  \text{ on } A^c
\end{cases}
$$
With this exchange, agent $2$ gives part ($5\$$) of his profit to agent $1$ in state $A$ and receives $2 \$ $ on the complement of $A$.
There are monotone preferences for which \textit{both agents are better off after the exchanges than before}.
Take for example preferences represented by expected utility where the utility functions belong to the class $\mathbb U:=\{ u:\mathbb R \rightarrow \mathbb R \cup \{-\infty\} \mid u$ is increasing on $\mathbb R$ and  $u(x)=  -\infty$ on $(-\infty,0)\}$. 
Thus a collective arbitrage could be viewed as an advantageous opportunity by all agents.
Our assumption of \textbf{NCA} is ruling out such possibilities.
\end{example}

\label{remimplement}
    To understand how a Collective Arbitrage is implemented, consider in the above example the two agents which agree at initial time to exchange at the terminal time precisely the $Y^1$ and $Y^2$ above. Then each one of the agents will freely decide her investment strategy - based on her own personal information -  which will generate some final payoff at terminal time. If the agents  selected the investment strategies having the final payoff $k^1$ and $k^2$ above (and this would be a rational choice as explained before),  then this would determine a collective arbitrage. Our assumption of \textbf{NCA} implies that such agreement and the possibility to generate $k^1$ and $k^2$ are not possible.  Any other possible agreement to exchange money which could give rise to a \textbf{CA} is similarly impossible.\\
    One could ask why the exchange $(Y^1,Y^2)$ is feasible, namely where does each agent $i$ find the capital  $|Y^i|$, for  $Y^i\leq 0$, that he needs to \textit{provide} to the other agents? And which is, for each agent, its cost?
    Observe that a collective arbitrage implies $k^i+Y^i \geq 0$ for each $i$. Thus $k^i \geq |Y^i|=-Y^i$ on the set $\{Y^i \leq 0 \}$. Each exchange $Y^i$, on the set where it is negative, can be financed from agent $i$, \textit{at zero initial cost}, by investing in his/her market $K_i$. In the example above, agent $2$ can generate at least $5 \$ $ in the event $A$ (on the event $A$ the exchange $Y^2$ is negative so that agent $2$ has to provide $|Y^2|=5 \$ $ to agent 1), while agent $1$ can generate at least $2 \$ $ in the complement event $A^c$. Since each agent $i$ knows that he/she can finance the exchange $Y^i$ (when negative) at zero initial cost, it is reasonable for each agent to enter at time zero in the agreement $(Y^1,Y^2)$. 

\begin{remark}\label{rem:semistatic}
 
Classically, terminal payoffs with zero initial cost from dynamic admissible trading strategies are represented by time $T$ stochastic integrals in the primal market, that is by elements $k \in K$. Often this setup was extended to admit also semi-static strategies, which leads to the terminal payoffs of the form $k+\sum_{m=1}^M \alpha^m(\phi^m_T -\phi^m_0)$, where $\alpha^m \in \R $, $(\phi^m_T)_m$ are the terminal payoffs of some options traded in the market and $(\phi^m_0)_m \in \R$ their initial known costs. The approach taken in this paper is different.
We do not \emph{extend} the market, but rather we entangle several agents in the market by allowing them to exchange among themselves the amount represented by $Y \in \mcY$.  

 When $\mcY \subseteq \mcY_0$, it is only the total amount $\sum_{i=1}^N Y^i$ exchanged by the agents that is required to be equal to $0$.  The novel characteristic of this approach is the multi-dimensional notion of a Collective Arbitrage, a feature absent in the literature on Arbitrage Theory. We shall see, in Subsection \ref{secFairness}, that an endogenous vector of pricing measures $(\widehat Q^1, \dots, \widehat Q^N)$ will arise a posteriori from the pricing-hedging duality, which will assign zero price to each component of  the optimum exchanges $\widehat Y$, namely $E_{\widehat Q^i}[\widehat Y^i]=0$ for each $i$.   

\end{remark}

\subsubsection{On the Collective FTAP}
\label{oncollectiveFTAP}

We analyse the conditions under which a new type of Fundamental Theorem of Asset Pricing (FTAP) holds, that we label Collective FTAP (CFTAP). Differently from the classical version, the CFTAP depends of course on the properties of the set of exchanges $\mcY$ and so we provide, in Sections 3 and 6, several versions of such a theorem. On the technical side, in the classical case the $\mathbf{NA}$ condition implies that the set $( K- L^{0 }_+(\Omega, \mathcal{F},P) )$ is closed in probability, \cite{DS2006} Theorem 6.9.2.  This property is paramount to prove the FTAP and the dual representation of the super-replication price. Analogously, in our collective setting we need to show the closedness in probability of the set $K^\mathcal Y$ defined in \eqref{KyandCY}. We show such closure under  some specific assumptions on the set $\mathcal{Y}$  and under the assumption of $\mathbf{NCA(\mathcal{Y})}$. This technical part of the paper, in which some mathematical techniques are developed and others are adapted from the classical theory, is deferred to Section \ref{sec:Kyclosed}.\newline
The key novel feature in the CFTAP is that equivalent martingale measures have to be replaced by vectors $(Q^1, \dots,Q^N)$ of equivalent martingale measures, one for each agent and her corresponding market, fulfilling in addition the polarity property
\begin{equation} \label{polC}
 \sum_{i=1}^N  E_{Q^i}[Y^i ] \leq 0\,\, \text{ for all integrable } Y \in  \mathcal{Y},
\end{equation}
see Theorems \ref{FTAP3},
\ref{thm:EASYclosedness01} and \ref{616} for the precise formulations.
The findings of this paper take particularly tractable, yet informative and meaningful, forms 
 in a finite probability space setup (see Section \ref{section:finiteomega} and the examples in Section \ref{Sec:examples} in the very simple setting of a one and two periods models). The fact that the agents are allowed to cooperate and the assumption of $\mathbf{NCA(\mathcal{Y})}$ have several consequences also for the pricing of contingent claims. This is particularly  evident for the collective super-replication of $N$ contingent claims (see the example in Section \ref{73}. Item \ref{73C}), a notion which turns out to be well defined under $\mathbf{NCA(\mathcal{Y})}$.

\subsection{Collective Super-replication}\label{CS}
In Section \ref{secSuperNew} we consider the problem of $N$ agents each super-replicating a contingent claim $g^i$, $i=1,\dots,N$, which is a $\mathcal{F}$-measurable random variable.  We set $g=(g^1,\dots,g^N)$. 
As an immediate extension of the classical super-replication price, we first introduce the overall super-replication price 
\begin{equation*}
    \ro (g):=\inf \left \{ \sum_{i=1} ^N m^i \mid  \exists  k_i \in  K_i, m \in \mathbb{R}^N \text{ s.t. } m^i + k^i \geq g^i \text{  }  \forall i  \right \}.
    \end{equation*}
If we use the notation $\pii(g^i)$ for the classical super-replication of the single claim $g^i$, we may easily recognize that
\begin{equation}
    \ro (g)=\sum_{i=1}^N\pii(g^i). \label{RRXX}
    \end{equation} 
    In the spirit of Systemic Risk Measures with random allocations in \cite{BFFMB}, we  introduce the Collective super-replication of the $N$ claims  $g=(g^1,\dots,g^N )$  as
\begin{equation*}
    \mcR(g):=\inf \left \{ \sum_{i=1} ^N m^i \mid \exists  k_i \in  K_i, m \in \mathbb{R}^N,  Y \in \mathcal{Y} \text{ s.t. } m^i+ k^i + Y^i \geq g^i  \text{  }  \forall i  \right \}, 
    \end{equation*}
    and show that under $\mathbf{NCA(\mathcal{Y})}$ the definition is well posed (see Lemma \ref{PropertiesRho} ). The functional $\mcR(g) $ and $\ro (g)$ both represent the minimal total amount needed to super-replicate simultaneously all claims $(g^1,...,g^N)$. For the Collective super-replication price $\mcR(g)$ 
we allow an additional exchange among the agents, as described by $\mathcal{Y}$. As $0 \in \mcY$, we clearly have $\mcR \leq \ro$. Thus Collective super-replication is less expensive than classical super-replication: cooperation helps to reduce the cost of super-replication. The amount $(\ro(g)-{\mcR}(g)) \geq 0$ is the  \emph{value of cooperation} with respect to $g$.

To illustrate the advantage of adopting the collective super-replication price, consider the following situation.
Suppose that each desk $i$ (investing in a given market ($i$)) of a large bank is asked to super-replicate a claim $g^i$ and to price such operation. The bank will charge to the clients the amount $\rho_{i,+ }(g^i):=\inf \left \{ m \in \mathbb R \mid   \exists k^i \in K_{i}  \text{ s.t. } m + k^i \geq g^i \right \} $
for each claim $g^i$, resulting in charging a total price $\sum_{1=1}^N\rho_{i,+ }(g^i)=\rho_+^N(g)$ to super-replicate all the claims $(g^1,\dots,g^N)=g$.
However, the $N$ desks could cooperate, via the exchanges $Y \in \mcY$, and the \textit{effective} total cost for the bank to super-replicate all the claims is only
$\mcR(g) \leq  \rho_+^N(g). $
The positive difference $\rho_+^N(g)-\mcR(g)$ is the \textit{gain} obtained by the bank by selling $g$ at the price $\rho_+^N(g)$
and spending only $\mcR(g)$ to super-replicate it. 

See Section \ref{73} Item \ref{73C} for a simple example, where we have $\mcR (g)< \ro(g)$. The Collective super-replication can also be considered as a theoretical tool to show how much the system of the $N$ agents can loose without cooperation.
\\
Under the $\mathbf{NCA(\mathcal{Y})}$ assumption and using the closure of the set $K^{\mcY}$, defined in \eqref{KyandCY} below, we prove in Section \ref{secpriceduality} the following version of the pricing-hedging duality
\begin{equation} \label{dualIntro}
 \mcR(g)=\sup_{Q \in \mathcal{M}^{\mathcal{Y}} }  \sum_{i=1}^N  E_{Q^i } [g^i],
 \end{equation}
where $\mathcal{M}^{\mathcal{Y}}$ is the set of vectors of martingale measures satisfying the polarity condition \eqref{polC}.

In the dual formulation of $\mcR(g)$ in \eqref{dualIntro} we are thus considering the specific subset $\mathcal M^{\mcY}$ of the set 
of (vectors of) martingale measures. 
Hence our approach fits in the general trend in the literature: adding constraints implies using a smaller set of martingale measures, which in turns implies a reduction of the no arbitrage price interval of the contingent claims. 
\\
When problem \eqref{dualIntro} admits an optimum $\widehat Q=(\widehat {Q}^1, \dots, \widehat {Q}^N) \in \mathcal{M}^{\mathcal{Y}}$, which clearly will depend on $\mcY$, 
we derive in Section \ref{secFairness} the following formula
\begin{equation} \label{Introrhoi}
\mcR(g)
= \sum_{i=1} ^N  \inf \left \{ m \in \mathbb R \mid   \exists k^i \in K_{i}, Y^i  \text { with } E_{\widehat{Q}^i } [Y^i]=0  \text{ s.t. } m + k^i + Y^i \geq g^i \right \}. 
\end{equation}
Note that in \eqref{Introrhoi}, $(Y^1,\dots,Y^N)$ is not required to belong to $\mcY$, but every $Y^i$ must have zero cost under each component of the endogenously determined pricing vector  $\widehat Q$.

Thus $\mcR(g)$ is the sum of the \emph{individual super-replication price} of each claim $g^i$ under the assumption that each agent is using the pricing functional assigned by $\widehat {Q}^i$, so that both $k^i$ and $Y^i$ have zero value under $\widehat {Q}^i$.

In the spirit of \cite{bffm} and \cite{BDFFM}, this fairness aspect, as well as the connection with B\"uhlmann's notion of a risk exchange equilibrium in \cite{Buhlmann}, is discussed in Section \ref{secFairness}.

\subsection{Relation with Classical Arbitrage}\label{LR}

The concept of arbitrage in financial markets finds its roots in the seminal work of de Finetti \cite{deFinetti31}, see Maggis \cite{Maggis23} for a historical excursion on this topic. However, it wasn't until the mid-1970s, with the groundbreaking works of Black and Scholes \cite{B&S73}, Merton \cite{Merton73}, and Ross \cite{Ross76}, that Arbitrage Pricing Theory became a cornerstone of financial mathematics. The late 1970s witnessed the establishment of a mathematical foundation of this theory through the contributions of Kreps \cite{Kreps81}, Harrison and Kreps \cite{HarrisonKreps79} and Harrison and Pliska \cite{HarrisonPliska81}.

A pivotal moment arrived with the highly influential paper by Dalang, Morton, and Willinger \cite{DMW90}. This work presented the Fundamental Theorem of Asset Pricing in discrete finite time in its current, widely recognized form. Subsequently, Schachermayer \cite{Schachermayer92}, Rogers \cite{Rogers94}, Kabanov and Kramkov \cite{KaKr95}, Jacod and Shiryaev \cite{JaShi98} each developed independent proofs of this theorem. A comprehensive overview of this topic in discrete time can be found in the book by F\"{o}llmer and Schied \cite{FollmerSchied2}, 

Over the past few decades, a vast literature has emerged surrounding variations of the one-agent no-arbitrage (or no-free-lunch) concept in more general settings, and in particular in continuous time. A detailed survey of this subject can be found in the book by Delbaen and Schachermayer \cite{DS2006} and references therein.  It is important to mention also the significant research conducted in the past decade on robust versions of these concepts.

Our approach offers novel insights that complement and extend the existing literature. We now present a comparison with the very general notion of arbitrage developed by Kreps \cite{Kreps81}. An arbitrage for \cite{Kreps81} is a {\it{single}} non zero element in the vector space of marketed bundles $M$, contained in a locally convex topological vector lattice $\mathcal X$, having non positive price and belonging to the positive cone $\mathcal X^+$, associated to the lattice structure of $\mathcal X$.
Due to its generality, this definition can encompass a wide range of possibilities. However, to the best of our knowledge, both this definition and the one of a free lunch have traditionally been interpreted as a single element (within the abstract space M) constituting an arbitrage.
The novelty of our approach lies in the concept of a Collective Arbitrage, which possesses the intrinsic multidimensional feature of being a collection of "arbitrage opportunities" in the precise sense already defined. 
The fact that $M$ has a product structure ($M= K^1 \times \dots \times K^N$) and that the positive cone $\mathcal X^+$  has the particular structure $(L^0_+-\mathcal{Y})$, $\mathcal{Y}$ being the particular set of exchanges, extends the ansatz of Kreps to a new setting. 

We emphasize that our results are not simply a consequence of existing methods but necessitate a dedicated technical effort, as elaborated upon throughout the paper, particularly in Section \ref{sec:Kyclosed}.
\\

\medskip
\textbf{Organization of the paper:} In Section \ref{setting} we introduce the setting for the financial market model. The definition and the properties of the Collective Arbitrage are collected in Section \ref{secCA}, while various versions of the Collective FTAP are proven in Section \ref{secFTAP}, under the assumption that the set $K^{\mcY}$ is closed in probability. In Section \ref{secSuperNew} we develop the theory of the Collective Super-replication, we prove the pricing-hedging duality (Section \ref{secpriceduality}) and show the fairness property (Section \ref{secFairness}) embedded in the collective super-replication scheme. 
The case of a finite probability space is resumed in Section \ref{section:finiteomega}.
In Section \ref{sec:Kyclosed} we prove, under different assumptions, that the set  $K^{\mcY}$ is closed in probability, and thus we prove several versions of the CFTAP. 
Most of the examples are collected in the final Section \ref{Sec:examples}.\\

To conclude this introduction, we mention that the following extensions are work in progress: Collective Arbitrage and Super-replication in continuous time,  Collective Free Lunch, Robust Collective Arbitrage, Collective Optimal Transport, Collective Risk Measures.

\section{The Setting}\label{setting}
Let $\mathcal T=\{0, \dots, T\}$ be the finite set of discrete times and consider a given filtered probability space $(\Omega, \mathcal{F}, \mathbb F ,P)$, with  $\mathbb F =\{\mathcal{F}_t\}_{t \in  \mathcal T}$,  
and $\mathcal{F}=\mathcal{F}_T$. Except when explicitly stated otherwise, we also assume that $\mcF_0$ is trivial, namely that $P(A)=0$ or $P(A)=1$ for all $A \in \mcF_0$. 
 We say that a probability measure $Q$ on $(\Omega, \mathcal{F})$ belongs to $\mathcal{P}_e$ if $Q \sim P$, or to $\mathcal{P}_{ac}$ if $Q\ll P$, respectively.
 Unless differently stated, all inequalities between random variables are meant to hold $P$-a.s..\\

\begin{remark}
\label{remextension}
In this paper we chose to work in a discrete time setting in order to convey our main messages about collective arbitrage, collective super-replication and the role of cooperation, focusing more on the conceptual aspects of the theory. However, most of the concepts we are introducing  
also fit well in a more general framework (as continuous time or robust finance). 
However, the techniques in this paper will need to be carefully adapted.
This is the subject of a paper in preparation.
\end{remark}

\noindent Let us consider $N$ agents and denote their set by $I = \{ 1, \dots , N \} $. We assume the existence of a zero interest rate riskless asset in which all agents can invest.
The (global) securities market consists in $J$ assets, for a given integer $J\geq 1$, with discounted price  processes $X^j=(X^j_t)_{t\in \mathcal T}$,
  $j=1, \dots , J$. 
We assume that  
agent $i$ can invest in the assets $X^j$, $j\in (i)$, where $(i)$ stands for a given subset of $\{1,\dots,J\}$. We set $d_i=\#(i)$. 
We assume (without loss of generality) that $\cup_i(i)=\{1,\dots,J\}$, as
we may ignore the assets that can not be used by any agent.

\begin{example}
Let  $N=3$, and $J = 4 $,  $(1)=\{1,2\}$, $(2)=\{2\}$, $(3)=\{1,3,4\}$. Thus, in this example, agent \# 1 may invest in assets $(X^1, X^2)$, agent \# 2 in assets $X^2$ and agent \# 3 in assets $(X^1,X^3,X^4)$.
\end{example}

\noindent
We denote by $$\fib=(\fit)_{t\in\mathcal{T}}\subseteq \mathbb F$$ the filtration representing the information available to agent $i$.
We assume that all processes $X^j$, with $ j \in (i)$, are adapted with respect to the  filtration $\fib$. 

A stochastic process $H=(H_t)_{t\in \mathcal T}$ is called an \label{admisstrat}\emph{admissible trading strategy for the agent} $i$ if it is $d_i$ - dimensional and  predictable with respect to $\fib$.
The space of admissible trading strategies for the agent $i$ is denoted by $\hi$. 
If $H \in \hi$, we set 
\begin{equation*}
(H\cdot X)_T:=\sum_{h \in (i) } (H^h\cdot X^h)_T,    
\end{equation*} 
where $(H^h\cdot X^h)_T:=\sum_{t=1}^T H^h_t(X^h_t-X^h_{t-1})$ denotes the stochastic integral of $H^h$ with respect to  the asset $X^h$, $h\in (i)$, and we write

\begin{equation}
\label{def:Ki}
\ki:=\{(H\cdot X)_T \mid H \in {\hi} \} \subseteq L^{0 }(\Omega, \mathcal{F}_T^{i},P).
\end{equation}

\noindent The set of martingale measures with bounded densities for  the assets in $(i)$ is defined by
\begin{equation}\label{MMM}
\mi:=\left\{ Q \in \mathcal{P}_{ac} \mid \frac {dQ} {dP} \in L^{\infty }(\Omega, \mathcal{F}_T^{i},P)  \text { and } X^j  \text { is a } (Q, \fib) \text {-martingale for all } j \in (i) \right\}.
\end{equation}
The existence of an element in $\mi$ which is also equivalent to $P$, characterizes the No-Arbitrage condition for agent $i$ (defined below), as recalled in Theorem \ref{DMW}.
\begin{remark}\label{remIntegrability}By setting $\frac{d \widehat P}{dP}:= \frac{c}{1+\sum_{j,t}| X^j_t |}$,  for some positive normalizing constant $c$, we see that $\widehat P$ is a probability equivalent to $P$ such that the price process $X$ is integrable under $\widehat P$.
The notions of No Arbitrage in Definition \ref{NAclassic}, of No Collective Arbitrage in Definition \ref{NCA}, of Super-replication Price  in Definitions \ref{defSup} and \ref{defpi} depend on the underlying probability only through its null set, i.e., they are invariant with respect to a change of equivalent measure. 
Thus we will directly assume that the probability $P$ satisfies
$$X^j_t \in L^1(\Omega,\fit, P) \quad \text{ for each } t\in \mathcal T \text{ and }  j \in (i) , \quad  i=1, \dots , N.$$
\end{remark}
We recall from  \cite{DS2006}, Section 6.11, or \cite{FollmerSchied2} Theorem 5.14 that $\mi$ can also be written as  
\begin{equation}\label{MartingaleMeasures}
\mi=\left\{ Q \in \mathcal{P}_{ac} \mid \frac {dQ} {dP} \in L^{\infty }(\Omega, \mathcal{F}_T^{i},P)  \text { and } E_Q[k]= 0  \text {  } \forall k \in \ki \cap L^{1 }(\Omega, \mathcal{F}_T^{i},P) \right\}.
\end{equation}

\begin{notation}
\label{productnotation}
We shorten $L^{0 }(\Omega, \mathcal{F}_T^{1},P) \times \dots \times  L^{0 }(\Omega, \mathcal{F}_T^{N},P)$ by the notation
$ L^{0 }(\Omega, \mathbf{F}_T,P)$, and similarly for $L^{1 }(\Omega, \mathbf{F}_T,P)$ and $L^{\infty}(\Omega, \mathbf{F}_T,P)$. 
For a vector of probability measures $\mathbf{P}:=(P^1, \dots , P^N)$, we also set $$ L^{0 }(\Omega, \mathbf{F}_T,\mathbf{P}):=L^{0 }(\Omega, \mathcal{F}_T^{1},P^1) \times \dots \times  L^{0 }(\Omega, \mathcal{F}_T^{N},P^N), $$ and similarly for $L^{1 }(\Omega, \mathbf{F}_T,\mathbf{P})$ and $L^{\infty }(\Omega, \mathbf{F}_T,\mathbf{P})$.
\end{notation}

\noindent Note that from now on, we denote the norm of a vector $v\in \mathbb{R}^N$ by $|v|$.

\subsection{Classical No Arbitrage and Super-hedging Duality}

In the classical setting  we have the following definition and  characterization of absence of arbitrage for agent $i$ and for the global market.
\begin{definition}
\label{NAclassic} 
$\,$

\begin{itemize}
\item Classical No Arbitrage Condition for agent $i$:

\begin{equation}
\label{eq:NAi}
 \mathbf{NA}_{i}: \text{  } \ki \cap L_{+}^{0}(\Omega, \mathcal{F}^i_T , P)=\{0\}.
 \end{equation}

\item Classical No Global  Arbitrage Condition.
Let $K$ represent the set of time $T$ stochastic integrals, with $\mathbb F$-predictable integrands, with  respect to all the assets  $\{X^1, \dots , X^J \}$ (the global filtration $\mathbb F$ was introduced at the beginning of Section \ref{setting}). Let $M $ be the set of martingale measures, with bounded densities, for the whole market $\{X^1, \dots , X^J \}$ with respect to the  filtration $\mathbb F$.
We say that there is No Global Arbitrage, denoted with $\mathbf{NA}$, if there is no classical arbitrage in the whole market $\{X^1, \dots , X^J \}$:
\begin{equation*}
 \mathbf{NA}: \text{  } K \cap L_{+}^{0}(\Omega, \mathcal{F}_T , P)=\{0\}.
 \end{equation*}
\end{itemize}
\end{definition}

\begin{theorem}[Fundamental Theorem of Asset Pricing (FTAP) - Dalang-Morton-Willinger]\label{DMW}
In our discrete time setting
\begin{equation*}
 \mathbf{NA}_{i} \Leftrightarrow  \mi\cap \mathcal{P}_e \not = \emptyset; \quad \mathbf{NA} \Leftrightarrow  M \cap \mathcal{P}_e \not = \emptyset.
\end{equation*}
\end{theorem}

\begin{definition}[Classical super-replication  for agent $i$]
For $f \in L^{0}(\Omega, \mathcal{F}^i_T , P)$, define
\begin{equation*}
\pii(f):=\inf\{m \in \mathbb{R} \mid \exists k \in \ki \text{ s.t. } m+k \geq f \}.
\end{equation*}
\end{definition}
\noindent  The following super hedging duality holds as in the classical case, so that we omit its proof. Under $\mathbf{NA}_i$ 
\begin{equation*}
\pii(f)=\sup_{Q \in\mi} E_{Q}[f] \ \text { for all } f \in L^1(\Omega, \mathcal{F}^i_T , P).
\end{equation*}

\section{Collective Arbitrage}

\subsection{Definitions and Properties}\label{secCA}
After recalling the classical setup and definitions, we now introduce some novel concepts in the multi-agent setting of Section \ref{setting}.

Our notions of Collective Arbitrage and Collective Super-replication are based on  the possibility that the $N$ agents may enter in  risk exchange mechanisms, in a possibly scenario-dependent way.

We model such exchanges as (possibly random) $N$-dimensional vectors $Y$ belonging to a set 
 $$\mathcal Y \subseteq L^{0 }(\Omega, \mathbf{F}_T,P), \, \text{ with} \, 0 \in \mcY,$$ using notation \ref{productnotation} and assuming in particular that $Y^i$ is $\fiT$-measurable for every component $i$ of such vectors.  These risk exchanges   are accounted at terminal time $T$ and are measured in cash unit.  We also assume $0 \in \mathcal Y$, allowing for the agents not to enter in any exchange at all. In general, to talk about risk exchange, we consider 
 zero-sum exchanges
$\sum_{i=1}^N Y^i = 0\quad P-\text{a.s.}$. However, this requirement is not needed for some technical results, so we will explicitly make this assumption later.
 
Each agent $i$ follows an investment strategy  $H_i\in\mathcal{H}^i$ in its own market $(i)$, with terminal value $k^i=(H_i\cdot X)_T=\sum_{h \in (i) } (H_i^h\cdot X^h)_T\in \ki$. The agents also enter in the risk exchange corresponding to a vector $Y \in \mathcal Y $. This procedure leads to the terminal time value $(H_i\cdot X)_T + Y^i$ for agent $i$. A \textbf{collective arbitrage} is given by a collection of strategies $(H_i)_{i\in I}$ and a risk exchange $Y \in \mcY$  such that each component $(H_i\cdot X)_T + Y^i$ is nonnegative and there exists at least one $j \in \{1, \dots, N \}$ for which $P((H_j\cdot X)_T + Y^j>0)>0$.
Now we can introduce the formal definition.

 \begin{definition}[No Collective Arbitrage]
\label{NCA} 
We say that No Collective Arbitrage for $\mathcal{Y}$ ($\mathbf{NCA}(\mathcal{Y})$) holds if

\begin{equation}
({\sf X}_{i=1} ^{N} K_i+\mathcal Y )\cap  L^{0 }_+(\Omega, \mathbf{F}_T,P)=\{0\}, \label{NCAY} 
 \end{equation}
 where ${\sf X}_{i=1} ^{N} \ki$ denotes the Cartesian product of the sets $\ki$ defined in \eqref{def:Ki}.
 \end{definition}
 By selecting the zero vector in ${\sf X}_{i=1} ^{N} \ki$ we obtain from \eqref{NCAY}  the condition $\mathcal Y \cap  L^{0 }_+(\Omega, \mathbf{F}_T,P)=\{0\}$, which rules out the possibility of collective arbitrage only through elements in $\mathcal{Y}$. 
 
 \begin{proposition}
The following conditions are equivalent: 
\begin{align}
 \mathbf{NCA}(\mathcal{Y}) \label{a}\\
 K^\mathcal Y \cap  L^{0 }_+(\Omega, \mathbf{F}_T , P)&=\{0\}, \label{b} \\
 C^\mathcal Y \cap  L^{1 }_+(\Omega, \mathbf{F}_T,P)&=\{0\}, \label{c}
 \end{align}
 where 
 \begin{equation}
 \label{KyandCY}
 K^\mathcal Y:=  {\sf X}_{i=1} ^{N} ( K_i - L^{0 }_+(\Omega, \mathcal{F}^i_T,P) ) + \mathcal Y \quad \text{ and } \quad C^\mathcal Y :=K^\mathcal Y \cap L^{1 }(\Omega, \mathbf{F}_T,P). 
 \end{equation}

 \end{proposition}
 \begin{proof}
$ \eqref{a} \Rightarrow  \eqref{b}$: take $(k^i-l^i+Y^i)_i \in K^\mathcal Y$ with $k^i-l^i+Y^i  \geq 0$ for all $i$. Then $k^i+Y^i \geq l^i \geq 0$ for all $i$, and $\eqref{a} \Rightarrow k^i+Y^i=0$ for all $i$, thus $l^i=0$ and so $k^i-l^i+Y^i =0$ for all $i$.\\ 
$  \eqref{b} \Rightarrow  \eqref{c}$: follows from   $C^\mathcal Y \subseteq K^\mathcal Y$ and $L^{1 }_+(\Omega, \mathbf{F}_T,P) \subseteq L^{0 }_+(\Omega, \mathbf{F}_T,P)$.\\ 
 $\eqref{c} \Rightarrow  \eqref{a}$: Take $k \in {\sf X}_{i=1} ^{N} K_i$ and $ Y \in \mathcal Y$ and suppose $f^i:=k^i+Y^i \geq 0 $ for all $i$. Set $f^i_m:=\min (f^i, m)$, for some $m \in \mathbb N$. Then $0 \leq f^i_m \leq m$, $f^i_m \in L^{1 }_+(\Omega, \mathcal{F}^i_T,P)$ and $f^i_m=f^i -(f^i - f^i_m)=k^i+Y^i-l^i$, where $l^i:=f^i-f^i_m \in L^{0 }_+(\Omega, \mathcal{F}^i_T,P)$. Thus $(f^1_m, \dots , f^N_m) \in ( {\sf X}_{i=1} ^{N} ( K_i - L^{0 }_+(\Omega, \mathcal{F}^i_T,P) ) + \mathcal Y ) \cap L^{1 }(\Omega, \mathbf{F}_T,P) = C^\mathcal Y$. Condition $\eqref{c}$ then implies $f^i_m=0$ for all $i$. Thus $k^i+Y^i=0$ for all $i$.
 \end{proof}

Define 
\begin{equation}\label{Y00}
\mathcal{Y}_0:=\left\{Y \in L^{0 }(\Omega, \mathbf{F}_T,P) \mid \sum_{i=1}^N Y^i = 0\right\}.
\end{equation}
When $\mathcal F_T^i=\mathcal F_T$, for all $i$, then $L^{0 }(\Omega, \mathbf{F}_T,P)=(L^{0}(\Omega, \mcF_T,P))^N$, so that \eqref{2345} and \eqref{Y00} coincide. In the case $\mathcal{Y}\subseteq \mathcal{Y}_0$, the exchange is a zero-sum game, in other words the exchange is self-financed in that no money is injected in or taken out of the system.
Hereafter we collect some simple considerations regarding  $\mathbf{NCA(\mathcal{Y})}$ and No Arbitrage.
\begin{enumerate}
\item When  $\mathcal{Y}=\mathcal{Y}_0 $ and $N=1$, then $\mathbf{NCA(\mathcal{Y})}$ coincides with the classical notion of $\mathbf{NA}_1 $, namely No classical Arbitrage for market $(1)$ of the only one agent, in Definition \ref{NAclassic}. 
\item When $\mathcal{Y}=\{0\}$, $N=1$, $\mathbb{F}^1=\mathbb{F}$ and we consider one single market consisting of all assets $\{X^1, \dots , X^J \}$,
then the notion of $\mathbf{NCA(\mathcal{Y})}$ reads as $K \cap  L_{+}^{0}(\Omega, \mathcal{F}_T,P)=\{0\}$, namely it coincides with the classical notion of $\mathbf{NA}$ for the whole market $\{X^1, \dots , X^J \}$.

\item From \eqref{NCAY}, we see that in general, for any set $\mathcal Y \subseteq L^{0 }(\Omega, \mathbf{F}_T,P)$ with $0 \in \mathcal Y$,

\begin{equation}\label{NCANA}
\mathbf{NCA(\mathcal{Y})} \Rightarrow \mathbf{NA}_{i} \text{ } \forall i.
\end{equation}
More precisely taking $Y=0$ in \eqref{NCAY} we get
\begin{equation*}
\mathbf{NCA(\mathcal{Y})} \Rightarrow \text{  }  K_i \cap  L_{+}^{0}(\Omega, \mathcal{F}^i_T,P)=\{0\} \text { for all } i \Leftrightarrow \mathbf{NA}_{i} \text{ } \forall i \Leftrightarrow \mi \cap \mathcal{P}_e \not = \emptyset \text{ } \forall i.
\end{equation*}

\item When $\mathcal{Y} \subseteq \{Y \in L^{0 }(\Omega, \mathbf{F}_T,P) \mid \sum_{i=1}^N Y^i \leq 0\}$, then 
\begin{equation}
    \label{NANCAY}
    \mathbf{NA} \Rightarrow \mathbf{NCA(\mathcal{Y})},
\end{equation}

where $\mathbf{NA}$ denotes No classical Arbitrage in the global market $K$. 
 
Indeed, we prove the implication $\mathbf{CA} \Rightarrow \mathbf{A}$. Suppose that there exists a vector  $k \in {\sf X}_{i=1} ^{N}  K_i$ and a vector $Y \in \mathcal Y $ such that each component $k^i + Y^i$ is non negative and there exists at least one $j \in \{1, \dots , N \}$ for which $P(k^j + Y^j>0)>0$.  Then by summing up all the components we obtain that ($\sum_{i=1}^N k^i $) is $\mathcal F_T-$measurable, $\sum_{i=1}^N k^i \geq -\sum_{i=1}^N Y^i   \geq 0$ and $P(\sum_{i=1}^N k^i > 0)>0$. Then $\sum_{i=1}^N k^i \in \sum_{i=1}^N K_i $ is an arbitrage in the whole market.

\item The converse implication in \eqref{NCANA} is false in general (as shown in Section \ref{toyex} Item \ref{72B}, Section \ref{73} Item \ref{73B}), 
as well as the converse implication in  \eqref{NANCAY} does not hold in general (as shown in Section \ref{toyex} Items \ref{72A},  \ref{72C} and  \ref{72D}; Section \ref{73} Items \ref{73A} and \ref{73C}). The features of these examples are presented at the end of Section \ref{Sec:examples} in Table \ref{tablerecap}.
In Sections \ref{61} and \ref{62} we provide sufficient conditions ensuring that these converse implications hold true. 
However, the interesting cases are: (i) when there exists a Classical Arbitrage in the global market  but $\mathbf{NCA(\mathcal{Y})}$ holds true (that is the converse of \eqref{NANCAY} does \textbf{not} hold); (ii)
when an arbitrage cannot be realised by each single agent, say  due to constraints in investment of the agents, but a Collective Arbitrage can be realised through the exchanges $\mathcal{Y}$ (that is the converse of \eqref{NCANA} does \textbf{not} hold).

\item Suppose that $0 \in \mathcal{Y} \subseteq \{Y \in L^{0 }(\Omega, \mathbf{F}_T,P) \mid \sum_{i=1}^N Y^i \leq 0\}$ and that the market of some agent $j$ coincides with the global market $K$. Then $\mathbf{NA} \iff \mathbf{NA}_j$ so that, from \eqref{NANCAY} and \eqref{NCANA},  we deduce $\mathbf{NA} \iff \mathbf{NCA}(\mcY)$.

\end{enumerate}

\subsection{Collective Fundamental Theorem of Asset Pricing (FTAP)}
\label{secFTAP}
In the classical case, the $\mathbf{NA}_i$ condition implies that the set $( K_i - L^{0 }_+(\Omega, \mathcal{F}^i_T,P) )$ is closed in $ L^{0 }(\Omega, \mathcal{F}^i_T,P)$, see \cite{DS2006}, Theorem 6.9.2. This property is paramount to prove the FTAP and the dual representation of the super-replication price. Analogously, in our collective setting we would like to show the closedness of the set $K^\mathcal Y$,  appearing in the definition of $\mathbf{NCA(\mathcal{Y})}$ in \eqref{b}. 
Some closure properties of the set of semistatic trading strategies have been recently investigated in \cite{NutzWieselZhao22}. However, an example of non-closedness is given in \cite{AcciaioSchachermayerLarssen17}, so
 clearly we expect this to be a delicate point. But in the present theory, developed in Section \ref{sec:Kyclosed}, we have two additional facts that will lead to such closure: some specific assumptions on the set $\mathcal{Y}$  and the assumption of $\mathbf{NCA(\mathcal{Y})}$.

 We now consider the polarity with respect to the dual system $(L^{1 }(\Omega, \mathbf{F}_T,P),L^{\infty }(\Omega, \mathbf{F}_T,P))$. Define for $C^\mathcal{Y}$ introduced in \eqref{KyandCY} 
\begin{align}
 (C^\mathcal{Y})^{0}&:=\left\{z \in L^{\infty }(\Omega, \mathbf{F}_T , P) \mid \sum_{i=1}^N E[z^if^i] \leq 0 \text{  }  \forall f \in C^\mathcal{Y} \right\}\label{defofpolar} ,\\
(C^\mathcal{Y})^{00}&:=\left\{f \in L^{1 }(\Omega, \mathbf{F}_T , P) \mid \sum_{i=1}^N E[z^if^i] \leq 0 \text{  }  \forall z \in (C^\mathcal{Y})^0 \right\},\notag\\
(C^{\mathcal{Y}})^{0}_1&:=\left \{(Q^1,\dots,Q^N)\in (\mathcal{P}_{ac})^N \mid \left (\frac {dQ^1} {dP}, \dots, \frac {dQ^N} {dP} \right ) \in (C^\mathcal{Y})^{0}  \right \}. \notag  
\end{align}

\begin{remark}
Note that $(C^\mathcal{Y})^{0}\subseteq  L^{\infty }_+(\Omega, \mathbf{F}_T , P)$ since by taking  $Y=0\in \mathcal{Y}$, $k=0 \in {\sf{X}_{i=1}^N} K_i$ in the definition of $K^{\mathcal
Y}$, we have $-L^{1 }_+(\Omega, \mathbf{F}_T , P)\subseteq C^\mathcal{Y}$.
\end{remark}

We now define the vector space
\begin{equation}\label{R0}
\mathbb R_0 ^N:=\left \{   x \in \mathbb R^N \mid \sum_{i=1} ^N x^i =0  \right \}
\end{equation}
and observe that it can be equivalently expressed as:
\begin{equation}\label{R00} 
\mathbb R_0 ^N= \mathrm{span}  \left \{  (e^i-e^j) \mid  i,j \in \{1, \dots , N \}  \right \},
\end{equation}
where $\{e^i\}_i$ is the canonical basis of $\mathbb R ^N$.
We will often require that 

\begin{equation}\label{RY} 
\mcY \subseteq  L^{0 }(\Omega, \mathbf{F}_T , P) \text { is a convex cone which contains  } \mathbb R_0 ^N,
\end{equation}
which also implies $\mcY + \mathbb R_0 ^N=\mcY$.
Under this condition every agents is allowed to exchange with any other agent any constant (i.e. deterministic) amount, as long as the total amount exchanged by all agents is equal to $0$. Equivalently, each couple of agents is allowed to exchange between them any given deterministic amount. Even if condition \eqref{RY} is not necessary for the development of our theory, this assumption simplifies the presentation of our results.
\begin{remark}
    \label{nicepolar}
If \eqref{RY} is fulfilled, then for all  $z \in (C^\mathcal{Y})^{0}$ we have $E[z^i]=E[z^j]$ for all $i$, $j$.  
Indeed, by taking the zero element in ${\sf X}_{i=1} ^{N} ( K_i - L^{0 }_+(\Omega, \mathcal{F}^i_T,P) )$, we get from the definition \eqref{defofpolar} of the polar set $(C^\mathcal{Y})^0$  that  $\sum_{i=1}^N E[z^i Y^i] \leq 0 \text{  }  \forall  Y \in \mathcal Y \cap L^{1}(\Omega, \mathbf{F}_T , P)$. By \eqref{R00}, we can take  $Y=\lambda (e^i-e^j)\in\R^N_0$ for an arbitrary $\lambda\in\R$, which immediately yields the desired equalities.

Thus, when $z \in (C^\mathcal{Y})^{0}$, $z \not = 0$,  and we consider the $\emph{polarity}$ condition  $\sum_{i=1}^N E[z^iY^i] \leq 0 \, \forall  Y \in \mathcal Y \cap L^{1}(\Omega, \mathbf{F}_T , P) $, it is possible to normalize it by dividing by the same amount $E[z^1]=\dots=E[z^N]$ to obtain
that $\sum_{i=1}^N  E_{Q^i}[Y^i ] \leq 0 \text { } \forall  Y \in \mathcal Y \cap L^{1}(\Omega, \mathbf{F}_T , P)$, with  $Q \in  (C^\mathcal{Y})_1^0 $.
\end{remark}

\begin{definition}
\label{defMY}
Recall the notation \eqref{MMM} for martingale measures and introduce
\begin{equation*}
\mathcal{M}^{\mathcal{Y}}= \left \{ (Q^1, \dots , Q^ N) \in M_{1} \times \dots  \times M_{N} \mid  \sum_{i=1}^N  E_{Q^i}[Y^i ] \leq 0 \text { } \forall Y \in  \mathcal{Y} \cap L^{1 }(\Omega, \mathbf{F}_T , P)  \right \}.
\end{equation*}
\end{definition}
 In the definition of $\mathcal{M}^{\mathcal{Y}}$, the term $\sum_{i=1}^N  E_{Q^i}[Y^i ]$ is well posed and finite, since each component $Y^i$ is required to be integrable under $P$ and each $Q^i$ has bounded density with respect to $P$.
\label{commentsepmeasures} In the classical No-Abitrage theory in discrete-time, the No-Arbitrage condition is characterized by equivalent martingale measures  (see e.g. \cite{DS2006} Chapter 6 Th. 6.1.1 or \cite{FollmerSchied2} Chapter 5). In continuous time, these have to be replaced by separating measures, namely local\slash sigma equivalent martingale measures. As the framework in this paper is in a discrete-time, we directly work with martingale measures.  

\begin{example}[Grouping Example]\label{EXgroup}
    Consider a partition $\alpha=(\alpha_1, \dots , \alpha _m)$ of $\{1, \dots , N \}$ into $m$ groups $\alpha_h$, $h=1, \dots, m$, for $1\leq m \leq  N$ and define 
\begin{equation}\label{ExG}
\mcY_{\alpha}:=\left \{   Y \in L^{0 }(\Omega, \mathbf{F}_T , P)   \mid \sum_{i \in \alpha_h}  Y^i =0, \text { for each }  h=\{1,\dots,m \} \right \}.
\end{equation}
Observe that for the partition made of only one element, namely the element $\{1,\dots,N\}$, the set in \eqref{ExG} coincides with the set in \eqref{Y00}.
We can attribute to the set $\mcY_{\alpha}$ an interpretation similar to the one given for the set in \eqref{Y00}, with the only difference that now the exchanges \emph{in each single group $\alpha_h$} must sum up to 0.
Observe that, for some fixed partition $\alpha $, these sets $\mcY_{\alpha}$  may not satisfy $\eqref{RY}$. 
In order to deal with these types of sets  $\mcY_{\alpha}$ and also with the larger class of subsets $\mcY \subseteq \mcY_{\alpha}$ with the property
\begin{equation}\label{YGropup}
Y \in \mcY \text { implies } (1_{\alpha_h}Y )\in \mcY \text{ for all } h=1,\dots,m,
\end{equation}
where $ (1_{\alpha_h})^i:=1$ if  $i \in \alpha_h$, and $ (1_{\alpha_h})^i:=0$ otherwise, we consider the vector space $\mathbb R_0 ^{\mcY_{\alpha}}$ defined by
\begin{equation}\label{R0Y}
\mathbb R_0 ^{\mcY_{\alpha}}:=\left \{   x \in \mathbb R^N \mid \sum_{i \in \alpha_h}  x^i =0 \text { for each }  h=\{ 1,\dots,m \} \right \}
\end{equation}
and suppose that 
\begin{equation}\label{mod22}
\mcY \subseteq \mcY_{\alpha} \text{  is a convex cone satisfying } \eqref{YGropup} \text { and containing }  \mathbb R_0^{\mcY_{\alpha}}.    
\end{equation}
With the modification \eqref{mod22} of \eqref{RY} we may appropriately normalize the elements in the polar $(C^\mathcal{Y})^{0}$ and deduce the polarity condition $\sum_{i=1}^N  E_{Q^i}[Y^i ] \leq 0 \text { } \forall  Y \in \mathcal Y \cap L^{1}(\Omega, \mathbf{F}_T , P)$, with  $Q \in  (C^\mathcal{Y})_1^0 $, also for such sets. To avoid non essential and more complex formalization, in the rest of the paper we do not focus any more on this issue, as the reader can follow the argument sketched here to address such cases.
\end{example}

\begin{remark}
In some of the following results we use the assumption $\mathcal{Y} \subseteq L^{1 }(\Omega, \mathbf{F}_T , P).$ This is the case, for example, when $\mathcal{Y} \subseteq (L^{0 }(\Omega, \mathcal{G}_T , P))^N$, for a sigma-algebra $\mathcal {G}_T$ that is $\emph{ finitely generated}$ and contained in  $\mathcal{F}^i_T, $ for all $i$. This is also the case developed in Section \ref{section:finiteomega} where the $\sigma$-algebras in $\mathbb F$ are generated by a finite number of atoms of $\Omega$.
\end{remark}
The following Lemma  gives conditions so that $\mathcal{M}^{\mathcal{Y}}$ is the normalized polar cone of $C^\mathcal{Y}$.
\begin{lemma}
\label{lemmaYL1}

We have
$(C^\mathcal{Y})^{0}_1 \subseteq \mathcal{M}^{\mathcal{Y}} $. If 
$\mathcal{Y} \subseteq L^{1 }(\Omega, \mathbf{F}_T , P)$  
then
$$  (C^\mathcal{Y})^{0}_1 = \mathcal{M}^{\mathcal{Y}}.$$
\end{lemma}

\begin{proof}
For the first inclusion take $(Q^1,\dots,Q^N)\in (C^{\mathcal{Y}})^{0}_1$ and set $z:=\left (\frac {dQ^1} {dP}, \dots, \frac {dQ^N} {dP} \right ) \in (C^\mathcal{Y})^{0}$. Then

\begin{equation}\label{eq444}
 \sum_{i=1}^N E[z^if^i] \leq 0 \text{ for all  }   f \in C^\mathcal{Y}.
\end{equation}
From \eqref{eq444}, by taking the zero element in ${\sf X}_{i=1} ^{N} ( K_i - L^{0 }_+(\Omega, \mathcal{F}^i_T,P) )$, we get   $\sum_{i=1}^N E[z^i Y^i] \leq 0 \text{  }  \forall  Y \in \mathcal Y \cap L^{1}(\Omega, \mathbf{F}_T , P)$. 

Furthermore,  as $0\in \mathcal{Y}$ and $0\in K^i$ for every $i=1,\dots,N$, we also get 
\begin{equation}\label{eq4}
E[z^i k^i ] \leq 0 \text { for all } k^i \in \ki \cap L^1(\Omega, \mathcal{F}^i_T,P) ), \text { for all } i.
\end{equation}

Since $\ki$ is a vector space, we deduce from \eqref{eq4} that $E[z^i k^i ] = 0$ for all $k^i \in \ki \cap L^1(\Omega, \mathcal{F}^i_T,P) $ and all $i$. Thus by \eqref{MartingaleMeasures}, $Q^i \in \mi$ for all $i$. This proves that
$(Q^1, \dots , Q^ N) \in  \mathcal{M}^{\mathcal{Y}}.$

To prove  $ \mathcal{M}^{\mathcal{Y}}  \subseteq (C^\mathcal{Y})^{0}_1$ under the additional assumption, consider $(Q^1, \dots , Q^ N) \in \mathcal{M}^{\mathcal{Y}}$ and let $f \in C^{\mathcal Y}$. Then $f^i=k^i-l^i+Y^i \in L^1(\Omega, \mathcal{F}^i_T , P)$. We claim that $E_{Q^i} [k^i-l^i+Y^i ] \leq E_{Q^i} [Y^i ]$ for all $i$. 
This will imply the thesis, since then $ \sum_{i=1}^N E_{Q^i} [f^i ] \leq \sum_{i=1}^N E_{Q^i} [Y^i ] \leq 0$, as $ \mathcal{Y} \subseteq L^{1 }(\Omega, \mathbf{F}_T , P)$.
To show the claim, observe that $k^i-l^i+Y^i \in L^1(\Omega, \mathcal{F}^i_T , P)$ and $l^i \geq 0$ imply $(k^i+Y^i)^- \in L^1(\Omega, \mathcal{F}^i_T , P)$. Then $k^i=(k^i+Y^i)^+-(k^i+Y^i)^- -Y^i \geq -(k^i+Y^i)^- -Y^i $, so that $(k^i)^- \leq (-(k^i+Y^i)^- -Y^i )^- \in L^1(\Omega, \mathcal{F}^i_T , P)$ and so $(k^i)^- \in L^{1 }(\Omega, \mathcal{F}^{i}_T , Q^i)$.  
Now $k^i$ is the time $T$ value of a stochastic integral and 
Lemma \ref{Ek=0} 
then implies that $k^i  \in L^{1 }(\Omega, \mathcal{F}^{i}_T , Q^i)$ and $E_{Q^i} [k^i ]=0$. 
But from $k^i-l^i+Y^i \in L^1(\Omega, \mathcal{F}^i_T , Q^i)$ we conclude $l^i \in L^1(\Omega, \mathcal{F}^i_T , Q^i)$.  Thus $E_{Q^i} [k^i-l^i+Y^i ] \leq E_{Q^i} [k^i]+E_{Q^i}[Y^i ]=E_{Q^i} [Y^i ]$.
\end{proof}

\begin{proposition}\label{propositionFTAP3}
Suppose $\mcY$ is a convex cone and $K^{\mathcal{Y}}$ is closed in probability. Then $\mathbf{NCA(\mathcal{Y})}$ implies the existence of $z \in (C^\mathcal{Y})^{0}$ 
 such that $ z^i>0$  for all $i$. If additionally $\mcY$ fulfills condition \eqref{RY},  then $\mathbf{NCA(\mathcal{Y})}$ implies $(C^\mathcal{Y})^{0}_1\cap \mathcal{P}^N_e \not = \emptyset$ and $\mathcal{M}^{\mathcal{Y}} \cap \mathcal{P}^N_e \not = \emptyset$.
\end{proposition}
\begin{proof}
The proof is based on the multidimensional version of Kreps-Yan Theorem \ref{KY}. Indeed the convex cone $G:=C^\mathcal{Y}=K^\mathcal Y \cap L^{1 }(\Omega, \mathbf{F}_T,P)$ is closed in $ L^{1 }(\Omega, \mathbf{F}_T,P)$ and from \eqref{c} we see that the assumptions on $G$ in Theorem \ref{KY} are satisfied and thus there exists  $z \in L^{\infty }(\Omega, \mathbf{F}_T , P)$ satisfying $z^i>0 $ for all $i$ and  
\eqref{eq444}, 
that is $z\in (C^\mathcal{Y})^{0}$.
By Remark \ref{nicepolar} 
 $E[z^i]=E[z^j]$ for all $i$, $j$, and introducing $Q^i\in\mathcal{P}_{ac}$ defined by $\frac {dQ^i} {dP}:= \frac {z^i}  {E[z^i]}$, we have \begin{equation}\label{10}
(Q^1, \dots , Q^ N) \in  (C^\mathcal{Y})^{0}_1.
\end{equation}
Since $z^i>0$ for all $i$, we also have $Q^i\in\mathcal{P}_e$ so that $(Q^1,\dots,Q^N)\in (C^\mathcal{Y})^{0}_1\cap \mathcal{P}^N_e$ and therefore the latter set is not empty.
The last claim follows from the inclusion $(C^\mathcal{Y})^{0}_1 \subseteq \mathcal{M}^{\mathcal{Y}} $ proved in Lemma \ref{lemmaYL1}. 
\end{proof}

\begin{example} 
In the example in Section \ref{toyex} Item \ref{72C} we show that if $\mathcal{Y}$  does not fulfill \eqref{RY}, it is possible that  $\mathbf{NCA(\mathcal{Y})}$ holds true, $\mathbf{NA}_{i} \text{ } \forall i$  hold true and there exists a positive element in $(C^\mathcal{Y})^{0} $, but $\mathcal{M}^{\mathcal{Y}}=(C^\mathcal{Y})^{0}_1 =\emptyset $.  
\end{example}

\begin{proposition} \label{prop311}
(i) If there exists $z \in (C^\mathcal{Y})^{0}$ 
such that $z^i>0$ for all $i$, then $\mathbf{NCA(\mathcal{Y})}$ holds true. In particular, 
$(C^\mathcal{Y})^{0}_1  \cap \mathcal{P}^N_e \not = \emptyset$ implies $\mathbf{NCA(\mathcal{Y})}$.\\
(ii) Suppose that $\mathcal{Y} \subseteq L^{1 }(\Omega, \mathbf{F}_T , P)$. 
If $\mathcal{M}^{\mathcal{Y}} \cap \mathcal{P}^N_e \not = \emptyset$, then $\mathbf{NCA(\mathcal{Y})}$ holds true.
\end{proposition}

\begin{proof}
The proof of (i) is very simple. Observe that if $f \in C^\mathcal{Y}$ then $\sum_{i=1}^N E[z^if^i] \leq 0$ and if $f^i \geq 0$ for all $i$, this then implies $\sum_{i=1}^N E[z^if^i] = 0$, so that $E[z^if^i] = 0$ for all $i$. Given that $z^i>0$ for all $i$, then $f^i=0$ for all $i$, and $C^\mathcal Y \cap  L^{1 }_+(\Omega, \mathbf{F}_T,P)=\{0\}$.  \newline
To prove (ii),  take $(Q^1, \dots , Q^ N) \in  \mathcal{M}^{\mathcal{Y}}$ and  $Q^i \sim P$ for all $i$, and let $(k^i-l^i+Y^i)_i \in C^\mathcal Y=({\sf X}_{i=1} ^{N} ( K_i - L^{0 }_+(\Omega, \mathcal{F}^i_T,P) ) + \mathcal Y ) \cap  L^{1 }_+(\Omega,  \mathbf{F}_T,P)$. Then $k^i-l^i+Y^i \geq 0$ and thus $k^i+Y^i \geq l^i \geq 0$, $k^i \geq -Y^i \in L^{1 }(\Omega, \mathbf{F}_T , P) $ for all i. Hence $(k^i)^- \in L^{1}(\Omega, \mathcal{F}^i_T,P) \subseteq L^{1}(\Omega, \mathcal{F}^i_T,Q^i) $. 
By Lemma \ref{Ek=0} 
we have $k^i \in L^{1}(\Omega, \mathcal{F}^i_T,Q^i) $ and $E_{Q^i}[k^i]=0$. Thus $k^i+Y^i \in L^{1}(\Omega, \mathcal{F}^i_T,Q^i)$ and from $k^i-l^i+Y^i \in  L^{1 }(\Omega,  \mathcal{F}^i_T,Q^i)$ we also deduce $ l^i \in L^{1 }(\Omega,  \mathcal{F}^i_T,Q^i)$. Thus $E_{Q^i}[k^i-l^i+Y^i] \leq E_{Q^i}[k^i]+E_{Q^i}[Y^i]=E_{Q^i}[Y^i]$, $\sum_{i=1}^N E_{Q^i}[k^i-l^i+Y^i] \leq \sum_{i=1}^N E_{Q^i}[Y^i]  \leq 0$, as $(Q^1, \dots , Q^ N) \in  \mathcal{M}^{\mathcal{Y}}$ and $ \mathcal{Y} \subseteq L^{1 }(\Omega, \mathbf{F}_T , P)$.
From $(k^i-l^i+Y^i) \in  L^{1 }_+(\Omega,  \mathcal{F}^i_T,P)$ for all $i$, we also get that $\sum_{i=1}^N E_{Q^i}[k^i-l^i+Y^i]  \geq 0$, so that $\sum_{i=1}^N E_{Q^i}[k^i-l^i+Y^i] = 0$, which then implies $E_{Q^i}[k^i-l^i+Y^i] = 0$ for all $i$ and $k^i-l^i+Y^i=0$ for all $i$.
\end{proof}

As a consequence of the previous two propositions we thus obtain

\begin{theorem}[Collective FTAP]\label{FTAP3}
Suppose that $\mcY$ is a convex cone and that $K^{\mathcal{Y}}$, given in \eqref{KyandCY}, is closed in probability. Then 
\begin{equation*}
 \mathbf{NCA(\mathcal{Y})} \iff \text { there exists }
 z \in (C^\mathcal{Y})^{0} 
\text{ such that } z^i>0 \text{ for all }i.     
\end{equation*}
If additionally $\mcY$ contains  $\mathbb R_0 ^N$, then \begin{equation}\label{CFTAP111}
 \mathbf{NCA(\mathcal{Y})} \iff  (C^\mathcal{Y})^{0}_1  \cap \mathcal{P}^N_e \not = \emptyset,
 \end{equation}
and if we also assume that
$\mathcal{Y}  \subseteq L^{1 }(\Omega, \mathbf{F}_T , P)$ then 
\begin{equation}\label{CFTAP11}
 \mathbf{NCA(\mathcal{Y})} \iff \mathcal{M}^{\mathcal{Y}} \cap \mathcal{P}^N_e \not = \emptyset.
 \end{equation}
 
\end{theorem}

In Section \ref{sec:Kyclosed} we provide conditions on the convex cone $\mathcal{Y}$ and on the market model ensuring that the condition  $\mathbf{NCA(\mathcal{Y})}$ implies that the set $K^{\mathcal{Y}}$ is closed in probability.

\section {Collective Super-replication Price}\label{secSuperNew}

In this section, we work in the multiperiod setting outlined in Section \ref{setting}.

\subsection{Definitions and Properties}

\noindent 
We now add to the market $N$ claims $$(g^1,\dots,g^N ) \in L^{0 }(\Omega, \mathbf{F}_T,P)$$
and we are now interested in the prices $(p^1,\dots,p^N) \in \mathbb R^N$ for the claims $(g^1,\dots,g^N)$ which do not allow for Collective Arbitrage in this extended market.

For this we introduce in the following Definition two novel concepts, namely those of super-replication of all claims and super-replication of all claims with exchanges, and introduce a notation for those of super-replication within each market and in the full market, respectively.
\begin{definition}[Super-replication functionals] \label{defSup} 
$\,$
Let $g=(g^1,\dots,g^N ) \in L^{0 }(\Omega, \mathbf{F}_T,P)$.
\begin{itemize}
    \item \textbf{Super-replication 
    of all $N$ claims} 
    \begin{equation*}
    \ro (g):=\inf \left \{ \sum_{i=1} ^N m^i \mid  \exists  k \in {\sf X}_{i=1} ^{N} K_i, m \in \mathbb{R}^N \text{ s.t. } m^i + k^i \geq g^i \text{  }  \forall i  \right \}.
    \end{equation*}

    \item \textbf{Super-replication 
    of all claims with exchanges} $\mathcal{Y}$
    \begin{equation*}
    \mcR(g):=\inf \left \{ \sum_{i=1} ^N m^i \mid \exists  k \in {\sf X}_{i=1} ^{N} K_i, m \in \mathbb{R}^N,  Y \in \mathcal{Y} \text{ s.t. } m^i+ k^i + Y^i \geq g^i  \text{  }  \forall i  \right \}.  
    \end{equation*}

    \item\textbf{(Classical) Super-replication within each market}
     \begin{equation}
     \label{rhoi+}
         \pii(g^i):=\inf \left \{ m^i \mid  \exists  k^i \in   K_i, m^i \in \mathbb{R} \text{ s.t. } m^i + k^i \geq g^i \right \}.
     \end{equation}
     \item \textbf{(Classical) Super-replication for the full market $K$} (Definition \ref{NAclassic}) of $G\in L^0(\Omega,\mathcal{F},P)$
     \begin{equation*}
         \rho_{full,+}(G):=\inf \left \{ m \mid  \exists  k \in  K, m\in\R\text{ s.t. } m + k \geq G \right \}.
     \end{equation*}
\end{itemize}
\end{definition}

\noindent The functional $\mcR(g) $ and $\ro (g)$ both represents the minimal amount needed to super-replicate simultaneously all claims $(g^1,...,g^N)$. For the Collective super-replication price $\mcR(g)$ we allow an additional exchange among the agents, as described by $\mathcal{Y}$. 

We say that the value $p \in \mathbb R$ is a \emph{selling price} for the claims $(g^1,\dots,g^N):=g $ if there exists a vector $(p^1,\dots,p^N) \in \mathbb R^N$, where $p^i$ is a selling price for $g^i$,  such that $\sum_{i=1} ^N p^i=p$.

\begin{proposition}\label{propselling}
  If $p \in \mathbb R$ is a selling price for $g$ satisfying $p > \mcR(g)$ then 
\begin{equation}\label{compatibleprice}
 {\exists}k \in {\sf X}_{i=1} ^{N} K_i, \, Y \in \mathcal Y, {\text{ s.t. }}  k^i+Y^i +(p^i-g^i) \geq 0 \text{   } \forall i \text { and }  P(k^j+Y^j +(p^j-g^j)) >0) > 0  \text{\ for some  } j
\end{equation}
    holds, namely there exists a $\mathbf{CA}(\mcY)$ in the extended market. Thus $\mcR(g)$ is the maximum selling price compatible with $\mathbf{NCA(\mathcal{Y})}$ in the extended market. 
\end{proposition}

\begin{proof}
Suppose that $p > {\mcR} (g)$. Then there exist $ \widehat m \in \mathbb R^N$, $\widehat Y \in \mcY $ and $\widehat k \in {\sf X}_{i=1} ^{N} K_i$ such that $\sum_{i=1} ^N \widehat m^i < p$ and  $ \widehat m ^i  + \widehat k^i + \widehat Y^i \geq g^i$ for all $i$.
Set $p^i:=\widehat m^i + \frac {p-\sum_{i=1} ^N \widehat m^i} {N} >\widehat m^i $ for all $i$. Then $\sum_{i=1} ^N p^i = p$. 
Suppose we sell $g^i$ at prices $p^i$, so that $p$ is a selling price for $g$, and we invest in the strategy given by $(\widehat m,\widehat k,\widehat Y)$. Then, for each $i$, the value at time $t=0$ is $(p^i - \widehat m^i ) > 0$, for all $i$, and its time $T$ value is $-g^i+\widehat m^i+\widehat k^i+\widehat Y^i \geq 0$. By investing $(p^i - \widehat m^i)$ at time $0$ in the bond, we have constructed a strategy, in the extended market, that at initial time has $0$ value and at final time $T$ has the value $0 < (p^i - \widehat m^i)+(-g^i+\widehat m^i+\widehat k^i+\widehat Y^i)=(p^i -g^i)+\widehat k^i+\widehat Y^i $ for all $i$, namely \eqref{compatibleprice}  is fulfilled.  
\end{proof}

\begin{proposition}    
    \begin{equation}
    \ro (g)=\sum_{i=1}^N\pii(g^i). \label{RR}
    \end{equation}

\end{proposition}
\begin{proof} 
   Observe that if $\pii(g^i)$ is finite for each $i$, the relation is straightforward. If $\pii(g^i)=-\infty$ for some $i$, then $\ro(g)=-\infty$ is still satisfied under the convention $+\infty-\infty=-\infty$. The remaining case is  $\pii(g^i)>-\infty$ for all $i$ and at least one is $+\infty$. In that case and \eqref{RR} reduces to $+\infty=+\infty$.
\end{proof}
We now turn to some simple properties of $\mcR$. 
In general $\mcR$ may take values $\pm \infty$ and 
in Section \ref{toyex}  Item \ref{72D} we provide an example where there exists a global Arbitrage, $\mathbf{NCA(\mathcal{Y})}$ holds true, $\mathcal{Y}$ does not fulfill  \eqref{RY} and $\mcR (g)=-\infty$ for all $g$.

We already mentioned in the Introduction that cooperation helps to reduce the cost of super-replication: $\mcR \leq \ro$. We thus define $(\ro(g)-{\mcR}(g)) \geq 0$ as the  (selling) \emph{value of cooperation} with respect to $g$.
 As both $\ro$ and $\mcR$ are positively homogeneous, we may also define a quantitative characteristic of the whole market, which we call the  
\textit{market cooperation value} 
$$V^\mcY:=\sup_{ g \in L^{\infty }(\Omega, \mathbf{F}_T,P) : ||g||_{\infty} \leq 1} (\ro(g)-{\mcR}(g)).$$ 

\begin{lemma}\label{PropertiesRho}
Suppose that  $\mathcal{Y}$ fulfills  \eqref{RY} and that $\mathbf{NCA(\mathcal{Y})}$ holds true. Then
\begin{enumerate}
\item $\mcR (0)=0$.
\item $\mcR$ is cash additive: $\mcR(g+c)=\mcR(g) +\sum_{i=1} ^N c^i$, for all  $c \in \mathbb R^N$.
\item $\mcR$ is monotone increasing with respect to the partial order of $L^{0 }(\Omega, \mathbf{F}_T,P) $.
\item $\mcR : L^{\infty }(\Omega, \mathbf{F}_T,P)  \rightarrow \mathbb R$ is Lipschitz continuous. 
\item $\mcR$ is positively homogeneous and 
    $V^{\mcY} \in [0,2N] $.
    
\item  $\mcR=\rho_+^{(\mcY+\mathbb{R}^N_0)}$.
\end{enumerate}
Items 2., 3. and 6. hold even without assuming $\mathbf{NCA(\mathcal{Y})}$ and condition \eqref{RY}.
\end{lemma}

\begin{proof}
We only prove Item 1, as  the remaining Items are standard. By definition, $\mcR (0) \leq 0$. Suppose by contradiction that $\mcR (0) < 0$. Then, by definition of $\mcR (0)$, there exist $m \in \mathbb R^N$, $k \in {\sf X}_{i=1} ^{N} K_i $ and  $Y \in \mathcal{Y}$  such that, for all  
  $i$,  $m^i + k^i + Y^i \geq 0$ and $\sum_{i=1} ^N m^i =-\varepsilon$ for some $\varepsilon >0$. Hence  $k^i + \left (Y^i + m^i + \frac {\varepsilon} N \right ) \geq \frac {\varepsilon} N$. As the vector with components $(m^i + \frac {\varepsilon} N )$ is in $\mathbb R_0 ^N$, we deduce that there exist $k \in {\sf X}_{i=1} ^{N} K_i $  and  $\widehat Y \in \mathcal{Y}$  such that $k^i + \widehat {Y}^i \geq \frac {\varepsilon} N$, for all $i$, that is a Collective Arbitrage.
\end{proof}

Under the specific assumptions on the market stated in Proposition \ref{prop47}, the Collective super-replication price  $\rho_+^{\mathcal{Y}_0}(g)$ of the $N$ claims $(g^1,\dots,g^N)=g$ with the exchanges $\mcY=\mcY_0$ is equal to the classical super-replication of the single claim $\left (\sum_{i=1}^Ng^i \right )$ for a representative agent who is allowed to invest in the whole market. 


\begin{proposition}\label{prop47}  
   In the case $\mathbb{F}^i=\mathbb{F}^j$ for all $i,j$ and $\mathcal{Y}=\mathcal{Y}_0$ given in \eqref{Y00}, then 
    \begin{equation}
    \label{rhoy0isrhofull}
     \rho_+^{\mathcal{Y}_0}(g)=\rho_{full,+}\left(G\right)    \text{ for } G=\sum_{i=1}^Ng^i.
    \end{equation}
\end{proposition}
\begin{proof}
 The inequality $\rho_+^{\mathcal{Y}_0}(g)\geq \rho_{full,+}\left(G\right)$ follows observing that for any vector $m\in\R^N$ satisfying the constraints for   $\rho_+^{\mathcal{Y}_0}(g)$, $\sum_{i=1}^Nm^i$ satisfies the constraints for $\rho_{full,+}\left(G\right)$ simply by summation. For the converse, take $G$ as given, and $m\in\R$, $k\in K$ such that $m+k\geq G$. Then $m+k+L=G$ for some $L\in L^0_+(\Omega,\mathcal{F},P)$ and setting $\widetilde{k}^i:=\frac1Nk$ (which still belong to $K$) and $\widetilde{L}^i=\frac1NL$ we have $\sum_{i=1}^N(\widetilde{m}^i+\widetilde{k}^i+L^i)=\sum_{i=1}^Ng^i$, for any $\widetilde{m}\in\R^N$ such that $\sum_{i=1}^N\widetilde{m}^i=m$. Set additionally $\widetilde{Y}:=-(\widetilde{m}+\widetilde{k}+\widetilde{L})+g$. Then $\widetilde{Y}\in\mathcal{Y}_0$ by direct verification. Thus, $\widetilde{m}+\widetilde{k}+\widetilde{Y}\geq g$ (componentwise). Moreover $\widetilde{k}+\widetilde{Y}\in (K)^N+\mathcal{Y}_0={\sf X}_{i=1}^N K_i+\mathcal{Y}_0$, where the last is due to Lemma \ref{auxNCAimplNA}. Hence $\widetilde{k}+\widetilde{Y}=k+Y$ for some $k\in {\sf X}_{i=1}^N K_i$ and $Y\in\mathcal{Y}_0$. Hence $\widetilde{m}+k+Y\geq g$, and $\rho_+^{\mathcal{Y}_0}(g)\leq \rho_{full,+}\left(G\right)$ follows.
\end{proof}

Set $ \mathcal R^{\mathcal{Y}} (g):=\left \{ m \in \mathbb R^N \mid  \exists  k \in {\sf X}_{i=1} ^{N} K_i, Y \in \mathcal Y \text{ s.t. } m^i+ k^i +Y^i \geq g^i \text{  }  \forall i  \right \}$
so that\\ $\mcR(g)=\inf\{\sum_{i=1}^Nm^i\mid m\in \mathcal R^{\mathcal{Y}} (g)\}.$
We say that $\mcR(g)$ is attained by $\widehat m $, if there exists $\widehat{m} \in \mathcal R^{\mathcal{Y}} (g)$ such that $\mcR(g)=\sum_{i=1}^N\widehat m^i$. In this case, we say that $\widehat m \in \mathbb R^N$, $\widehat Y \in \mcY $ and $\widehat k \in {\sf X}_{i=1} ^{N} K_i$ are optimizers of $\mcR(g)$ if  $\widehat m ^i  + \widehat k^i + \widehat Y^i \geq g^i$ for all $i$.

\subsubsection{Super Replication of Any One Claim}

Differently from the one dimensional case, if one has many contingent claims  $(g^1,...,g^N) \in L^{0 }(\Omega, \mathbf{F}_T,P)$ to super-replicate, there are many possible definitions to consider. For example, in the following definitions,  $\pi_+^N(g)$ is an immediate extension to the multidimensional setting of the classical one dimensional notion of  super-replication.

\begin{definition}\label{defpi} Let $g=(g^1,\dots,g^N ) \in L^{0 }(\Omega, \mathbf{F}_T,P)$.
\begin{align*}
\pi_+^N(g)&:=\inf\{m \in \mathbb{R} \mid  \exists k \in {\sf X}_{i=1} ^{N} K_i  \text{ s.t. } m+k^i \geq g^i \text{  }  \forall i   \}, \\
{\pi_+^{\mathcal{Y}}}(g)&:=\inf\{m \in \mathbb{R} \mid  \exists k \in {\sf X}_{i=1} ^{N} K_i,  Y \in \mathcal{Y} \text{ s.t. } m+k^i+Y^i \geq g^i  \text{  }  \forall i  \}.
\end{align*}
\end{definition}

\noindent Notice that $\pi^N_+(g)$ and $\pi_+^{\mathcal{Y}}(g)$  both represent the minimal amount needed to replicate \emph{any one}  claim among $g^1, \dots,g^N$, but they differ as for $\pi^N_+$  no  cooperation is allowed.


Both $\pi^N_+(g)$ and $\pi_+^{\mathcal{Y}}(g)$ guarantees the possibility to super-replicate just one claim among  $g^1, \dots,g^N$  (so not necessarily all the contingent claims $g^1, \dots,g^N$), no matter which one is chosen.  \\

For example consider a multinational Corporation having several national-based subsidiary companies each one allowed to invest only in the national market and each one facing a risk $g^i$. The Corporation could be interested in knowing the minimal amount that is required to \emph{save any one } of the subsidiary companies. This would amount in computing $\pi^N_+(g)$, or $\pi_+^{\mathcal{Y}}(g)$ if cooperation is allowed. One can think of a principal/head of corporation who wants to make sure to have an amount of money $m$ such that at terminal time she  can cover one option of her choice in a basket $g_1,\dots,g_N$. Each agent $i$ in the system invests in its own market $K_i$ as if she could use the amount $m$ to super-replicate $g^i$.

\begin{remark} One might also define all intermediate possibilities of super-replication costs, namely the cost to replicate  $n \leq N$  contingent claims out of $(g^1,...,g^N)$. But in this paper we will only discuss the two extreme cases represented by ${\pi_+^{\mathcal{Y}}}$ and $\mcR $ (or $ \pi_+^N$ and $\ro$).
\end{remark} 
\begin{proposition}\label{22}
Recall from \eqref{rhoi+} the definition $\pii(f)=\inf\{m \in \mathbb{R} \mid \exists k \in \ki \text{ s.t. } m+k \geq f \}$ for $f \in L^{0}(\Omega,\fiT,P)$. Then
\begin{equation}
\pi^N_+(g)= \max_i \{\pii(g^i) \} \label{PP}.
\end{equation}

\end{proposition}
\begin{proof}
 From the definitions it immediately follows that $\pi^N_+(g) \leq \max_i \{\pii(g^i) \}$. By contradiction suppose that $\pi^N_+(g)< \max_i \{\pii(g^i) \}=\rho_{j,+}(g^j)$, for some $j$. By definition of $\pi^N_+(g)$, there exists $m <  \rho_{j,+}(g^j)$ such that $m+k^i \geq g^i$ for all $i$, thus such $m$ satisfies the constraints in $ \rho_{j,+}(g^j)$ and thus $\rho_{j,+}(g^j) \leq m$.
\end{proof}

We now show that under \eqref{RY}, the relation between $\mcR(g)$ and $ \pi_+^{\mathcal{Y}}$ is very simple and thus in the remaining part of the paper we will mainly focus our analysis on $\mcR(g)$.

\begin{proposition}\label{propositionR0}
Suppose that  $\mathcal{Y}$ fulfills  \eqref{RY}. Then 
\begin{equation}\label{rhoNpi} 
\mcR(g)=N \pi_+^{\mathcal{Y}}(g).
\end{equation}
Moreover, $ m_* \in \mathbb R$, $ Y_* \in \mcY $ and $ k_* \in {\sf X}_{i=1} ^{N} K_i $ are optimizers of $\pi_+^{\mathcal{Y}}(g)$ iff  $\widehat m \in \mathbb R^N$, $\widehat Y \in \mcY $, $\widehat k \in {\sf X}_{i=1} ^{N} K_i$ are optimizer of $\mcR(g)$, with the following conditions:
$\widehat m ^i=N p^i$, where $p \in \mathbb R^N$ is any vector satisfying $\sum_{i=1}^N p^i=\pi_+^{\mathcal{Y}}(g)$,  $\widehat k = k_*$, $\widehat Y^i=Y_*^i+m_*-\widehat m^i$.

\end{proposition} 
\begin{proof}
Let $m \in \mathbb R^N$ and set $\Bar{m}:=\sum_{i=1} ^N m^i$, then the vector with components $(m^i - \frac {\Bar m} {N})$ is in $\mathbb R^N_0$.  If $Y\in \mathcal{Y}$, then by \eqref{RY}, the vector $\widehat Y$ with components $ \widehat Y^i =\left (Y^i + m^i - \frac {\Bar{m}} {N} \right )$ is in $ \mathcal{Y}$ and
\begin{align*}
\mcR (g)
& = \inf \left \{ \sum_{i=1} ^N m^i \mid   \exists k \in {\sf X}_{i=1} ^{N} K_i, Y \in \mathcal{Y}  \text{ s.t. } m^i + k^i + Y^i \geq g^i \text{  } \forall i  \right \} \\
&=  \inf \left \{ \sum_{i=1} ^N m^i \mid   \exists k \in {\sf X}_{i=1} ^{N} K_i, Y \in \mathcal{Y}  \text{ s.t. } \frac {\bar m} {N} + k^i + \left (Y^i + m^i - \frac {\bar m}  {N} \right ) \geq g^i \text{  }  \forall i \right \} \\
& =  \inf \left \{ \bar m \mid   \exists k \in {\sf X}_{i=1} ^{N} K_i, \widehat Y \in \mathcal{Y}  \text{ s.t. } \frac {\bar m} {N} + k^i + \widehat Y^i  \geq g^i  \text{  } \forall i  \right \} = N \pi_+^{\mathcal{Y}}(g).
\end{align*}
The remaining statements are obvious computations.
\end{proof}

\begin{remark}\label{remR0}
We showed in Proposition \ref{propositionR0} that $ \pi_+^{\mathcal{Y}}=\frac 1 N \mcR$ which, in particular, illustrates the symmetrization allowed by including $\mathbb R_0^N$ in $\mcY$ in \eqref{RY}.
Instead, from Proposition \ref{22}, we have 
$\pi^N_+(g)= \max_i \{\pii(g^i) \}$ and from \eqref{RR} we deduce $ \pi^N_+ \geq \frac 1 N \ro$,  so that $ \pi^N_+$ could be 
strictly greater than $\frac 1 N \ro$.
In the toy example of Section \ref{toyex} Item \ref{72A} we show that it is possible to profit from cooperation, meaning that $\pi_+^{\mathcal{Y}}(g)<\pi^N_+(g)$.
\end{remark}

\subsection{Pricing-hedging Duality}
\label{secpriceduality}

\begin{assumption}
\label{assdualfair}
 $\mcY \subseteq L^{0 }(\Omega, \mathbf{F}_T , P)$ is a convex cone, $\mathbf{NCA(\mathcal{Y})}$ holds true and $K^\mathcal Y$ is closed in $L^{0 }(\Omega, \mathbf{F}_T , P)$.
\end{assumption}

This assumption implies that $C^\mathcal Y$  is closed in $L^{1 }(\Omega, \mathbf{F}_T , P)$. We stress that under appropriate conditions on $\mathcal Y$ and on the market model (see Section \ref{sec:Kyclosed}), the assumption that $K^\mathcal Y$ is closed in $L^{0 }(\Omega, \mathbf{F}_T , P)$  is redundant.

As in the classical case and using the same argument applied in Remark \ref{remIntegrability}, we assume that $(g^1,...,g^N) \in  L^{1 }(\Omega, \mathbf{F}_T,P)  $ and we recall that $X^j_t \in L^1(\Omega,\fit, P)$ for each $t\in \mathcal T$ and each $i$ and $ j\in (i)$.

\begin{theorem}[The pricing-hedging duality for $\pi_+^{\mathcal{Y}}$ and $\mcR $]  
\label{propsupermulti}

Suppose Assumption \ref{assdualfair} holds true. For any $g \in L^{1 }(\Omega, \mathbf{F}_T,P)$ we have 
\begin{equation}\label{supgen1}
{\pi_+^{\mathcal{Y}}}(g)=\sup_{z \in (C^\mathcal{Y})^{0} \setminus \{0\}} \frac {\sum_{i=1}^N E[z^ig^i]} {\sum_{i=1}^N E[z^i]},
\end{equation} 

\begin{equation}\label{rhodualgeneral}
  \mcR(g) =\inf\left\{\sum_{i=1}^N  m^i  \mid m\in\R^N ,\sum_{i=1}^N m^i E[z^i] \geq \sum_{i=1}^N E[z^ig^i] \text{  }  \forall z \in (C^\mathcal{Y})^0 \right\}. \\ 
\end{equation}
Suppose additionally that  $\mathcal{Y}$ fulfills  \eqref{RY} then
\begin{align}\label{eqpiY}
{\pi_+^{\mathcal{Y}}}(g)&=\sup_{Q \in (C^\mathcal{Y})^{0}_1}  \frac 1 N \sum_{i=1}^N E_{Q^i} [ g^i ],\\
\mcR(g)&=\sup_{Q \in (C^\mathcal{Y})^{0}_1 }  \sum_{i=1}^N  E_{Q^i } [g^i] \label{suprho}
\end{align} 
and, if ${\pi_+^{\mathcal{Y}}}(g) < +\infty $ (resp. $\mcR(g)<+\infty $) then $\pi_+^{\mathcal{Y}}(g)$ (resp. $\mcR(g)$) is attained by some $ m_* \in \mathbb R$ (resp. by some $\widehat m \in \mathbb R^N $). If additionally, $\mathcal{Y}  \subseteq L^{1 }(\Omega, \mathbf{F}_T , P)$ then we can replace $(C^\mathcal{Y})^{0}_1$ with $\mathcal M ^{\mcY}$ in \eqref{eqpiY} and \eqref{suprho}.
\end{theorem}

\begin{proof}
First observe that by Proposition \ref{propositionFTAP3} and $\mathbf{NCA(\mathcal{Y})}$, there exists a non zero element in $(C^\mathcal{Y})^{0}$.  The condition $\exists k^i \in \ki$ s.t.   $m+k^i +Y^i \geq g^i$ holds iff $g^i-m-Y^i \in ( K_i - L^{0 }_+(\Omega, \mathcal{F}^i_T,P) )$.
Thus the conditions $\exists k \in {\sf X}_{i=1} ^{N} K_i$ and $\exists Y \in \mathcal Y$ s.t. $m+k^i +Y^i \geq g^i$ for all $i$ hold iff 
$\exists Y \in \mathcal Y$ s.t. $g-m1-Y \in {\sf X}_{i=1} ^{N} ( K_i - L^{0 }_+(\Omega, \mathcal{F}^i_T,P) ) $ 
iff $g-m1 \in ({\sf X}_{i=1} ^{N} ( K_i - L^{0 }_+(\Omega, \mathcal{F}^i_T,P) )+\mathcal{Y})=K^{\mathcal Y}$. 
Since $g \in L^{1 }(\Omega, \mathbf{F}_T,P)  $, we obtain
$$ {\pi_+^{\mathcal{Y}}}(g)=\inf\{m \in \mathbb{R} \mid g-m1 \in C^{\mathcal Y} \}.$$
We now apply a standard argument. Since by Assumption \ref{assdualfair} the convex cone $K^\mathcal Y$ is closed in probability, $C^\mathcal Y =K^\mathcal Y \cap L^{1 }(\Omega, \mathbf{F}_T,P)$ is norm-closed in $L^{1 }(\Omega, \mathbf{F}_T,P)$, hence $C^\mathcal Y$ is a convex cone $\sigma(L^{1 }(\Omega, \mathbf{F}_T,P),L^{\infty}(\Omega, \mathbf{F}_T,P))$ closed.  The bipolar theorem\footnote{See e.g. \cite{Aliprantis} Theorem 5.103.2, in conjunction with footnote 8 on page 215, for the notation and characterization of polars and bipolars of cones. Here, we are using the K\"{o}the one-sided polar or one-sided polar in the reference.} implies that $ (C^\mathcal{Y})^{00}=C^\mathcal{Y}.$
Thus 
\begin{align*}
{\pi_+^{\mathcal{Y}}}(g)
&=\inf\{m \in \mathbb{R} \mid g-m1 \in C^\mathcal{Y}\} \\
&=\inf\{m \in \mathbb{R}  \mid \sum_{i=1}^N E[z^i(g^i-m]) \leq 0 \text{  }  \forall z \in (C^\mathcal{Y})^0 \} \\ 
&=\inf\{m \in \mathbb{R} \mid m  \geq \frac {\sum_{i=1}^N E[z^ig^i]} {\sum_{i=1}^N E[z^i]} \text{  } \text{  }  \forall z \in (C^\mathcal{Y})^0 \setminus \{0\} \} \\
&=\sup_{z \in (C^\mathcal{Y})^0 \setminus\{0\}} \frac {\sum_{i=1}^N E[z^ig^i]} {\sum_{i=1}^N E[z^i]}, 
\end{align*}
which shows \eqref{supgen1}.
 Using a similar argument one can directly prove \eqref{rhodualgeneral}. If $\mathcal{Y}$ fulfills  \eqref{RY} then $E[z^i]=E[z^j]$ for all $i,j$ and all $z \in (C^\mathcal{Y})^{0}$. Then using  $\frac {dQ^i} {dP}:= \frac {z^i}  {E[z^i]}$, equation \eqref{eqpiY} follows from \eqref{supgen1} and equation \eqref{suprho} follows from \eqref{eqpiY}  and \eqref{rhoNpi} (alternatively,  \eqref{suprho} follows directly from \eqref{rhodualgeneral}).
Thanks to Proposition \ref{propositionR0} we have to prove the existence of the optimizer only for ${\pi_+^{\mathcal{Y}}}(g)$. Define $m_*=RHS \eqref{eqpiY}$. By the assumption and the condition $\mathbf{NCA(\mathcal{Y})}$, that assures $(C^\mathcal{Y})^{0}_1 \not= \emptyset$, $m_* \in \mathbb R$. Then by the above argument $g-m_* 1 \in C^{\mathcal Y} $, thus $g-m_* 1 = k_*-l_*+Y_*$ for some  $k_* \in {\sf X}_{i=1} ^{N} K_i$, $l_* \in L^{0 }_+(\Omega, \mathbf{F}_T,P) ) $ and $ Y_* \in \mathcal Y$. Thus $ m_*  +  k_*^i + Y_*^i \geq g^i$ for all $i$, showing that $m_*$, $k_*$ and $ Y_*$ are optimizers for ${\pi_+^{\mathcal{Y}}}(g)$.  The last sentence follows from $(C^\mathcal{Y})^{0}_1 = \mathcal{M}^{\mathcal{Y}}$, as proved in Lemma \ref{lemmaYL1}.
\end{proof}

In the example in Section \ref{toyex} Item \ref{72D}, $(C^\mathcal{Y})^{0}_1 =\emptyset$ and so we will use the representation in the formula \eqref{rhodualgeneral} to compute $\mcR(g)$.

\begin{remark}[Comparison between $\mcR$ and $\ro$, and between ${\pi_+^{\mathcal{Y}}}$ and $\pi^N_+ $]  
Under Assumption \ref{assdualfair}, suppose that  $\mathcal{Y}$ fulfills  \eqref{RY}. Then
\begin{equation*}
\mcR(g)=\sup_{Q \in (C^\mathcal{Y})^{0}_1 }  {\sum_{i=1}^N  {E_{Q^i } [g^i]}} \leq \sup_{Q \in M^{\mathcal{Y}}}  {\sum_{i=1}^N  {E_{Q^i } [g^i]}} \leq  {\sum_{i=1}^N  \sup_{Q^i  \in \mi} {E_{Q^i } [g^i]}} = \sum_{i=1}^N \pii(g^i)=\ro (g),
\end{equation*} 
and
\begin{align*}
{\pi_+^{\mathcal{Y}}}(g)&=\sup_{Q \in (C^\mathcal{Y})^{0}_1 } \frac 1 N \sum_{i=1}^N E_{Q^i} [ g^i ]  \leq  \sup_{Q \in \mathcal{M}^{\mathcal{Y}}} \frac 1 N \sum_{i=1}^N E_{Q^i} [ g^i ] \\
&\leq \max_i \left \{ \sup_{Q^i  \in \mi} {E_{Q^i } [g^i]} \right \} = \max_i \{\pii(g^i) \}=\pi^N_+(g).
\end{align*} 
\end{remark}

\subsection{
Collective and Individual Super-replication Costs}\label{secFairness} 

We let $ Q=(Q^1,...,Q^N)$, $Q^i \in \mathcal P_{ac},$ and recall the notation \ref{productnotation} for $L^{1 }(\Omega, \mathbf{F}_T,\mathbf{{Q}}) $. We set
\begin{align}\label{121}
A_{Q^i}&:= \{Y^i \in L^1{(\Omega, \mathcal{F}_T^i, Q}^i ) \mid E_{{Q}^i } [Y^i]=0 \}, \quad A_Q:= A_{Q^1} \times \dots \times  A_{Q^N}, \notag \\
\rho_{A_Q}(g)&:=  \inf \left \{ \sum_{i=1} ^N m^i  \mid   \exists k \in {\sf X}_{i=1} ^{N} K_i, Y \in A_Q  \text{ s.t. } m^i + k^i + Y^i \geq g^i \text { } \forall i \right \},\notag \\
\rho_{Q^i }(g^i)&:=\inf \left \{ m \in \mathbb R \mid   \exists k^i \in K_{i}, Y^i \in A_{Q^i}   \text{ s.t. } m + k^i + Y^i \geq g^i \right \} \notag \\
&=\inf \left \{ E_{{Q}^i } [{Y}^i]  \mid Y^i \in L^1{(\Omega, \mathcal{F}_T^i, Q}^i )   \text{ s.t. } \exists k^i \in K_{i}  \text{ with }  k^i + Y^i \geq g^i \right \}.
\end{align}
Equation \eqref{121} is obtained by removing  the constant $m$ and incorporating it in the exchange variables and by replacing $(m^i + Y^i) $ with $ Y^i$.

Observe that in \eqref{121} the $Y^i$'s are not any more constrained to have zero expectation.
For the fixed agent $i$ and for the fixed probability $Q^i$, $\rho_{Q^i }(g^i)$ is thus the smallest \textit{cost} needed to super-replicate the claim $g^i$, \textit{independently} from any other agent, namely is the \emph{individual super-replication cost}. Observe that the infimum is taken with respect to all possible $Y^i \in L^1{(\Omega, \mathcal{F}_T^i, Q}^i )$ and without any reference to other agents. \\

Under the assumption \ref{assdualfair} and if $\mcR(g)<\infty$ we know from Theorem \ref{propsupermulti} that the pricing hedging duality holds true and $\mcR(g)$ is attained, so that  there exists an optimizer $(\widehat m, \widehat k, \widehat Y ) \in \mathbb R ^N \times {\sf X}_{i=1} ^{N} K_i \times  \mathcal{Y} $ of $\mcR(g)$.
\begin{proposition}\label{fairness}
Assume that \ref{assdualfair} holds, $\mathcal{Y}$ fulfills  \eqref{RY} and that $(g^1,...,g^N) \in  L^{1 }(\Omega, \mathbf{F}_T,P)$.  Suppose that $\mcR(g)<\infty$, that there exists $\widehat Q \in M^{\mathcal{Y}} $ satisfying
$$\mcR(g)
=\sup_{Q \in M^{\mathcal{Y}}}  {\sum_{i=1}^N  {E_{Q^i } [g^i]}}  ={\sum_{i=1}^N  {E_{\widehat{Q}^i } [g^i]}}, $$
and one optimizer $(\widehat m, \widehat k, \widehat Y )$  of $\mcR(g)$ such that $ \widehat Y  \in  \mathcal{Y} \cap L^{1 }(\Omega, \mathbf{F}_T,\mathbf{\widehat{Q}}) $. Then
\begin{enumerate} 
\item Any optimizer $(\widehat m, \widehat Y ) \in \mathbb R ^N \times \mathcal{Y} \cap L^{1 }(\Omega, \mathbf{F}_T,\mathbf{\widehat{Q}}) $ of $\mcR(g)$ satisfies
\begin{equation}\label{AAB}
E_{\widehat{Q}^i } [g^i]-\widehat{m}^i ={E_{\widehat{Q}^i } [\widehat{Y}^i]} \text { for all } i,
\end{equation}
\begin{equation}\label{Y0}
{\sum_{i=1}^N  { E_{\widehat{Q}^i } [\widehat{Y}^i] } } = 0.
\end{equation}
\item Among the optimizer of $\mcR(g)$ we can always find one optimizer
$(\tilde m, \tilde Y ) \in \mathbb R ^N \times \mathcal{Y} \cap L^{1 }(\Omega, \mathbf{F}_T,\mathbf{\widehat{Q}}) $ such that   
$$E_{\widehat{Q}^i } [g^i]-\tilde m^i =E_{\widehat{Q}^i } [\tilde Y^i] = 0 \text{ for all } i.$$
\item The following formulation holds
\begin{align*}
\mcR(g)
&= \sum_{i=1} ^N  \inf \left \{ m \in \mathbb R \mid   \exists k^i \in K_{i}, Y^i \in L^1(\Omega, \mathcal{F}_T^i, \widehat{Q}^i ) \text { with } E_{\widehat{Q}^i } [Y^i]=0  \text{ s.t. } m + k^i + Y^i \geq g^i \right \} \\
& =\sum_{i=1} ^N  \rho_{\widehat{Q}^i }(g^i)=\rho_{A_{\widehat Q}}(g).
\end{align*}

\item 
Any optimizer $(\tilde m, \tilde Y ) \in \mathbb R ^N \times \mathcal{Y} $ of $\mcR(g)$ in Item 2  is an optimizer of $\rho_{A_{\widehat Q}}(g) $.
In particular, $\tilde m^i$ is an optimizer for $\rho_{\widehat{Q}^i }(g^i)$ and $\tilde m^i=\rho_{\widehat{Q}^i }(g^i)=E_{\widehat{Q}^i } [g^i]$.

\end{enumerate} 
\end{proposition}

\begin{proof}

\emph{Item 1. } From the assumption we have
\begin{equation*}
\mcR(g)={\sum_{i=1}^N  {E_{\widehat{Q}^i } [g^i]}}=\sum_{i=1} ^N \widehat{m}^i
\end{equation*} 
 with $\widehat {m}^i + \widehat{k}^i + \widehat{Y}^i \geq g^i $, 
 $(\widehat m, \widehat{k}^i, \widehat Y)  \in \mathbb R ^N \times K_{i} \times \mathcal{Y} \cap L^{1 }(\Omega, \mathbf{F}_T,\mathbf{\widehat{Q}})$. Thus $\widehat{k}^i \geq g^i - \widehat{Y}^i -\widehat {m}^i  \in L^1{(\Omega, \mathcal{F}_T^i, \widehat{Q}^i })$, which then implies 
 By Lemma  \ref{Ek=0} 
 that $\widehat{k}^i \in L^1{(\Omega, \mathcal{F}_T^i, \widehat{Q}^i })$ and $E_{\widehat{Q}^i } [\widehat{k}^i]=0$. By taking the expectation, we then get $$ E_{\widehat{Q}^i } [g^i]-\widehat{m}^i \leq {E_{\widehat{Q}^i } [\widehat{Y}^i]}+{E_{\widehat{Q}^i } [\widehat{k}^i]}={E_{\widehat{Q}^i } [\widehat{Y}^i]} \text { for all } i $$
so that
\begin{equation*}
0=\mcR(g)-\mcR(g)= {\sum_{i=1}^N  {E_{\widehat{Q}^i } [g^i]}}-\sum_{i=1} ^N \widehat{m}^i \leq {\sum_{i=1}^N  {E_{\widehat{Q}^i } [\widehat{Y}^i]}}\leq 0,
\end{equation*}
where the last inequality follows from $\widehat Q \in M^{\mathcal{Y}} $.

\emph{Item 2. } 
Observe that for any $x \in \mathbb R_0 ^N$, $(\widehat m, \widehat Y ) \in \mathbb R ^N \times \mathcal{Y} \cap L^{1 }(\Omega, \mathbf{F}_T,\mathbf{\widehat{Q}})  $ 
is an optimizer for $\mcR(g)$ if and only if ($\tilde m=\widehat m+x, \tilde Y=\widehat Y-x ) \in \mathbb R ^N \times \mathcal{Y} \cap L^{1 }(\Omega, \mathbf{F}_T,\mathbf{\widehat{Q}}) $
is an optimizer for $\mcR(g)$, and obviously  $\sum_{i=1} ^N \widehat{m}^i =\sum_{i=1} ^N \tilde m^i $. 
 Set, for each $i$,  $x^i:={E_{\widehat{Q}^i } [\widehat{Y}^i]}$ so that, from \eqref{Y0}, $x \in   \mathbb R_0 ^N$. Then $E_{\widehat{Q}^i } [\tilde Y^i] = 0$ 
 for all $i$ and ($\tilde m, \tilde Y ) \in \mathbb R ^N \times \mathcal{Y} \cap L^{1 }(\Omega, \mathbf{F}_T,\mathbf{\widehat{Q}}) $ 
 is an optimizer for $\mcR(g)$ and thus satisfies \eqref{AAB}, namely
$E_{\widehat{Q}^i } [g^i]-\tilde m^i =E_{\widehat{Q}^i } [\tilde Y^i] = 0$ for all  $i $.\\

\emph{Item 3. }  Observe that
\begin{align}
\rho_{A_Q}(g)
&=  \inf \left \{ \sum_{i=1} ^N m^i  \mid   \exists k \in {\sf X}_{i=1} ^{N} K_i, Y \in A_Q  \text{ s.t. } m^i  + k^i + Y^i \geq g^i \text { } \forall i \right \} \\
&= \sum_{i=1} ^N  \inf \left \{ m^i \in \mathbb R \mid   \exists k^i \in K_{i}, Y^i \in A_{Q^i}  \text{ s.t. } m^i + k^i + Y^i \geq g^i \right \} \label{prova} \\
& =\sum_{i=1} ^N  \rho_{{Q}^i }(g^i).
\end{align}
Thus  we need to show $\rho_{A_{\widehat{Q}}}(g)=\mcR(g)$.
From Item 2, we know that among the optimizers  for $\mcR(g)$ there exists  ($\tilde m, \tilde Y ) \in \mathbb R ^N \times \mathcal{Y} \cap L^{1 }(\Omega, \mathbf{F}_T,\mathbf{\widehat{Q}})  $ such that
$$(\tilde m, \tilde Y ) \in \left \{ (m,Y ) \in \mathbb R ^N \times  A_{\widehat{Q}}  \mid  \exists k \in {\sf X}_{i=1} ^{N}  K_i \text{ s.t. } m^i  + k^i + Y^i  \geq g^i  \text{  } \forall i  \right \}$$
and thus 
\begin{equation} \label{rhoA}
\rho_{A_{\widehat{Q}}}(g) \leq \mcR(g).
\end{equation}
Now consider  $\rho_{A_{\widehat{Q}}}(g)$ in \eqref{prova}. From the inequality $ m^i + k^i + Y^i \geq g^i$ for all $i$ in \eqref{prova} with  $k^i \in K_{i}, Y^i \in A_{\widehat{Q}^i}$, by using $E_{\widehat{Q}^i} [Y^i]=0$ and the 
Lemma \ref{Ek=0} 
\begin{equation} \label{weakduality11}
\rho_{A_{\widehat{Q}}}(g) \geq  \sum_{i=1} ^N E_{\widehat{Q}^i}[g^i]=\mcR(g).
\end{equation}
From \eqref{rhoA} and \eqref{weakduality11}  we deduce 
\begin{equation*}
\mcR(g) \geq \rho_{A_{\widehat Q}}(g) \geq  \sum_{i=1} ^N E_{\widehat{Q}^i}[g^i]= \mcR(g),
\end{equation*}
which proves Item 3.\\
\emph{Item 4. } 
Moreover, for the optimizer $(\tilde m, \tilde Y ) \in \mathbb R ^N \times \mathcal{Y} \cap L^{1 }(\Omega, \mathbf{F}_T,\mathbf{\widehat{Q}}) $  in Item 2, we obtain $E_{\widehat{Q}^i } [\tilde Y^i] = 0$, for all $i$, and
 \begin{equation*}
 \rho_{A_{\widehat Q}}(g) =  \sum_{i=1} ^N E_{\widehat{Q}^i}[g^i]=\mcR(g)= \sum_{i=1} ^N \tilde m^i.
\end{equation*}
Thus $( \tilde m, \tilde Y)$ is also an optimizer for  $\rho_{A_{\widehat Q}}(g)$, and so $\tilde m^i$ is an optimizer for $\rho_{\widehat{Q}^i }(g^i)$.
\end{proof}

\subsubsection{Economic Interpretation and Comparison with B\"uhlmann's Risk Exchange Equilibrium}\label{comparison}

In this subsection the assumptions of Proposition $\ref{fairness}$ hold true and we adopt the same notations. 
We already observed that cooperation reduces the cost for the system of $N$ agents: $\mcR(g) \leq  \ro (g)$.
Furthermore, from Proposition \ref{fairness} we make the following conclusions.
\begin{enumerate}

\item Given $\mcR(g)$, we define  the \emph{cost allocation for agent $i$} as
\begin{equation*}
 \rho^i(g):=E_{\widehat Q^i } [ g^i ]=\tilde m^i,
\end{equation*}
where $\tilde m^i$ is given in Item 2 of Proposition \ref{fairness}.
From $\mcR(g)= {\sum_{i=1}^N {E_{\widehat{Q}^i } [g^i]}}$ and Item 4 Proposition $\ref{fairness}$, we see that the cost allocation  $\rho^i(g)$ for agent $i$ is equal to the super-replication cost $\rho_{\widehat{Q}^i }(g^i)$, associated to the pricing measure $\widehat{Q}^i$, that agent $i$ would compute independently from all other agents $j \neq i$.
Moreover, as  $\widehat Q \in \mathcal{M}^{\mathcal{Y}} \subseteq  M_{1} \times \dots  \times M_{N}$, we have
\begin{equation*}
\rho^i(g):=E_{\widehat Q^i } [ g^i ] \leq \sup_{ Q^i \in \mi} E_{ Q^i }[ g^i ]:=\pii(g^i),
\end{equation*}
so that \emph{the cost allocation $\rho^i(g)$ for agent $i$  is smaller than the individual super-replication cost $\pii(g^i)$.}

\item 
$\mcR(g)=\sum_{i=1} ^N  \rho_{{\widehat Q}^i }(g^i) $ is the sum of the individual super-replication price of each claim $g^i$, when the price is assigned by $(\widehat Q^1, \dots , \widehat Q^N)$.

\end{enumerate}

We now describe the connection of the notion of collective super-replication  with B\"uhlmann's risk sharing equilibrium \cite{Buhlmann}, which provides an additional rationale to our approach.

 B\"uhlmann's reformulation of a risk sharing equilibrium\textit{ in this context} would consist of  a pair of vectors $(\widehat Q,\widehat Y)$ where $\widehat Q=(\widehat Q^1, \dots , \widehat Q^N) \in M^{\mcY}$ and $ \widehat Y \in L^1(\Omega, \mathcal{F}_T^1, \widehat{Q}^1) \times \dots \times L^1(\Omega, \mathcal{F}_T^N, \widehat{Q}^N )$  such that 
 \begin{enumerate}[label=\alph*)]
     \item $\widehat{Y}^i$ is an optimizer of $$\rho_{\widehat{Q}^i }(g^i)=\inf \left \{ E_{{\widehat Q}^i } [{Y}^i]  \mid Y^i \in L^1{(\Omega, \mathcal{F}_T^i, \widehat Q}^i )   \text{ s.t. } \exists k^i \in K_{i}  \text{ with }  k^i + Y^i \geq g^i \right \}$$ given in \eqref{121},
     \item 
     $\sum_{i=1}^N (\widehat{Y}^i -E_{\widehat Q^i}[\widehat Y^i])=0.$
 \end{enumerate}
 
The economic rationale is clear: in a) each agent is optimizing his own initial cost $\rho_{\widehat{Q}^i }(g^i) $ assigned via $\widehat Q^i$, while the exchanges take place at terminal time and satisfy the clearing condition in b), namely the total amount exchanged $(\sum_{i=1}^N \widehat Y^i)$ must be equal to the initial cost $(\sum_{i=1}^N E_{\widehat Q^i}[\widehat Y^i])$.  \\
Take now
$$\mcY =\left \{ Y \in L^1(\Omega, \mathbf{F}_t, P ) \mid \sum_{i=1}^N Y^i=0 \right \}.$$ 
From Proposition \ref{fairness}, Item 4, we may easily deduce 
\begin{proposition}
    Under the same assumptions of Proposition \ref{fairness}, the pair $(\widehat Q, \widehat Y)$ in Proposition \ref{fairness} satisfies B\"uhlmann-type conditions a) and b).
\end{proposition}

\begin{remark}
    We consider  $N$ probability measures $(\widehat Q^1,...,\widehat Q^N)$, instead of just one as in B\"uhlmann, because of the generality of the set $\mcY$. If $\mcY$ is  the set $$\mcY =\left \{ Y \in (L^1(\Omega, \mathcal{F}_T, P ))^N \mid \sum_{i=1}^N Y^i=0 \text{ and } \mathcal{F}_T^i=\mathcal{F}_T \,\, \forall i \right \},$$ 
    then the probability measures $\widehat Q^1,...,\widehat Q^N$  all coincide on $F_T$, giving one single probability, as in B\"uhlmann.
\end{remark}

\subsection{Collective super-replication when $\mathcal M ^{\mcY}$ is a singleton } 

Recall that under $\mathbf{NA}_i$
\begin{equation*}
\pii(f):=\inf\{m \in \mathbb{R} \mid \exists k \in \ki \text{ s.t. } m+k \geq f \}=\sup_{Q \in\mi} E_{Q}[f] \ \text { for all } f \in L^{1}( \Omega, \mathcal F^i_T,P).
\end{equation*}
We suppose that  $\mathcal{Y}$ fulfills  \eqref{RY}, that $\mathbf{NCA(\mathcal{Y})}$ holds true and that $K^\mathcal Y$ is closed in $L^{0 }(\Omega, \mathbf{F}_T , P)$. Then by Proposition \ref{propositionFTAP3},
$(C^\mathcal{Y})^{0}_1\cap \mathcal{P}^N_e \not = \emptyset$ and $\mathcal{M}^{\mathcal{Y}} \cap \mathcal{P}^N_e \not = \emptyset$.
We further assume that $\mathcal{M}^{\mathcal Y}$ is reduced to one single element $\mathcal{M}^{\mathcal Y}=\{Q_*^1, \dots, Q_*^N\}$, so that $Q_*^i \in M_i \cap \mathcal P_e $ for all $i$ and $(C^\mathcal{Y})^{0}_1=\mathcal{M}^{\mathcal{Y}}$. However, $M_1 \times \dots \times  M_N$ may not be reduced to a singleton and $\pii(f)=\sup_{Q \in M_i} E_{Q}[f]$ is a sublinear functional which in general is not linear.

Thus using \eqref{suprho} and \eqref{RR}, for $g=(g^1,...,g^N) \in L^{1 }(\Omega, \mathbf{F}_T , P)$ we get

\begin{equation*}
\mcR(g)=\sum_{i=1}^N E_{Q_*^i}[ g^i] \leq \sum_{i=1}^N \pii(g^i)=\ro (g),
\end{equation*}
showing also in this setting that cooperation may be advantageous.

We leave a detailed study of  the completeness of the market in the framework of the present paper for future investigation.

\subsection{Collective Sub-replication Price}
 
Similarly to the super-replication case, but not symmetrically, we define

\begin{definition}[Collective Sub-replication]
Let $g=(g^1,...,g^N)  \in L^{0 }(\Omega, \mathbf{F}_T,P)  $.
\begin{align*}
\rho_-^{\mathcal{Y}} (g)&:=\sup \left \{ \sum_{i=1} ^N m^i \mid \exists  k \in {\sf X}_{i=1} ^{N} K_i, m \in \mathbb{R}^N,  Y \in \mathcal{Y} \text{ s.t. } m^i+ k^i - Y^i \leq g^i  \text{  }  \forall i  \right \},\\
{\pi_-^{\mathcal{Y}}}(g)&:=\sup\{m \in \mathbb{R} \mid \exists k \in {\sf X}_{i=1} ^{N} K_i , Y \in \mathcal{Y} \text{ s.t. } m+k^i-Y^i \leq g^i \text{  } \forall i \}.
\end{align*}
\end{definition}
\noindent Observe that, differently from the definition of super-replication, in the definition of the sub-replication we are subtracting the contribution $Y^i$ of the vector $Y$. We clearly have 
$$ {\rho_-^{\mathcal{Y}}}(g)=-\rho_+^{\mathcal{Y}}(-g), \quad {\pi_-^{\mathcal{Y}}}(g)=-\pi_+^{\mathcal{Y}}(-g),$$
\begin{align*} 
{\rho_-^{\mathcal{Y}}}(g)) &\geq \rho^N_-(g):=\sup \left \{\sum_{i=1} ^N m^i  \mid \exists k \in {\sf X}_{i=1} ^{N} K_i , m \in \mathbb{R}^N, \text{ s. t. } m^i+k^i \leq g^i \text{  } \forall i \right \},\\
\pi_-^{\mathcal{Y}}(g)) &\geq \pi^N_-(g):=\sup\{m \in \mathbb{R} \mid \exists k^i \in \ki \text{ such that } m+k^i \leq g^i \text{  } \forall i \},
\end{align*} 
$$ \rho_-^{N}(g)=-\rho_+^{N}(-g)=\sum_{i=1}^N -\pii(-g^i), \quad  \pi_-^{N}(g)=-\pi_+^{N}(-g)=\min_i \{-\pii(-g^i) \},$$
where we used Proposition  \ref{22}.
\noindent The (total) \emph{value of cooperation} with respect to $g$ is then: 

\begin{equation*}
(\ro(g)-{\rho_+^{\mathcal{Y}}}(g)) + ({\rho_-^{\mathcal{Y}}}(g)) - {\ro}_-(g)) \geq 0.
\end{equation*}

\noindent If  $\mathcal{Y}$ fulfills  \eqref{RY}, $\mathbf{NCA(\mathcal{Y})}$ holds true and $K^{\mathcal Y} $ is closed in $L^{0 }(\Omega, \mathbf{F}_T , P)$, we deduce from Theorem \ref{propsupermulti} and from $ {\rho_-^{\mathcal{Y}}}(g)=-\rho_+^{\mathcal{Y}}(-g)$, $ {\pi_-^{\mathcal{Y}}}(g)=-\pi_+^{\mathcal{Y}}(-g)$, that
\begin{equation}\label{supgen}
\rho_-^{\mathcal{Y}}(g)=\inf_{Q \in (C^\mathcal{Y})^{0}_1 }  \sum_{i=1}^N  E_{Q^i } [g^i] \text { }   \text  { and }  \text { }  \pi_-^{\mathcal{Y}}(g)=\inf_{Q \in (C^\mathcal{Y})^{0}_1 } \frac 1 N  \sum_{i=1}^N  E_{Q^i } [g^i] 
\end{equation} 
for any $g \in L^{1 }(\Omega, \mathbf{F}_T,P).  $
\section{ Finite Dimensional Multiperiod Markets}
\label{section:finiteomega}

In this Section, we work in the multiperiod setting outlined in Section \ref{setting}. We assume that $L^{0 }(\Omega, \mathcal{F},P)$ has finite dimension, 
or equivalently that the sigma-algebra $\mathcal F$ is generated by a partition consisting of a finite number of atoms. 

Under this assumption any vector space contained in $(L^{0 }(\Omega, \mathcal{F},P))^N$ is a finite dimensional closed vector space that can be also considered as a finitely generated convex cone. The positive orthant $L_+^{0 }(\Omega, \mathbf{F}_T,P)$ is clearly a finitely generated convex cone. It is also evident that the sum of two (or a finite number of) finitely generated convex cones is again a convex cone that is finitely generated (for completeness in Appendix \ref{app3} we prove these simple facts).
By \cite{Aliprantis} Corollary 5.25,  finitely generated convex cones in a topological metric space are closed. \\
We thus conclude that if $\mathbb K \subseteq (L^{0 }(\Omega, \mathcal{F},P))^N $ is a vector space and $\mcY \subseteq (L^{0 }(\Omega, \mathcal{F},P))^N$ is a finitely generated convex cone (as in the case when $\mcY$ is a vector space), then $$\mathbb K-L_+^{0 }(\Omega, \mathbf{F}_T,P)+\mcY$$ is a (finitely generated) \emph{closed} convex cone in $(L^{0 }(\Omega, \mathcal{F},P))^N$.
Recall the definitions of
\begin{equation*}
\mathbb R_0 ^N:=\left \{   x \in \mathbb R^N \mid \sum_{i=1} ^N x^i =0  \right \}
\end{equation*}
and 
\begin{equation*}
 \mcR(g):=\inf \left \{ \sum_{i=1} ^N m^i \mid \exists  k \in {\sf X}_{i=1} ^{N} K_i, m \in \mathbb{R}^N,  Y \in \mathcal{Y} \text{ s.t. } m^i+ k^i + Y^i \geq g^i  \text{  }  \forall i  \right \}.
 \end{equation*}
 We stress that any vector space  $\mcY \subseteq L^{0 }(\Omega, \mathbf{F}_T,P)$  is a finitely generated convex cone.  
\begin{theorem}\label{th51}
    Suppose that $\mcY \subseteq L^{0 }(\Omega, \mathbf{F}_T,P)$ is a finitely generated convex cone containing $\mathbb R_0 ^N$. Then
     \begin{equation}\label{CFTAPfinite}
     \mathbf{NCA(\mathcal{Y})} \text { iff } \mathcal{M}^{\mathcal{Y}} \cap \mathcal{P}^N_e \not = \emptyset.
     \end{equation}
    Assume that $\mathbf{NCA(\mathcal{Y})}$ holds true.
Then the super-replication  price $\mcR$ is a cash additive, monotone increasing, continuous functional  $\mcR : L^{0 }(\Omega, \mathbf{F}_T,P)  \rightarrow \mathbb R$ and 
        \begin{enumerate}
            \item The following pricing-hedging duality holds $$ \mcR(g)=\max_{Q \in \mathcal{M}^{\mathcal{Y}} }  \sum_{i=1}^N  E_{Q^i } [g^i]=\sum_{i=1}^N  E_{\widehat Q^i } [g^i], \text{ with } \widehat Q \in\mathcal{M}^{\mathcal{Y}};$$ 
            \item There exists an optimizer $(\widehat m, \widehat k, \widehat Y ) \in \mathbb R ^N \times {\sf X}_{i=1} ^{N} K_i \times  \mathcal{Y} $ of $\mcR(g)$ which satisfies
            $$E_{\widehat{Q}^i } [g^i]=\widehat m^i \quad  \text { and } \quad  E_{\widehat{Q}^i } [\widehat Y^i] = 0 \quad \text{ for all } i.$$
            \item $$\mcR(g) =\sum_{i=1} ^N  \rho_{\widehat{Q}^i }(g^i)=\sum_{i=1} ^N \widehat m^i, $$ where, for each $i$, $\widehat m^i$ is an optimizer of $$\rho_{\widehat{Q}^i }(g^i)=\inf \left \{ m \in \mathbb R \mid   \exists k^i \in K_{i}, Y^i \in L^1(\Omega, \mathcal{F}_T^i, \widehat{Q}^i ) \text { with } E_{\widehat{Q}^i } [Y^i]=0  \text{ s.t. } m + k^i + Y^i \geq g^i \right \}.  $$
        \end{enumerate}
\end{theorem}
\begin{proof}
  The assumption that $L^{0 }(\Omega, \mathcal{F},P)$ has finite dimension and the discussion preceding this Theorem, imply that $K^\mathcal{Y}={\sf X}_{i=1} ^{N} K_i -L^{0 }_+(\Omega, \mathbf{F}_T , P) + \mathcal Y$ is a closed convex cone. Moreover, the assumption that $\mcY \subseteq(L^{1 }(\Omega, \mathcal{F},P))^N$  is always fulfilled as $\mcY \subseteq (L^{0 }(\Omega, \mathcal{F},P))^N=(L^{1 }(\Omega, \mathcal{F},P))^N$. Thus we may apply Theorem \eqref{FTAP3} that proves the Collective FTAP in $\eqref{CFTAPfinite}$.
  Lemma \ref{PropertiesRho} provides the properties of $\mcR$ and in particular shows that it is finite valued (Lemma \ref{PropertiesRho} Item 4). From \eqref{suprho} we deduce the pricing hedging duality in which, by Lemma \ref{lemmaYL1}, $(C^\mathcal{Y})^{0}_1 $ can be replaced by $\mathcal{M}^{\mathcal{Y}}$. The existence of the dual maximizer $\widehat Q \in\mathcal{M}^{\mathcal{Y}}$ in Item 1 is a consequence of the compactness of $\mathcal{M}^{\mathcal{Y}}$ due to the finite dimensional setup, which proves Item 1. The remaining Items follow directly from Proposition \ref{fairness}.
\end{proof}

\section{Collective FTAP and the Closure of the Set $K^{\mathcal Y}$ in Multi-period Markets }
\label{sec:Kyclosed}

In multiperiod markets, we will provide sufficient conditions for: (1) the equivalence of  $\mathbf{NCA(\mathcal{Y})}$  and $\mathbf{NA}_{i} $ for all $i$ (Section \ref{61}); (2) the equivalence of $\mathbf{NCA(\mathcal{Y})}$  and $\mathbf{NA}$ in the global market (Section \ref{62}). In  Section \ref{63} we instead show a general version of the Collective FTAP when $\mathbf{NCA(\mathcal{Y})}$ is not (in general) equivalent to any one of the classical notion of No Arbitrage.\\
In the sequel, for technical reasons, we assume that the set $\mcY$ is closed in probability. Observe that the sets in \eqref{Y00} and in \eqref{ExG} as well as $\mcY_0 \cap L^{0 }(\Omega, \mathbf{F}_t,P)$  satisfy such requirement, respectively.

\subsection{ When $\mathbf{NCA(\mathcal{Y})}$ is Equivalent to $\mathbf{NA}_{i} $ for all $i$ } \label{61}
We adopt the same multi-period setting used in Section $\ref{setting}$. 
We point out that the main result in this section is Theorem \ref{TH111}, which collects five equivalent formulations of $\mathbf{NCA(\mathcal{Y})} $.

\begin{theorem}\label{THOne}
If $\mathcal Y \subseteq L^{0 }(\Omega, \mathbf{F}_0,P)$, $0 \in \mathcal Y$ and $\mathcal{Y} \subseteq \mathcal Y_0$ then 
\begin{equation}\label{equiv}
\mathbf{NCA(\mathcal{Y})}  \Leftrightarrow \mathbf{NA}_{i} \text{ } \forall i \Leftrightarrow \mi \cap \mathcal{P}_e \not = \emptyset \text{ } \forall i,
\end{equation}
and the first equivalence holds true even if, for each $i$, the initial sigma algebra $\mathcal {F}^i_0 $ is not trivial. 
\end{theorem}

\begin{proof}
The second equivalence in \eqref{equiv} is the classical Dalang-Morton-Willinger FTAP (Theorem $\ref{DMW}$).  For the first equivalence, observe that 
$\mathbf{NCA(\mathcal{Y})} \Rightarrow \mathbf{NA}_{i} \text{ } \forall i$ was shown in \eqref{NCANA}.
To prove $\mathbf{NA}_{i} \text{ } \forall i \Rightarrow \mathbf{NCA(\mathcal{Y})}$ we show that a Collective Arbitrage implies the existence of a classical arbitrage for some agent. Suppose that  $({\sf X}_{i=1} ^{N} K_i+\mathcal Y) \cap  L^{0 }_+(\Omega, \mathbf{F}_T,P) \not= \{0\}$, and let $k \in {\sf X}_{i=1} ^{N} K_i$ and $Y \in \mathcal Y$ satisfy $k^i+Y^i \geq 0$ for all $i$ and $P(k^j+Y^j > 0)>0$ for some $j$. Then there are two alternatives:
\begin{enumerate}
\item Either $Y^i=0 $ for all $i$, in which case we have $k^i \geq 0$ for all $i$ and $P(k^j > 0)>0$, namely $k^j$ is a classical arbitrage for agent $j$;
\item Or, as a consequence of $\sum_{i=1}^N Y^i=0$, there exists a $n \in \{1, \dots , N \}$ such that $P(Y^n<0)>0$.
In this case, take $A:=\{Y^n<0\} \in \mathcal {F}^n_0 $. Thus, $(k^n+Y^n)1_A \geq 0$, so that $1_A k^n \geq -1_A Y^n \geq 0$ and $P(1_A k^n >0)>0$.
But then $1_A k^n=\sum_{h \in (n) } (1_A H^h\cdot X^h)_T   \in K_n $, since $1_A H^h$ is predictable, due to $A \in \mathcal {F}^n_0 $, so that $1_A k^n$ in an arbitrage for agent $n$.
\end{enumerate}
\end{proof}
\begin{remark}\label{remforthm63}
Differently from the rest of the paper, in Theorems \ref{TH222} and \ref{thm:closedness01} we will not assume that $\mathbf{F}_0$ is trivial, as these results will be applied in the multiperiod setting of Theorem \ref{thm:EASYclosedness01} for the time step $(t,t+1)$, where $\mathbf{F}_t$ is not trivial.  One may recognize from its proof that Theorem \ref{TH222} holds true also if we replace the time $0$ by any time $t \in \{0,...,T-1\}$ and the condition $\mathcal{Y} \subseteq L^{0 }(\Omega, \mathbf{F}_0,P)$ with $\mathcal{Y} \subseteq L^{0 }(\Omega, \mathbf{F}_t,P)$.  Of course, in this case all relevant notions need to be defined and computed from  time $t$ to $T$.
\end{remark}

\begin{theorem}\label{TH222}
We do not assume that the initial sigma algebra $\mathcal {F}^i_0 $ is trivial, for any $i$. Suppose that $\mathcal{Y} \subseteq L^{0 }(\Omega, \mathbf{F}_0,P)$ is a convex cone closed in probability and that $\mathcal{Y} \subseteq \mathcal Y_0$.
Then  $\mathbf{NCA(\mathcal{Y})}$ implies that $({\sf X}_{i=1} ^{N} K_i+\mathcal Y)$ and $K^{\mathcal Y} $ are closed in probability.
\end{theorem}

\begin{proof}[Proof of Theorem \ref{TH222}]
Step 1: $K^{\mathcal Y} $ is closed in probability.\\
Take a sequence $W_n=k_n+Y_n-\ell_n\in K^{\mathcal Y}$.
$\mathbf{NCA}(\mcY)$ implies $\mathbf{NA}_i$ for every $i=1,\dots, N$ as in \eqref{NCANA}, and by Lemma \ref{lemmaugly} applied componentwise we infer $\inf_nY^i_n>-\infty$ a.s. for every $i=1,\dots,N$. Since we are assuming $\mcY\subseteq \mcY_0$, this in turns implies $\sup_n\sum_{j=1}^N\abs{Y^j_n}<+\infty$ a.s.
An iteration of  \cite{KaratzasSchachermayer2022}  Theorem 3.1 yields a subsequence $(Y_{n_k})_k \subseteq \mathcal{Y}$ and a $Y_\infty\in L^{0}(\Omega, \mathbf{F}_0,P) $ such that $\lim_H\frac1H\sum_{h=1}^H Y_{n_h}=Y_\infty$ a.s. (we have $Y_\infty\in L^{0}(\Omega, \mathbf{F}_0,P)$ by a.s. boundedness of the original sequence $(Y_n)_n$). Since $\mcY$ is a convex cone which is closed in probability, $Y_\infty\in\mcY$.
Now, we see that $\frac1H\sum_{h=1}^H W_{n_h}\rightarrow_HW_\infty$, and
$$\frac1H\sum_{h=1}^HW_{n_h}-\frac1H\sum_{h=1}^HY_{n_h}=\frac1H\sum_{h=1}^Hk_{n_h}-\frac1H\sum_{h=1}^H\ell_{n_h}\in {\sf X}_{i=1} ^{N} K_i-L^{0 }(\Omega, \mathbf{F}_T,P).$$
Since the LHS converges in probability, so does the RHS. Moreover, as argued above $\mathbf{NA}_i$ holds for every $i=1,\dots, N$, so that ${\sf X}_{i=1} ^{N} K_i-L^{0 }_+(\Omega, \mathbf{F}_T,P)$ is the product of sets closed in probability, and therefore closed in probability itself. We conclude that $W_\infty-Y_\infty\in {\sf X}_{i=1} ^{N} K_i-L^{0 }_+(\Omega, \mathbf{F}_T,P)$, which implies $W_\infty \in K^\mcY$ since $Y_\infty\in\mcY$ as previously discussed. This shows that $K^{\mathcal Y} $ is closed in probability. \\
Step 2: $({\sf X}_{i=1} ^{N} K_i+\mathcal Y)$ is closed in probability.\\
 Observe that $({\sf X}_{i=1} ^{N} K_i+\mathcal Y)$ can be obtained from $K^{\mathcal Y} $ by replacing in $K^{\mathcal Y} $ the cone $L^{0 }_+(\Omega, \mathbf{F}_0,P)$ with $\{0\}$. Now recall that the set ${\sf X}_{i=1} ^{N} K_i$ is closed in probability (even without assuming $\mathbf{NA}_i$ for every $i=1,\dots, N$). Thus, going throughout the previous proof in Step 1 and replacing there the elements of $L^{0 }_+(\Omega, \mathbf{F}_0,P)$ with $\{0\}$, one can similarly conclude that $({\sf X}_{i=1} ^{N} K_i+\mathcal Y)$ is closed in probability.
\end{proof}

\begin{lemma}
\label{lemmaugly}
We do not assume that the initial sigma algebra $\mathcal {F}_0 $ is trivial. Suppose that $K$ satisfies $\mathbf{NA}$ (see Definition \ref{NAclassic}). Suppose $(k_n)_n\subseteq K$, $(Z_n)_n\subseteq  L^0(\Omega,\mathcal{F}_0,P)$, $(W_n)_n\subseteq  L^0(\Omega,\mathcal{F}_T,P)$ are given sequences such that $k_n+Z_n\geq W_n$ a.s. and $\lim_nW_n= W_\infty$ a.s. for some $W_\infty\in L^0(\Omega,\mathcal{F}_T,P)$. Then $\inf_{n}Z_n>-\infty$ a.s..    
\end{lemma}
\begin{proof}
Define $\ell_n:=k_n+Z_n-W_n\in L_+^0(\Omega,\mcF_T,P)$. Then $k_n+Z_n-\ell_n=W_n$.
Suppose by contradiction that \[A:=\{\inf_nZ_n=-\infty\}\in\mcF_0\,\,\text{ satisfies }\,\,P(A)>0.\]
    Observe that $\liminf_nZ_n=-\infty$ on $A$. As a consequence, by \cite{DS2006} Proposition 6.3.4 Item (i) there exists a sequence $(\tau_n)_n$ of $\mcF_0$-measurable random variables taking values in $\N$ such that $\lim_nZ_{\tau_n}1_A= -\infty1_A$ a.s.
    Define $$H_n:=\frac{1}{1+\abs{Z_n}}1_A\in L^0(\Omega,\mcF_0,P).$$ Then $\lim_nH_{\tau_n}Z_{\tau_n}= -1_A$.
   Moreover, $\lim_nW_{\tau_n}= W_\infty$ a.s., so that $\lim_nH_{\tau_n}W_{\tau_n}= 0 $ a.s. Observe now that $H_nk_n\in K$ for every $n$ (since $H_n$ is $\mcF_0$-measurable), and that \[H_{\tau_n}k_{\tau_n}=\lim_{M}\sum_{m=0}^M1_{\{\tau_n=m\}}H_m k_m.\]
   In the RHS, each term of the summation still belongs to $K$ since $1_{\{\tau_n=m\}}\in L^0(\Omega,\mcF_0,P)$, and so does the limit since $K$ is closed by \cite{DS2006} Proposition 6.8.1. Hence, also $H_{\tau_n}k_{\tau_n}\in K$ for every $n$, and one verifies similarly that $H_{\tau_n}\ell_{\tau_n}\in L^0_+(\Omega,\mcF_T,P)$. 
   Now we have \[H_{\tau_n}k_{\tau_n}-H_{\tau_n}\ell_{\tau_n}= H_{\tau_n}W_{\tau_n}-H_{\tau_n}Z_{\tau_n}\rightarrow_n -(-1_A)=1_A.\]
 which implies $ 1_A\in K-L^0_+(\Omega,\mcF_T,P)$, since $K-L^0_+(\Omega,\mcF_T,P)$ is closed by $\mathbf{NA}$ (see \cite{DS2006} Theorem 6.9.2).  This contradicts $\mathbf{NA}$. Hence, $P(A)=0$.
   
\end{proof}
\begin{theorem}\label{TH111}
Suppose that  $\mathcal{Y} \subseteq L^{0 }(\Omega, \mathbf{F}_0,P)$ fulfills  \eqref{RY} and   that $\mathcal{Y} \subseteq \mathcal Y_0$, so that $\mcY=\R^N_0$ and in particular
${\mathcal Y} $ is closed in probability. Then the following conditions are equivalent
\begin{enumerate}
\item   $\mathbf{NCA(\mathcal{Y})},  $
\item $\mathcal{M}^{\mathcal{Y}} \cap \mathcal{P}^N_e \not = \emptyset,$
\item $\mi \cap \mathcal{P}_e \not = \emptyset \text{ } \forall i, $
\item For each $i$ there exists an equivalent martingale measure for all processes $X^j$ for all $j \in (i),$ 
\item  $\mathbf{NA}_{i} \text{ } \forall i.$
\end{enumerate}
\end{theorem}

\begin{proof}
Under $\mathbf{NCA(\mathcal{Y})} $, by Theorem \ref{TH222} we know that $K^{\mathcal Y} $ is closed in probability and so we can apply Proposition \ref{propositionFTAP3} and conclude that  $\mathbf{NCA(\mathcal{Y})} $ implies $\mathcal{M}^{\mathcal{Y}} \cap \mathcal{P}^N_e \not = \emptyset$. The implications (2) $\Rightarrow$  (3) and (3) $\Rightarrow$  (4) are trivial and (4)  $\Rightarrow$   (5) follows from the easy part of the classical Dalang-Morton-Willinger FTAP (Theorem $\ref{DMW}$), which is true even if $\frac {dQ^i} {dP}$ is not bounded. Finally, the implication $\mathbf{NA}_{i} \text{ } \forall i \Rightarrow \mathbf{NCA(\mathcal{Y})}, $  was proved in Theorem $\ref{THOne} $. 
\end{proof}

In the examples in Section \ref{toyex} Item \ref{72A} and Section \ref{73} Item \ref{73A} the equivalence between $\mathbf{NCA(\mathcal{Y})}$ and $\mathbf{NA}_{i} \text{ } \forall i$   holds true, while in the Examples in Section \ref{toyex} Item \ref{72B} and Section \ref{73} Item \ref{73B} we show that, when $\mathcal{Y} $ is not contained in $L^{0 }(\Omega, \mathbf{F}_0,P)$, it is possible that $\mathbf{NA}_{i} \text{ } \forall i$  hold true but there exists a $\mathbf{CA}$.

 \subsection{ When $\mathbf{NCA(\mathcal{Y})} $ is Equivalent to No Global Arbitrage $\mathbf {NA}$ } \label{62}

We adopt the same multi-period setting used in Section \ref{setting}. 
The key assumption throughout this subsection is a common filtration for all agents: $\mathbb F^i=\mathbb F^i$ for all $i,j$. Thus $L^{0 }(\Omega, \mathbf{F}_T,P)= (L^0 (\Omega, \mathcal{F}, P))^N$. This fact and the selection $\mathcal{Y}=\mathcal{Y}_0$, namely we allow to exchange \emph{all} $\mathcal F_T$-measurable random variables such that the sum of their components is equal to $0$, imply that $\mathbf{NCA(\mathcal{Y})} $ is equivalent to No Global Arbitrage $\mathbf {NA}$. The intuition for this is that $\mcY $ is so large that it contains exchanges of any stochastic integrals so an agent trying to achieve a collective arbitrage from an arbitrage can add to her portfolio the stochastic integrals involving stocks not in her markets by exchanges with the other agents at no cost.
The equivalence between \textbf{NCA} and \textbf{NA} can also be reformulated by saying that a Collective Arbitrage for the $N$ agents is equivalent to a (classical) Arbitrage for a representative agent who is allowed to invest in the global market.
 We stress however that such equivalence holds only under the assumptions of Proposition \ref{propNCANA}. The fact that in general one can not use a representative agent formulation (see the toy Example \ref{extoy} ) is an indication that our results do not trivialize to the classical case.

\begin{proposition}\label{propNCANA}
Suppose that $\mathbb F^i=\mathbb F^j=\mathbb F$ for every $i,j=1,\dots,N$.  

Then for $\mathcal{Y}=\mathcal{Y}_0$ defined in \eqref{Y00} we have:
\begin{enumerate}
\item $\mathbf{NCA}(\mathcal{Y}_0)$ and $\mathbf{NA}$ are equivalent (see Definition \ref{NAclassic} for the latter).
\item Under either  $\mathbf{NCA}(\mathcal{Y}_0)$ or $\mathbf{NA}$, $K^{\mathcal{Y}_0}$ is closed in probability.
\end{enumerate}
\end{proposition}
\begin{proof}
$\,$
1.
In the current setup, $\mathcal{Y}_0 \subseteq (L^0(\Omega,\mathcal{F},P))^N$ and we have $\mathbf{NA} \Rightarrow \mathbf{NCA(\mathcal{Y}_0)}$ by \eqref{NANCAY}, so we only need to prove the converse. To this end, by Lemma \ref{auxNCAimplNA} we have
\begin{equation}
\label{formulakyproduct}
K^{\mcY_0}={\sf X}_{i=1}^N K+\mathcal{Y}_0-(L^0_+(\Omega,\mathcal{F},P))^N.
\end{equation}
  In particular by definition of $\textbf{NCA}(\mathcal{Y}_0)$ and picking $Y=0\in\mathcal{Y}_0$ we get $(K-L^0_+(\Omega,\mathcal{F},P))\cap L^0_+(\Omega,\mathcal{F},P)=\{0\}$ that is $\mathbf{NA}$.

  2.
Since \textbf{NA} holds,  $K-L_+^0(\Omega,\mathcal{F},P)$ is closed in probability (see e.g. \cite{DS2006} Theorem 6.9.2, noticing that we are assuming $\mathcal{F}=\mathcal{F}_T$).
Let now $\somma:(L^0(\Omega,\mathcal{F},P))^N \rightarrow (L^0(\Omega,\mathcal{F},P))$ be the map, continuous for convergence in probability, defined by $\somma(W)=\sum_{i=1}^NW^i$.
It is then enough to show that $$K^{\mcY_0}=\somma^{-1}(K-L_+^0(\Omega,\mathcal{F},P)).$$

   The inclusion $K^{\mcY_0}\subseteq\somma^{-1}(K-L_+^0(\Omega,\mathcal{F},P))$ is directly checked: an element in $K^{\mcY_0}$ is of the form $k+Y-\ell$. Then $\somma(k+Y-\ell)=\sum_{i=1}^Nk^i+0-\sum_{i=1}^N\ell^i\in K-L_+^0(\Omega,\mathcal{F},P)$. It remains to prove the opposite inclusion.   
   Select $Z\in\somma^{-1}(K-L_+^0(\Omega,\mathcal{F},P))$. Consequently, $\sum_{j=1}^NZ^j=k-l$ for some $k\in K$ and $l\in L_+^0(\Omega,\mathcal{F},P)$. By Lemma \ref{lemma:exaust}  $\somma({\sf X}_{i=1}^NK_i)=K$  and
we write $$k=\somma (\widehat{k})\text{ for some }\widehat{k}\in {\sf X}_{i=1}^NK_i.$$
Consider now $W=Z-\widehat{k}$ and observe that $\somma(W)=k-l-k=-l\in -L_+^0(\Omega,\mathcal{F},P)$. We can rewrite 
\begin{align*}
\begin{bmatrix}
           W^{1} \\
           \vdots\\
           W^{N} 
         \end{bmatrix}&=
\begin{bmatrix}
         W^{1}-\somma(W) \\
         W^2\\
           \vdots\\
           W^{N} 
         \end{bmatrix}+
         \begin{bmatrix}
         \somma(W)\\
         0\\
           \vdots\\
           0 
         \end{bmatrix}.
\end{align*}
Observe that the first term on the RHS, call it $\widehat{Y}$, satisfies $\widehat{Y}\in\mcY_0$ and the second term, call it $-\widehat{l}$, belongs to $-(L_+^0(\Omega,\mathcal{F},P))^N$ since $\somma(W)=-l$. Thus, $Z=W+\widehat k =\widehat{k}-\widehat{l}+\widehat{Y}\in K^{\mcY_0}$ which concludes the proof.
\end{proof}

\begin{example}\label{extoy}[On the relevance of the condition $\mcY_0 = \mcY$ for the equivalence of $\mathbf{NCA(\mathcal{Y})}$ and   $\mathbf{NA}$]\label{Ex621}
Take $N=J=3$, $(i)=\{i\}$ for all $i=1, 2 , 3$, a finite $\Omega$ and one single filtration for all assets: $\mathcal{F}^i_t=\mathcal{F}^j_t=\mathcal{F}_t$ for all $t$ and all $i,j \in \{1,2,3 \}$. If  $\mathcal{Y}=\mathcal{Y}_0=\{Y \in (L^{0 }(\Omega, \mathcal{F}_T,P))^3 \mid \sum_{i=1}^3 Y^i = 0\}$, we obtain from Theorem \ref{propNCANA} that $\mathbf{NCA(\mathcal{Y})}$ is equivalent to:
\begin{equation*}
  K^{\mathcal{Y}} \cap  (L^{0 }_+(\Omega, \mathcal{F}_T,P))^3=\{0\} \Leftrightarrow \text{  }  K^3 \cap  (L^{0 }_+(\Omega, \mathcal{F}_T,P))^3=\{0\} \Leftrightarrow \left (\bigcap_{i=i} ^3  M_i \right ) \cap \mathcal{P}_e \not = \emptyset.
\end{equation*}
 
   Now take a different $\mathcal{Y}:=\{Y \in (L^{0 }(\Omega, \mathcal{F}_T,P))^3 \mid Y^1 =0 \text { and } Y^2 + Y^3 = 0\} \not= \mathcal{Y}_0$.
One can easily verify that
\begin{equation*}
\mathcal{M}^{\mathcal{Y}}=\{(Q^1, Q, Q ) \mid Q^1 \in M_1 \text { and } Q \in M_2 \cap M_3 \}.
\end{equation*}
We may select the asset processes $\{X^1, X^2, X^3 \}$ in such a way that: (a) $ \left (\bigcap_{i=i} ^3  M_i \right )\cap \mathcal{P}_e =\emptyset $, so that there exists a global Arbitrage in the market in the classical  sense; (b) $\mathcal{M}^{\mathcal{Y}}  \cap (\mathcal{P}_e)^3 \not = \emptyset$, so that, by Proposition \ref{prop311} and finiteness  $\Omega$, $\mathbf{NCA(\mathcal{Y})}$ holds true.
 The existence of an element in  $\mathcal{M}^{\mathcal{Y}}  \cap (\mathcal{P}_e)^3 $ only guarantees the existence of an equivalent martingale measure for $X_1$ and another one for both processes $\{X_2, X_3 \}$, but not necessarily an equivalent martingale measures for all three processes $\{X^1,X^2, X^3 \}$.
\end{example}

\subsection{ A General Setting with $ \mathcal{Y} \subseteq (L^0(\Omega, \mathcal F, P))^N$ } \label{63} 
We adopt the same multi-period setting used in Section $\ref{setting}$. 
The key assumption throughout this subsection is a common filtration for all agents: $\mathbb{F}^i=\mathbb{F}^j=\mathbb{F}$ for all $i,j\in\{1,\dots,N\}$.

With respect to the setting in the previous two Sections, we do not impose any measurability constraints on $\mcY$ (except for the very general requirement that $\mathcal{Y}\subseteq (L^0(\Omega, \mathcal F, P))^N$) nor any structural geometric properties (as $\mathcal{Y}=\mathcal{Y}_0$). 

As clarified in Section \ref{632}, the multi-period results are based on the one period model. We thus first present the one-period market in Section \ref{631}, which contains most of the technical details relative to the closure of $K^{\mcY}$. 

In some results we will not only need that $\mathcal{Y}$ is closed in probability but also that $\mathcal{Y}-(L^0_+(\Omega,\mathcal{F},P))^N$ is closed in probability.
The following auxiliary Lemma shows that this latter condition is automatic as soon as $\mcY \subseteq \mcY _0$ is closed in probability.

\begin{lemma}\label{lemmaYL}
    Suppose that $\mcY \subseteq \mcY _0$ is a convex cone closed in probability. Then $\mathcal{Y}-(L^0_+(\Omega,\mathcal{F},P))^N$ is closed in probability.
\end{lemma}

\begin{proof}
    Take a sequence $(Y_n-l_n)_n\subseteq \mathcal{Y}-(L^0_+(\Omega,\mathcal{F},P))^N$, converging in probability (w.l.o.g. almost surely) to $Z\in(L^0(\Omega,\mathcal{F},P))^N $. Take $A:=\{\sup_n\abs{Y_n}=+\infty\}$. If $P(A)>0$, then for some index $k$,  $A_k:=\{\inf_n Y_n^k=-\infty\}$ satisfies $P(A_k)>0$. Consequently $Z^k1_{A_k}=\lim_n (Y^k_n1_{A_k}-l_n1_{A_k})\leq \liminf_n Y^k_n1_{A_k}=-\infty1_{A_k}$, which is a contradiction. Thus, $P(A)=0$ and $(Y_n)_n$ is a.s. bounded. So is $(l_n)_n$, since $(Y_n-l_n)$ converges.  
    By an iteration of  \cite{KaratzasSchachermayer2022} Theorem 3.1 we can extract subsequences, $(Y_{n_k})_k,(l_{n_k})_k$, whose Cesaro means converge almost surely, say to $Y$ and $l$ respectively. Since $\mathcal{Y}$ and $(L^0_+(\Omega,\mathcal{F},P))^N$ are convex and closed, the Cesaro means and their limits belong to the respective set: $Y\in\mathcal{Y}, l\in (L^0_+(\Omega,\mathcal{F},P))^N$.  Since $(Y_{n_l}-l_{n_l})_l$ converges almost surely to $Z$, also its Cesaro means converge to $Z$ (see \cite{Williams} Lemma 12.6). Thus, $Z=Y-l\in\mathcal{Y}-(L^0_+(\Omega,\mathcal{F},P))^N$.
\end{proof}

\subsubsection{The One-period Market with General $ \mathcal{Y} \subseteq (L^0(\Omega, \mathcal F, P))^N$} \label{631}

In this section we consider a one period market  $\mathcal{T}=\{0,1\}$ and thus the elements in $K_i$ are just one period stochastic integrals. 

\begin{theorem}
\label{thm:closedness01}
    Suppose that $\mathcal{T}=\{0,1\}$,  $\mathcal{F}^i_0=\mathcal{F}^j_0=\mathcal{F}_0$  for all $i,j\in\{1,\dots,N\}$, and that $\mathcal{F}_0$ is not necessarily trivial. Assume that $\mathcal{Y}\subseteq (L^0(\Omega, \mathcal F, P))^N$ is a vector space satisfying:
    \begin{enumerate}
        \item\label{A1_6.12} $\mathcal{Y}$ and $(\mathcal{Y}-(L^0_+(\Omega,\mathcal{F},P))^N)$ are closed in probability;
        \item\label{A2_6.12} For every $Z\in L^0(\Omega,\mathcal{F}_0,P)$ and $Y\in\mathcal{Y}$, we have $ZY\in\mathcal{Y}$.
        \end{enumerate}
Then $( {\sf X}_{i=1} ^{N}  K_i + \mathcal Y) $ is closed in probability. If $\mathbf{NCA(\mathcal{Y})}$ holds, then $K^\mathcal Y=  {\sf X}_{i=1} ^{N}  K_i -( L^{0 }_+(\Omega, \mathcal{F},P) )^N + \mathcal Y $ is closed in probability.
\end{theorem}

\begin{lemma}
    \label{lemma:schacher}
    Define
    \begin{align}
        \mathcal{H}_\mathcal{Y}&=\left\{(H_1,\dots,H_N)\in {\sf X}_{i=1}^N\mathcal{H}_i\mid ((H_i\cdot X)_T)_{i=1}^N\in\mathcal{Y} \right\}\label{Hy}\\
        \mathcal{H}^\bot_\mathcal{Y}&=\left\{(H_1,\dots,H_N)\in {\sf X}_{i=1}^N\mathcal{H}_i\mid \sum_{i=1}^NH_iH_i^*=0\,\forall (H_1^*,\dots,H_N^*)\in\mathcal{H}_\mathcal{Y} \right\}\,\label{Hyorth}
    \end{align}
i.e.    $\mathcal{H}_\mathcal{Y} $ is the (vector) space of vectors of strategies, each admissible for the corresponding agents, replicating elements in $\mathcal{Y}$, and $\mathcal{H}^\bot_\mathcal{Y}$ is its a.s. orthogonal complement in ${\sf X}_{i=1}^N\mathcal{H}_i$. Note that in \eqref{Hyorth} above $H_i,H_i^*$ are both $d_i$ dimensional, and $H_iH_i^*$ stands for the inner product $\omega$-wise.
   
Under 
the same assumptions
of Theorem \ref{thm:closedness01}, $\mathcal{H}_\mathcal{Y}$ and $\mathcal{H}^\bot_\mathcal{Y}  $ are closed in probability, we have 
$${\sf X}_{i=1}^N\mathcal{H}_i=\mathcal{H}_\mathcal{Y}\oplus\mathcal{H}^\bot_\mathcal{Y}$$ and in particular
\begin{equation}
    \label{nullityonHyorth}
    (H_1,\dots,H_N)\in \mathcal{H}^\bot_\mathcal{Y}\,, ((H_i\cdot X)_T)_{i=1}^N\in\mathcal{Y}\Longrightarrow (H_1,\dots,H_N)=0\,.
\end{equation}
\begin{equation}
    \label{Hyorthscales}
    (H_1,\dots,H_N)\in \mathcal{H}^\bot_\mathcal{Y}\,,\,Z\in L^0(\Omega,\mathcal{F}_0,P)\Longrightarrow Z(H_1,\dots,H_N)\in \mathcal{H}^\bot_\mathcal{Y}\,.
\end{equation}
\end{lemma}

\begin{proof}
    We follow \cite{Schachermayer92} Lemma 2.4. 
    It is elementary to check that $\mathcal{H}_\mathcal{Y}\cap\mathcal{H}^\bot_\mathcal{Y}=\{0\}$, in that if $(H_1,\dots,H_N)$
    belongs to such intersection we have $\sum_{i=1}^NH_iH_i=0$ which immediately gives $H_1=\dots=H_N=0$.
    Let us observe first that $\mathcal{H}_\mathcal{Y}$ is closed in probability: if we take a sequence in $\mathcal{H}_\mathcal{Y}$, say $(H_1(n),\dots,H_N(n))_n$, converging in probability to $(H_1,\dots,H_N)$, up to passing to a subsequence the convergence is a.s. and we have $$Y_n:=((H_1(n)\cdot X)_T,\dots,(H_N(n)\cdot X)_T)\rightarrow_n((H_1\cdot X)_T,\dots,(H_N\cdot X)_T)\quad \mathrm{a.s.}.$$
    Since  $\mathcal{Y}$ is closed in probability $((H_1\cdot X)_T,\dots,(H_N\cdot X)_T)\in\mathcal{Y}$.

We also observe that $\mathcal{L}:={\sf X}_{i=1}^N\mathcal{H}_i\cap {\sf X}_{i=1}^N(L^2(\Omega,\mathcal{F},P))^{d_i}$ is itself a Hilbert space, being a closed subspace of a Hilbert space.
Furthermore, $\mathcal{H}_\mathcal{Y}^2:=\mathcal{H}_\mathcal{Y}\cap \mathcal{L}$ is a closed vector subspace of $\mathcal{L}$ by  closedness in probability of $\mathcal{H}_\mathcal{Y}$. Call $\mathcal{H}_\mathcal{Y}^{2,\bot}$ its orthogonal complement in $\mathcal{L}$ and denote by $\pi,\pi^\bot$ the corresponding orthogonal projections.
We also stress that, whenever $(H_1,\dots,H_N)\in{\sf X}_{i=1}^N\mathcal{H}_i$, for $\abs{(H_1,\dots,H_N)}:=\sum_{i=1}^N\abs{H_i}$ (the absolute value of a vector stands for the sum of the absolute values of its components) we have $\frac1{1+\abs{(H_1,\dots,H_N)}}(H_1,\dots,H_N)\in \mathcal{L}$. This is a consequence of assuming a common initial sigma algebra $\mathcal{F}_0$. We now define for $(H_1,\dots,H_N)\in{\sf X}_{i=1}^N\mathcal{H}_i$

\begin{align}
    \mathrm{Proj}((H_1,\dots,H_N)|\mathcal{H}_\mathcal{Y})&:=(1+\abs{(H_1,\dots,H_N)})\pi\left(\frac1{1+\abs{(H_1,\dots,H_N)}}(H_1,\dots,H_N)\right)\\
    \mathrm{Proj}((H_1,\dots,H_N)|\mathcal{H}^\bot_\mathcal{Y})&:=(1+\abs{(H_1,\dots,H_N)})\pi^\bot\left(\frac1{1+\abs{(H_1,\dots,H_N)}}(H_1,\dots,H_N)\right)
\end{align}
It is immediate to verify that $\mathrm{Proj}((H_1,\dots,H_N)|\mathcal{H}_\mathcal{Y})+\mathrm{Proj}((H_1,\dots,H_N)|\mathcal{H}^\bot_\mathcal{Y})=(H_1,\dots,H_N)$.
Thus, we are only left to the verification that $\mathrm{Proj}((H_1,\dots,H_N)|\mathcal{H}^\bot_\mathcal{Y})\in\mathcal{H}^\bot_\mathcal{Y}$ and the same holds with $\mathcal{H}_\mathcal{Y}$ mutatis mutandis.
We verify the first claim, as the second one follows similarly. We heavily rely again on Assumption \ref{A2_6.12}  in Theorem \ref{thm:closedness01}. We first show that $\mathcal{H}^{2,\bot}_\mathcal{Y}\subseteq \mathcal{H}^{\bot}_\mathcal{Y}$: 
take indeed $(H_1,\dots,H_N)\in \mathcal{H}^{2,\bot}_\mathcal{Y}$ and $ (H_1^*,\dots,H_N^*)\in\mathcal{H}_\mathcal{Y}$. If $A:=\{\sum_{i=1}^N H_iH^*_i>0\}$ had positive probability, observing that, by Assumption \ref{A2_6.12} in Theorem \ref{thm:closedness01},  
$1_A\frac{1}{1+\abs{(H_1^*,\dots,H_N^*)}}(H_1^*,\dots,H_N^*)\in \mathcal{H}_\mathcal{Y}\cap {\sf X}_{i=1}^N(L^2(\Omega,\mathcal{F},P))^{d_i}$, we would get 
$$E_{P}\left[1_A\frac{1}{1+\abs{(H_1^*,\dots,H_N^*)}}\sum_{i=1}^NH_iH_i^*\right]>0$$ 
which is a contradiction. Arguing similarly with $\{\sum_{i=1}^N H_iH^*_i<0\}$ yields $\sum_{i=1}^N H_iH^*_i=0$, and the desired inclusion. To conclude, we observe that  if $(H_1,\dots,H_N)\in \mathcal{H}_\mathcal{Y}$ and $Z\in L^0(\Omega,\mathcal{F}_0,P)$ then $Z(H_1,\dots,H_N)\in \mathcal{H}_\mathcal{Y}$, just by Assumption \ref{A2_6.12} in Theorem \ref{thm:closedness01}, and the same property holds for $\mathcal{H}^\bot_\mathcal{Y}$
as a consequence. Equation \eqref{nullityonHyorth} follows by observing that $(H_1,\dots,H_N)$ would then belong to $\mathcal{H}_\mathcal{Y}\cap\mathcal{H}^\bot_\mathcal{Y}$.
\end{proof}

\begin{proof}[Proof of Theorem \ref{thm:closedness01}]
    We begin by proving that $K^\mathcal Y=  {\sf X}_{i=1} ^{N}  K_i -( L^{0 }_+(\Omega, \mathcal{F},P) )^N + \mathcal Y $ is closed in probability. Observe that the assumption $\mathbf{NCA(\mathcal{Y})}$ is only used at the very last step of the proof. Let $W_n=k_n+Y_n-l_n$ define a sequence in ${K}_\mathcal{Y}$, converging in probability (w.l.o.g, almost surely) to some $W_\infty\in (L^0(\Omega,\mathcal{F},P))^N$.
    Suppose that $k_n^i=(H_i(n)\cdot X)_T$ for $H_i(n)\in\mathcal{H}_i$ as $i=1,\dots, N$. Set $(\hat H_1(n),\dots,\hat H_N(n)):=
    \mathrm{Proj}((H_1,\dots,H_N)|\mathcal{H}^\bot_\mathcal{Y})$. Then $k_n^i=(\hat H_i(n)\cdot X)_T+ \hat Y^i_n=\hat k_n^i+ \hat Y^i_n$ where $\hat  Y_n:=( \hat Y_n^1,\dots, \hat Y_n^N)\in\mathcal{Y}$ by Lemma \ref{lemma:schacher}. Hence $W_n=\hat k_n+(Y_n+\hat{Y}_n)-l_n=\hat{k}_n+\bar Y_n-l_n$ and 
    \begin{equation}\label{WkYl}
    W_n-\hat k_n=\bar Y_n-l_n.
       \end{equation}
    Define now  
    \begin{equation}
        A:=\left\{\liminf_{n}\sum_{i=1}^N\abs{\hat H_i(n)}=+\infty\right\}\in \mathcal{F}_0.
     \end{equation}
     
     Suppose first that $P(A)=0$.  In this case, we can apply  an argument already used in the classical proof of the FTAP.  
     By applying \cite{FollmerSchied2} Lemma 1.64 to the sequence $\{\eta_n\}_n:=(\hat H_1(n),\dots,\hat H_N(n))_n$, one can extract a strictly increasing sequence $(\tau_n)_n$ of $\mathcal{F}_0$-measurable integer valued random variables and a $\eta\in {\sf X}_{i=1}^N(L^0(\Omega,\mathcal{F}_0,P))^{d_i}$ such that $$\lim_n(\hat H_1(\tau_n),\dots,\hat H_N(\tau_n))= \eta  \text { } P-\mathrm{a.s.}.$$
One now verifies that: $(\hat H_1(\tau_n),\dots,\hat H_N(\tau_n))\in\mathcal{H}_\mathcal{Y}^\bot$, $\eta \in \mathcal{H}_\mathcal{Y}^\bot$, $\bar {Y}_{\tau_n}\in\mathcal{Y}$, $l_{\tau_n}\in (L^0_+(\Omega,\mathcal{F},P))^N$. To this end, observe that

\begin{equation}
    \label{decomposetimes1}
    (\hat H_1(\tau_n),\dots,\hat H_N(\tau_n))=\lim_{K}\sum_{k=0}^K1_{\{\tau_n=k\}}(\hat H_1(k),\dots,\hat H_N(k))\quad P-\mathrm{a.s.}\,.
\end{equation}
Since $1_{\{\tau_n=k\}}\in L^0(\Omega,\mathcal{F}_0,P)$ and $(\hat H_1(k),\dots,\hat H_N(k)) \in \mathcal{H}^\bot_\mathcal{Y}$, we use \eqref{Hyorthscales} to get that \\$1_{\{\tau_n=k\}}(\hat H_1(k),\dots,\hat H_N(k))\in \mathcal{H}^\bot_\mathcal{Y}$. Since $\mathcal{H}^\bot_\mathcal{Y}$ is a vector space, also the sums in RHS of \eqref{decomposetimes1} belong to $\mathcal{H}^\bot_\mathcal{Y}$, and since the latter is also closed under almost sure convergence, we also get that the limit over $K$, namely $ (\hat H_1(\tau_n),\dots,\hat H_N(\tau_n))$, belongs to $\mathcal{H}^\bot_\mathcal{Y}$. Since then $\eta$ is the a.s. limit of a sequence in $\mathcal{H}^\bot_\mathcal{Y}$ by definition, it also belongs to $\mathcal{H}^\bot_\mathcal{Y}$. A similar argument yields also $l_{\tau_n}\in (L^0_+(\Omega,\mathcal{F},P))^N$ and $$ \bar{Y}_{\tau_n}=\lim_{K}\sum_{k=0}^K1_{\{\tau_n=k\}} \bar{Y}_{k}\in\mathcal{Y}.$$
From $W_{\tau_n}\rightarrow_nW_\infty$  and using  \eqref{WkYl} we then get 
\begin{align*}
    &W_\infty-((\eta_1\cdot X)_T,\dots, (\eta_N\cdot X)_T)\\
    &=\lim_n \left(W_{\tau_n}-((\hat H_1(\tau_n)\cdot X)_T,\dots,(\hat H_N(\tau_n)\cdot X)_T)\right)=\lim_n \left ( \bar{Y}_{\tau_n}-l_{\tau_n} \right) 
\end{align*}
Then $ \left ( \bar{Y}_{\tau_n}-l_{\tau_n} \right) \in \left ( \mathcal{Y}-(L^0_+(\Omega,\mathcal{F},P))^N \right )$  is almost surely converging and, by Assumption \ref{A1_6.12} in Theorem \ref{thm:closedness01}, its limit will also belong to $\mathcal{Y}-(L^0_+(\Omega,\mathcal{F},P))^N$. This implies that $W_\infty-((\eta_1\cdot X)_T,\dots , (\eta_N\cdot X)_T)=Y_\infty-l_\infty$ for some $Y_\infty\in\mathcal{Y}, l_\infty\in(L^0_+(\Omega,\mathcal{F},P))^N $, which concludes the proof. 

We now prove that $P(A)>0$ is in contradiction with $\mathbf{NCA}(\mathcal{Y})$. Set
         \begin{align}
        A_n&:=A\cap \left\{\sum_{i=1}^N\abs{\hat H_i(n)}>0\right\}\in \mathcal{F}_0,\\
        Z_n&:=\frac{1}{\sum_{i=1}^N\abs{\hat H_i(n)}}1_{A_n}\in L^0(\Omega,\mathcal{F}_0,P).  
    \end{align}
Observe that for every $n$: $Z_n(\hat H_1(n),\dots,\hat H_N(n))\in \mathcal{H}^\bot_\mathcal{Y}$ by \eqref{Hyorthscales},  and $Z_n\bar Y_n\in\mathcal{Y}$ by Assumption \ref{A2_6.12}. Additionally, 
$$\abs{Z_n(\hat H_1(n),\dots,\hat H_N(n))}=\frac{1}{\sum_{i=1}^N\abs{\hat H_i(n)}}\sum_{i=1}^N\abs{\hat H_i(n)}1_{A_n}=1_{A_n}$$ which in particular implies that 
$Z_n(\hat H_1(n),\dots,\hat H_N(n))$ is $P$ a.s. bounded, and that
\begin{equation}
    \label{limhn}
    \lim_{n}\abs{Z_n(\hat H_1(n),\dots,\hat H_N(n))}=\lim_n1_{A_n}=1_A\quad P-\mathrm{a.s.}.
\end{equation}
By applying again \cite{FollmerSchied2} Lemma 1.64 to the sequence $\{\xi_n\}_n:=\{Z_n(\hat H_1(n),\dots,\hat H_N(n))\}_n$, we deduce the existence of a strictly increasing sequence $(\sigma_n)_n$ of $\mathcal{F}_0$-measurable integer valued random variables and of an element $\xi\in {\sf X}_{i=1}^N(L^0(\Omega, \mathcal F_0,P)^{d_i}$ with 
\begin{equation}
    \label{limhn1}
\lim_nZ_{\sigma_n}(\hat H_1(\sigma_n),\dots,\hat H_N(\sigma_n))= \xi\quad P-\mathrm{a.s.}.
\end{equation}
In particular by \eqref{limhn} and since $\sigma_n\uparrow_n\infty$
\begin{equation}
    \label{decomposetimes2}
\abs{\xi}=   \lim_{n}\abs{Z_{\sigma_n}(\hat H_1(\sigma_n),\dots,\hat H_N(\sigma_n))}=\lim_{n}\abs{Z_n(\hat H_1(n),\dots,\hat H_N(n))}=1_A
\end{equation} 
As in  \eqref{decomposetimes1}, we can use the decomposition
\begin{equation}
    \label{decomposetimes}
    Z_{\sigma_n}(\hat H_1(\sigma_n),\dots,\hat H_N(\sigma_n))=\lim_{K}\sum_{k=0}^K1_{\{\sigma_n=k\}}Z_{k}(\hat H_1(k),\dots,\hat H_N(k))\quad P-\mathrm{a.s.}\,,
\end{equation} 
the facts that $1_{\{\sigma_n=k\}}Z_{k}\in L^0(\Omega,\mathcal{F}_0,P)$ and  $(\hat H_1(k),\dots,\hat H_N(k)) \in \mathcal{H}^\bot_\mathcal{Y}$, and \eqref{Hyorthscales} to deduce $1_{\{\sigma_n=k\}}Z_{k}(\hat H_1(k),\dots,\hat H_N(k))\in \mathcal{H}^\bot_\mathcal{Y}$. Arguing as after \eqref{decomposetimes1}, we then deduce that $\xi \in \mathcal{H}^\bot_\mathcal{Y}$ and similarly also $l_{\sigma_n}\in L^0_+(\Omega,\mathcal{F},P)$ and $Z_{\sigma_n} \bar{Y}_{\sigma_n}\in\mathcal{Y}.$
Now  \eqref{limhn1} yields:
$$Z_{\sigma_n} \hat{k}_{\sigma_n}=(Z_{\sigma_n}\hat H_1(\sigma_n)\cdot X)_T,\dots,(Z_{\sigma_n}\hat H_N(\sigma_n)\cdot X)_T)\rightarrow_n((\xi_1\cdot X)_T,\dots (\xi_N\cdot X)_T)\quad P-\mathrm{a.s.}$$
with $\xi\in \mathcal{H}^\bot_\mathcal{Y}$.
Moreover, observe that $Z_{\sigma_n}W_{\sigma_n}\rightarrow_n 0$ and then, using \eqref{WkYl}, 
$$-((\xi_1\cdot X)_T,\dots, (\xi_N\cdot X)_T)=\lim_n(Z_{\sigma_n}W_{\sigma_n}-Z_{\sigma_n}\hat{K}_{\sigma_n})= \lim_n (Z_{\sigma_n}\bar{Y}_{\sigma_n}-Z_{\sigma_n}l_{\sigma_n})\quad P-\mathrm{a.s.}.$$
Since $(Z_{\sigma_n}\bar{Y}_{\sigma_n}-Z_{\sigma_n}l_{\sigma_n}) \in \left ( \mathcal{Y}-(L_+^0(\Omega,\mathcal{F},P))^N \right )$, for each $n$, and $\mathcal{Y}-(L_+^0(\Omega,\mathcal{F},P))^N$ is closed in probability by Assumption \ref{A1_6.12} in Theorem \ref{thm:closedness01}, we conclude that $-((\xi_1\cdot X)_T,\dots ,(\xi_N\cdot X)_T)\in \mathcal{Y}-(L_+^0(\Omega,\mathcal{F},P))^N$. In particular there exist $Y_\infty\in\mathcal{Y},l_\infty\in(L_+^0(\Omega,\mathcal{F},P))^N $ with 
\begin{equation}\label{345}
((\xi_1\cdot X)_T,\dots (\xi_N\cdot X)_T)+Y_\infty=l_\infty \geq 0.    
\end{equation}
Now $\mathbf{NCA}(\mathcal{Y})$ yields 
$((\xi_1\cdot X)_T,\dots (\xi_N\cdot X)_T)+Y_\infty=0$. 
Since $\xi\in \mathcal{H}_\mathcal{Y}^\bot$,  the condition in \eqref{nullityonHyorth} then implies $\xi=0$, which yields a contradiction with \eqref{decomposetimes2}, if $P(A)>0$.

This shows that $K^{\mathcal Y} $ is closed in probability. To show that also $({\sf X}_{i=1} ^{N} K_i+\mathcal Y)$ is closed in probability, we apply the same argument used in the Step 2 of the proof of Theorem \ref{TH222} and in particular the observation that $({\sf X}_{i=1} ^{N} K_i+\mathcal Y)$ can be obtained from $K^{\mathcal Y} $ by replacing in $K^{\mathcal Y} $ the cone $L^{0 }_+(\Omega, \mathbf{F}_0,P)$ with $\{0\}$. 
 Thus, going throughout the above proof in the present theorem and replacing there the elements of $L^{0 }_+(\Omega, \mathbf{F}_0,P)$ with $\{0\}$, one can conclude, similarly to \eqref{345}, that
\begin{equation}
((\xi_1\cdot X)_T,\dots (\xi_N\cdot X)_T)+Y_\infty=l_\infty=0,
\end{equation}
yielding the desired contradiction, without even using $\mathbf{NCA}(\mathcal{Y})$. 
\end{proof}

\subsubsection{ The Multi-period Market with General $ \mathcal{Y} \subseteq (L^0(\Omega, \mathcal F, P))^N$ } \label{632} 

In this section we suppose that $\mathcal{T}=\{0,\dots,T\}$, that $\mathbb{F}^i=\mathbb{F}^j=\mathbb{F}$ for all $i,j\in\{1,\dots,N\}$.
We now extend the notation introduced in Section \ref{setting} : for $0\leq s< S\leq T$ we set
\begin{align}
(H \cdot X)_{s}^S&:=\sum_{h \in (i)} (H^h \cdot X^h)_s^S \quad  \text { for } H\in\mathcal{H}_i,\label{kst4}\\
\label{Kst4}
K_{i,s:S}&:=\left\{(H \cdot X)_{s}^S\mid H\in\mathcal{H}_i\right\}\subseteq K_i, \quad K_{s:S}:=\left\{(H \cdot X)_{s}^S\mid H\in\mathcal{H}\right\}\subseteq K,
\end{align}
where $(H^h \cdot X^h)_s^S:=\sum_{t=s+1}^{S} H^h_{t}(X^h_{t}-X^h_{t-1})$. 
Observe that $K_i=K_{i,0:T}$ and that $(H^h \cdot X^h)_{s}^S$ is the stochastic integral of $H^h$ with respect to $X^h$ on the time set $\{s,s+1,\dots,S\}$.
We also consider, for some fixed $\hatt\in\{1,\dots,T\}$, a vector subspace $\mcY \subseteq L^0(\Omega,\mathcal{F}_{\widehat t},P)^N$, and put
$$K^{\mathcal{Y}}_{s:S}:={\sf X}_{i=1}^NK_{i,s:S}-(L_+^0(\Omega,\mathcal{F},P))^N+\mathcal{Y}, \quad 0\leq s< S\leq T.$$

We consider the following conditions:
\begin{itemize}
\item Classical No Arbitrage Condition for agent $i$ in time period $\{s,\dots,S\}$
\begin{equation}
\label{eq:NAisS4}
 \mathbf{NA}_{i,s:S}: \quad K_{i,s:S} \cap L_{+}^{0}(\Omega, \mathcal{F}, P)=\{0\}.
 \end{equation}
\item Classical No Global  Arbitrage Condition in time period $\{s,\dots,S\}$.
\begin{equation*}
 \mathbf{NA}_{s:S}:\quad K_{s:S} \cap L_{+}^{0}(\Omega, \mcF, P)^N=\{0\}.
 \end{equation*}

 \item No Collective  Arbitrage in time period $\{s,\dots,S\}$.
\begin{equation*}
 \mathbf{NCA}_{s:S}(\mcY):\quad =\left({\sf X}_{i=1}^NK_{i,s:S}+\mathcal{Y}\right)\cap (L_{+}^{0}(\Omega, \mcF, P))^N=\{0\},
 \end{equation*}
 which is equivalent to $K^{\mcY}_{s:S}\cap (L_{+}^{0}(\Omega, \mcF, P))^N=\{0\}$.
\end{itemize}

\noindent In this setting, 
we have that $K^{\mathcal{Y}}_{0:T}=K^{\mathcal{Y}} $  and $\mathbf{NCA(\mathcal{Y})}$, as defined in \eqref{b}, can be rewritten as  $$ K^{\mathcal{Y}}_{0:T} \cap(L_+^0(\Omega,\mathcal{F},P))^N=\{0\}.$$ 
\begin{theorem}
\label{thm:EASYclosedness01}
    Suppose that $\mathcal{T}=\{0,\dots,T\}$ and that $\mathbb{F}^i=\mathbb{F}^j=\mathbb{F}$ for all $i,j\in\{1,\dots,N\}$.
    Fix $\hatt\in\{1,\dots,T\}$, assume that $\mcY \subseteq (L^0(\Omega,\mathcal{F}_{\widehat t},P))^N$ is a vector space and that
    \begin{enumerate}
        \item\label{EA1} $\mathcal{Y}$ is closed in probability and $\mathcal{Y} \subseteq \mathcal{Y}_0$;
        \item\label{EA2} for any $Z\in L^0(\Omega,\mathcal{F}_{\hatt-1},P)$ and $Y\in\mathcal{Y}$ we have $ZY\in\mathcal{Y}$.
        \end{enumerate}
If $\mathbf{NCA(\mathcal{Y})}$ holds, then $K^{\mathcal{Y}}_{0:T}$ and ${\sf X}_{i=1}^NK_{i,0:T}+\mathcal{Y}$
are closed in probability.

Suppose that $\mcY$ additionally satisfies \eqref{RY}.
Then
\begin{equation}
\label{eq:ncaiffoneperiod4}
\mathbf{NCA(\mathcal{Y})}\Longleftrightarrow \mathbf{NCA}_{t-1:t}(\mathcal{Y})\,\,\forall t=1,\dots, T \Longleftrightarrow\begin{cases}
\mathbf{NA}_{i,t:t+1}\,\,\,\forall \,i=1,\dots,N,\,t=\hatt,\dots,T-1,\\
\mathbf{NCA}_{\hatt-1:\hatt}(\mcY),\\
\mathbf{NA}_{t:t+1}\,\,\,\forall \,t=0,\dots,\hatt-2.
\end{cases}
\end{equation}

\end{theorem}

\begin{proof} 
By Theorem \ref{TH222} and Reamrk \ref{remforthm63},
\begin{equation}\label{KKK4}
 K^{\mathcal{Y}}_{\widehat t:T} \text{ \, and  \, } {(\sf X}_{i=1}^NK_{i,\widehat t:T}+\mathcal{Y}) \text { are closed in probability. }    
\end{equation}
 Observe that $Zk \in {\sf X}_{i=1}^NK_{i,\widehat t:T}$ for all $Z\in L^0(\Omega,\mathcal{F}_{\hatt-1},P)$ and all $k \in {\sf X}_{i=1}^NK_{i,\widehat t:T}$. Consider the vector space  $\overline{\mathcal{Y}}:={\sf X}_{i=1}^NK_{i,\widehat t:T} +\mathcal{Y}$ and note that $\overline{\mathcal{Y}}-(L^0_+(\Omega,\mathcal{F},P))^N=K^{\mathcal{Y}}_{\widehat t:T} $. Then by \eqref{KKK4},  $\overline{\mathcal{Y}}$ and $\overline{\mathcal{Y}}-(L^0_+(\Omega,\mathcal{F},P))^N$  are closed in probability and so $\overline{\mathcal{Y}}$ satisfies  Items 1 and 2 in the statement of Theorem \ref{thm:closedness01} for the time period  $\{\widehat t-1,\widehat t \}$. Thus, by  Theorem \ref{thm:closedness01}, ${\sf X}_{i=1}^NK_{i,\hatt-1:\hatt} \, +\overline{\mathcal{Y}}$ is closed in probability. 
 One can  check that  
 $\mathbf{NCA}(\mathcal{Y})$ implies $\mathbf{NCA}_{\hatt-1,\hatt}(\mathcal{\overline Y})$. Thus, again by Theorem \ref{thm:closedness01}, also
 $K^{\mathcal{\overline Y}}_{\widehat t -1:\widehat t}$ \, is closed in probability.
  But 
\begin{align}
K^{\mathcal{\overline{Y}}}_{\hatt -1:\hatt}&={\sf X}_{i=1}^NK_{i,\hatt-1:\hatt}\, +\overline{\mathcal{Y}}-(L_+^0(\Omega,\mathcal{F},P))^N \notag\\
 &={\sf X}_{i=1}^NK_{i,\hatt-1:\hatt}+ \,  
{\sf X}_{i=1}^NK_{i,\hatt:T}
+\mathcal{Y}-(L_+^0(\Omega,\mathcal{F},P))^N \notag \\
 &={\sf X}_{i=1}^NK_{i,\hatt-1:T}+\mathcal{Y}-(L_+^0(\Omega,\mathcal{F},P))^N =K^{\mathcal{Y}}_{\hatt -1:T} \label{aab4}
\end{align}
showing that $K^{\mathcal{Y}}_{\hatt -1:T}$ is closed in probability. Similarly, as ${\sf X}_{i=1}^NK_{i,\hatt-1:\hatt}+\overline{\mathcal{Y}}={\sf X}_{i=1}^NK_{i,\hatt-1:T}+\mathcal{Y}$, we have that $$  {\sf X}_{i=1}^NK_{i,\hatt-1:T}+\mathcal{Y}$$ is also closed in probability. Thus by
backwards induction
$$ K^{\mathcal{ Y}}_{0:T} \text{ and } {\sf X}_{i=1}^NK_{i,0:T}+\mathcal{Y}$$ are both closed in probability.

As to the equivalence \eqref{eq:ncaiffoneperiod4}, we claim that 
\begin{equation}
\label{eq:ncaiffsplit4}
\mathbf{NCA(\mathcal{Y})}\Longleftrightarrow \begin{cases}
\mathbf{NA}_{i,\hatt:T}\,\,\,\forall \,i=1,\dots,N,\\
\mathbf{NCA}_{\hatt-1:\hatt}(\mcY),\\
\mathbf{NA}_{0:\hatt-1}.
\end{cases}
\end{equation}
Furthermore, by standard No Arbitrage arguments (\cite{FollmerSchied2} Proposition 5.11) it can be shown that
\begin{equation}\label{77}
\begin{cases}
\mathbf{NA}_{i,\hatt:T}\,\,\,\forall \,i=1,\dots,N,\\
\mathbf{NCA}_{\hatt-1:\hatt}(\mcY),\\
\mathbf{NA}_{0:\hatt-1}.
\end{cases}\Longleftrightarrow  \begin{cases}
\mathbf{NA}_{i,t:t+1}\,\,\,\forall \,i=1,\dots,N,\,t=\hatt,\dots,T-1,\\
\mathbf{NCA}_{\hatt-1:\hatt}(\mcY).\\
\mathbf{NA}_{t:t+1}\,\,\,\forall \,t=0,\dots,\hatt-2.
\end{cases}
\end{equation}
From \eqref{eq:ncaiffsplit4} and \eqref{77} the equivalence $\mathbf{NCA(\mathcal{Y})}\Leftrightarrow {\mathbf{NCA}_{t-1:t}(\mathcal{Y})\,\,\forall t=1,\dots, T}$, follows by applying Theorem \ref{TH111} in  one period  for the case $t=\hatt,\dots,T-1$, and Proposition \ref{propNCANA} for $t=0,\dots,\hatt-2$.
We come to the proof of \eqref{eq:ncaiffsplit4}. To begin with, observe that since $\R^N_0\subseteq \mcY$, we have using \eqref{R00} and Item \ref{EA2} in the statement that 
\begin{equation}
\label{eqn:includey04}
 (\mcY_0\cap(L^0(\Omega,\mathcal{F}_{\hatt-1},P))^N)\subseteq \mcY.  
\end{equation}

Let us now show ($\Rightarrow$) in \eqref{eq:ncaiffsplit4}. Assuming $\mathbf{NCA(\mathcal{Y})}$, by inclusion arguments we  have $\mathbf{NCA}_{\hatt-1:\hatt}(\mcY)$ directly, and $\mathbf{NA}_{i,\hatt:T}\,\,\,\forall \,i=1,\dots,N$, using \eqref{NCANA}. To show $\mathbf{NA}_{0:\hatt-1}$, observe that by \eqref{eqn:includey04} we have again by inclusion arguments $$\left({\sf X}_{i=1}^NK_{i,0:\hatt-1}+(\mcY_0\cap(L^0(\Omega,\mathcal{F}_{\hatt-1},P))^N)\right)\cap (L_+^0(\Omega,\mathcal{F},P))^N=\{0\}.$$
This means that $\mathbf{NCA}_{0:\hatt-1}(\mcY_0\cap(L^0(\Omega,\mathcal{F}_{\hatt-1},P))^N)$ holds. Invoking Proposition \ref{propNCANA}, we get that also $\mathbf{NA}_{0:\hatt-1}$ holds.

We now prove ($\Leftarrow$) in \eqref{eq:ncaiffsplit4}, by showing that the contrapositive statement holds. Let us say that for an $N$-dimensional random vector $Z$ property $\mathbf{CA}$ holds if $Z^i\geq 0$ (a.s.) for every $i=1,\dots,N$ and $P(Z^j>0)>0$ for some $j$. Take a collective arbitrage, namely $k\in {\sf X}_{i=1}^NK_{i}, Y\in\mcY$ such that $\mathbf{CA}$ property holds for $(k+Y)$. Based on \eqref{kst4} we write $k=\sum_{t=0}^{T-1}k_{t,t+1}$ for $k_{t,t+1}\in {\sf X}_{i=1}^NK_{i,t:t+1}, t=0,\dots, T-1$. Observe that setting $\hatk=\frac1N\sum_{i=1}^N\sum_{t=0}^{\hatt-2}k^i_{t,t+1}\in K_{0:\hatt-1}$ and taking the vector $\bfone\in \R^N$ with all components equal to $1$ we have $\sum_{t=0}^{\hatt-2}k_{t,t+1}=\hatk\bfone+\sum_{t=0}^{\hatt-2}k_{t,t+1}-\hatk\bfone$ and $\hatY=\sum_{t=0}^{\hatt-2}k_{t,t+1}-\hatk\bfone\in\mcY_0\cap(L^0(\Omega,\mathcal{F}_{\hatt-1},P))^N\subseteq \mcY$ (the latter inclusion is by \eqref{eqn:includey04}). Then 
$$\hatk\bfone+k_{\hatt-1,\hatt}+\sum_{t=\hatt}^{T-1}k_{t,t+1}+(\hatY+Y)=k+Y$$ satisfies property $\mathbf{CA}$, and $(\hatY+Y)\in\mcY$ since the latter is a convex cone.
We now apply an iterative procedure backwards.

\textbf{Case 1: $Z:=\hatk \bfone+k_{\hatt-1,\hatt}+(\hatY+Y)$ does \emph{not} satisfy property $\mathbf{CA}$}. 
Then, by the definition of the property $\mathbf{CA}$, there are two alternatives:
\begin{itemize}
    \item $P(Z^j<0)>0$ for some $j$. In this case, since $\{Z^j<0\}\in\mcF_{\hatt}$, we have $h^j:=1_{\{Z^j<0\}}\sum_{t=\hatt}^{T-1}k^j_{t,t+1}\in K_{j,\hatt:T}$ and $Z^j+\sum_{t=\hatt}^{T-1}k^j_{t,t+1}=k^j + Y^j.$ Thus
    $$h^j=1_{\{Z^j<0\}}\left(Z^j+ \sum_{t=\hatt}^{T-1}k^j_{t,t+1}  \right)-1_{\{Z^j<0\}}Z^j =1_{\{Z^j<0\}}\left(k^j + Y^j \right)-1_{\{Z^j<0\}}Z^j  \geq 0,$$ since $k+Y$ satisfies property $\mathbf{CA}$ and so $k^j+Y^j \geq 0$. Moreover, from $P(Z^j<0)>0$  we obtain $P(h^j>0)>0$, violating then $\mathbf{NA}_{j,\hatt:T}$.
    \item $Z^j\leq 0$ $P$-a.s. for every $j$. In this case, $h^j:=\sum_{t=\hatt}^{T-1}k^j_{t,t+1}\in K_{j,\hatt:T}$ and $h^j=k^j+Y^j-Z^j \geq k^j+Y^j$ so that $h^j$ violates $\mathbf{NA}_{j,\hatt:T}$, for every $j$.
\end{itemize}

\textbf{Case 2: $\hatk\bfone+k_{\hatt-1,\hatt}+(\hatY+Y)$  \emph{does} satisfy property $\mathbf{CA}$}.

\textbf{Case 2.1: $\hatk\bfone$  does \emph{not} satisfy property $\mathbf{CA}$}.
Then there are two alternatives:
\begin{itemize}
    \item $P(\hatk<0)>0$. In this case, since $\{\hatk<0\}\in\mcF_{\hatt-1}$, we have $1_{\{\hatk<0\}}k_{\hatt-1,\hatt}\in {\sf X}_{i=1}^N K_{i,\hatt-1,\hatt}$ and $1_{\{\hatk<0\}}(\hatY+Y)\in \mcY$ by Assumption \ref{EA2} in the statement of the theorem. Thus, $h:=1_{\{\hatk<0\}}k_{\hatt-1,\hatt}+1_{\{\hatk<0\}}(\hatY+Y) \in K^{\mcY}_{\hatt-1,\hatt}$ and
    $$h=1_{\{\hatk<0\}}k_{\hatt-1,\hatt}+1_{\{\hatk<0\}}(\hatY+Y)=1_{\{\hatk<0\}}\left(\hatk\bfone+k_{\hatt-1,\hatt}+(\hatY^j+Y^j)\right)-1_{\{\hatk<0\}}\hatk\bfone \geq 0$$ componentwise, by the condition in Case 2. Moreover, as $P(\hatk<0)>0$, we also have $P(h^j>0)>0$, for each $j$, violating then $\mathbf{NCA}_{\hatt-1:\hatt}(\mcY)$.
    \item $\hatk\leq 0$ $P$-a.s. Note that $h=k_{\hatt-1,\hatt}+(\widehat Y+Y) \in K^{\mcY}_{\hatt-1,\hatt}$ and $h=[\widehat k \mathbf 1 +k_{\hatt-1,\hatt}+(\widehat Y+Y)]-\widehat k \mathbf 1 \geq 0 $ componentwise, by the condition in Case 2 which also implies that some component of $h$ is positive with positive probability. Thus  $h$ violates $\mathbf{NCA}_{\hatt-1:\hatt}(\mcY)$.
\end{itemize}

\textbf{Case 2.2: $\hatk\bfone$   \emph{does} satisfy property $\mathbf{CA}$}. Then it is immediate to see that $\hatk$ violates $\mathbf{NA}_{0:\hatt-1}$.

\end{proof}

As a corollary of Theorem \ref{FTAP3} and Theorem \ref{thm:EASYclosedness01}, we deduce
\begin{theorem}[Collective FTAP]\label{616}
Suppose that $\mathcal{T}=\{0,\dots,T\}$ and that $\mathbb{F}^i=\mathbb{F}^j=\mathbb{F}$ for all $i,j\in\{1,\dots,N\}$.
    Fix $\hatt\in\{1,\dots,T\}$, and let $\mcY \subseteq (L^0(\Omega,\mathcal{F}_{\widehat t},P))^N$ be a vector space containing $\R^N_0$.    If assumptions \ref{EA1} and \ref{EA2} of Theorem \ref{thm:EASYclosedness01} are satisfied then 
\begin{equation*}
 \mathbf{NCA(\mathcal{Y})} \text { iff } (C^\mathcal{Y})^{0}_1  \cap \mathcal{P}^N_e \not = \emptyset
 \end{equation*}
 and 
 \begin{equation}\label{eqCC}  
(C^\mathcal{Y})^{0}_1=\{(Q^1,\dots,Q^N) \in (C^\mathcal{Y})^{0}_1 \mid  Q^i=Q^j \text{ on }\mcF_{\hatt-1}\,\forall\,i,j =1,\dots,N \}.
 \end{equation}
If additionally 
$\mathcal{Y}  \subseteq L^{1 }(\Omega, \mathbf{F}_T , P)$, then 
\begin{equation*}
 \mathbf{NCA(\mathcal{Y})} \text { iff }  \mathcal{M}^{\mathcal{Y}} \cap \mathcal{P}^N_e \not = \emptyset.
 \end{equation*}
\end{theorem}
\begin{proof}
We only need to prove \eqref{eqCC}.
Recall that $e^i, e^j$ are elements in the canonical basis of $\R^N$ and so $e^i-e^j\in\R^N_0$ for all $i,j$.
Observe that  \eqref{eqn:includey04} holds under the assumptions of Theorem \ref{thm:EASYclosedness01}. From \eqref{eqn:includey04} we have that $1_A(e^i-e^j)\in\mcY$ for any given $A\in\mcF_{\hatt-1}$. This in turns implies (since we are assuming here that $\mcY$ is a vector space) that $Q^i(A)=Q^j(A)$ for every $i,j$. 
\end{proof}

\begin{example}
In the example in Section \ref{73} Item \ref{73C} we provide an example of a two periods market with two assets and two agents, where there exists a (global) Arbitrage, $\mathbf{NCA(\mathcal{Y})}$ holds true  and  it is possible, when super-replicating the claims, to profit from cooperation thanks to \eqref{eqCC}.
\end{example}

\section{Examples} \label{Sec:examples}
In the Sections \ref{toyex} and \ref{73}, the set $\Omega$ is finite and therefore, by the results of Section \ref{section:finiteomega}, the convex cone $K^{\mathcal{Y}} $ is closed in probability. At the end of Section \ref{Sec:examples}, in Table \ref{tablerecap}, we recapitulate some key features of the examples presented below.

\subsection{A One-Period Toy Example}
\label{toyex}

The Example in Item \ref{72A} below is possibly the simplest example of a market  where there exists a (global) Arbitrage, $\mathbf{NCA(\mathcal{Y})}$ holds true and  it is possible, in computing  the super-replication  $\pi_+^{\mathcal{Y}}(g)$ of anyone claim, to profit from cooperation.\\
Let $\mathcal T=\{0,1\}$; $|\Omega|=2$ with $P(\omega_1)>0 $ and $P(\omega_2)>0$;  $\mathcal F_0 $ be the trivial $\sigma$-algebra and let $\mathcal {F}^1_1=\mathcal{F}^2_1:=\mathcal{F}_1$ be the $\sigma$-algebra  of all the subsets of $\Omega$. As before, assume zero interest rate. Consider, in addition to the riskless asset, two assets with (discounted) price processes $X^1=(X^1_0,X^1_1)$, $X^2=(X^2_0,X^2_1)$ adapted to $(\mathcal F_0, \mathcal F_1)$ and let $X^i_t>0$ for all $i$ and $t$. The return of each asset $X^i$ is frequently denoted by a pair of real numbers $(d_i,u_i)$, where $0<d_i<1<u_i$, and   $X^i_1:=X^i_0(u_i,d_i)$. Recall also the notation for the following three spaces:\\
$K_1:=\{ H^1(X^1_1-X^1_0) \mid H^1 \in \mathbb R \} $; \quad $K_2:=\{ H^2(X^2_1-X^2_0) \mid H^2 \in \mathbb R \}$;\\
$K:=\{ H^1(X^1_1-X^1_0)+H^2(X^2_1-X^2_0) \mid (H^1,H^2) \in \mathbb R^2 \}$.\\ 
Suppose that $M_1 \cap \mathcal P_e=\{Q^1\}$ and  $M_2 \cap \mathcal P_e=\{Q^2\}$,  so that  $\mathbf{NA}_1$ and $\mathbf{NA}_2$ hold true and each single market is complete, and that $Q^1 \not= Q^2$, so that $M \cap \mathcal P_e=M_1 \cap M_2 \cap \mathcal P_e = \emptyset$, which implies the \textbf{existence of a (global) Arbitrage in the market} $K$. In terms of the market parameters, $Q^1 \not= Q^2$ if and only if 
$$
q_1=\frac{1-d_1}{u_1-d_1}\neq \frac{1-d_2}{u_2-d_2}:=q_2.
$$

This is the case, for example, when 
\begin{equation} \label{numericEx}
X^1_0=2, \text {  } X^1_1=(3,1); \text {  } X^2_0=4, \text {  } X^2_1=(9,3),
\end{equation}
as depicted in Figure \ref{figtree1}, where in this case $(u_1,d_1)=(\frac{3}{2},\frac{1}{2})$, $(u_2,d_2)=(\frac{9}{4},\frac{3}{4})$, $Q^1=(\frac 1 2, \frac 1 2 )$ and $Q^2=(\frac 1 6, \frac 5 6 )$.

\begin{figure}
\begin{center}
\caption{Tree for the stocks $(X^1,X^2)$ at times $t=0,1$.}
\label{figtree1}

\tikzstyle{level 1}=[level distance=2cm, sibling distance=2.5cm,->]

\tikzstyle{bag} = [text width=1.5em, text centered]
\tikzstyle{end} = []

\begin{tikzpicture}[grow=right, sloped]
\node[bag](c1){$2$}
    child {
        node[bag](y12){$1$}        
            edge from parent 
            node[above] {}
            node[below]  {$\blue{1/2}$}
    }
    child {
        node[bag](y11){$3$}        
            edge from parent 
            node[above] {$\blue{1/2}$}
            node[below]  {}
    };

\node[bag](y21) at ([xshift=1.1cm]y11) {$\red{\omega_1}$};
\node[bag](y22) at ([xshift=1.1cm]y12) {$\red{\omega_2}$};
\node[bag](c2) at ([xshift=6cm]c1){$4$}
    child {
        node[bag](z12){$3$}        
            edge from parent 
            node[above] {}
            node[below]  {$\blue{5/6}$}
    }
    child {
        node[bag](z11){$9$}        
            edge from parent 
            node[above] {$\blue{1/6}$}
            node[below]  {}
    };

\node[bag](z21) at ([xshift=1.1cm]z11) {$\red{\omega_1}$};
\node[bag](z22) at ([xshift=1.1cm]z12) {$\red{\omega_2}$};

\end{tikzpicture}
\end{center}
\end{figure}
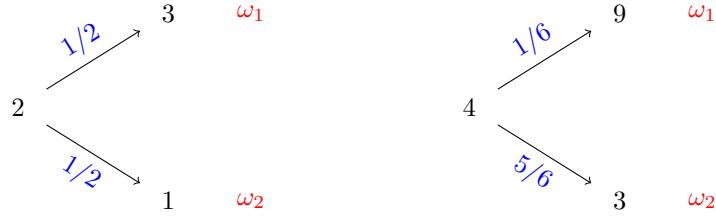

  

We now consider two agents, where agent $i$ is allowed to invest in the single market $K_i$, $i=1,2$. In this setup, we will make several choices for the set $\mcY$ describing different ways the agents may cooperate.

\begin{enumerate}
    \item \label{72A}
(Setting of the Sections \ref{section:finiteomega} and \ref{61}) Take
\begin{equation*}
\mathcal{Y}=\{Y \in (L^{0}(\Omega,\mcF_0,P))^2 = \mathbb R^2 \mid  Y^1 + Y^2 = 0\}.
\end{equation*}
By Theorem \ref{THOne}, $\mathbf{NCA(\mathcal{Y})} \iff (\mathbf{NA_1} $ and $\mathbf{NA_2}$). From the above setup, we thus conclude that $\mathbf{NCA(\mathcal{Y})}$ holds true. Observe that  $\mathcal M^{\mathcal{Y}} \cap \mathcal P^2_e:=\left \{ (Q^1,Q^2)  \right \}$ is not empty and thus we may also deduce that $\mathbf{NCA(\mathcal{Y})}$ holds true from the Collective FTAP Theorem \ref{TH111}.

So now we can apply the pricing-hedging duality from Section \ref{secSuperNew}. 
Each single agent may super-replicate, and actually replicate since each single market is complete, any single contingent claim $f \in L^{1 }(\Omega,\mcF_1,P)$ and  $\rho_{i,+}(f)=\sup_{Q \in M_{i}}E_{Q}[f]=E_{Q^i}[f]$ for all  $f \in L^{1 }(\Omega,\mcF_1,P)$, $i=1,2$.\\
From the pricing-hedging duality in Theorem \ref{propsupermulti} and from \eqref{RR}, \eqref{PP}, \eqref{rhoNpi} we deduce for any $g=(g^1,g^2) \in (L^{1 }(\Omega,\mcF_1,P))^2$
\begin{align}
\mcR(g)&=\sup_{Q \in \mathcal M^{\mathcal{Y}}}  {\sum_{i=1}^2  {E_{Q^i } [g^i]}}=E_{Q^1}[g^1]+E_{Q^2}[g^2]=\ro (g),\\
\pi^N_+(g)&= \max_i \{\pii(g^i) \}=\max \{ E_{Q^1}[g^1] , E_{Q^2}[g^2] \},\\
{\pi_+^{\mathcal{Y}}}(g)&=\frac 1 2 \mcR(g)= \frac 1 2 E_{Q^1}[g^1]+\frac 1 2 E_{Q^2}[g^2]. 
\end{align}
The fact that  $\ro =\mcR$ is obvious, as now  the elements in $\mathcal{Y}$ are vectors of $\R^2$ and so the two definitions coincide. 
The value of $\ro $ is in agreement with intuition, as each agent $i$ is allowed to  (super-) replicate the single claim $g^i$ investing only in the corresponding market and then the total cost for super-replicating both claims must be $E_{Q^1}[g^1]+E_{Q^2}[g^2]$. \\
We also see that it is possible to profit from cooperation since ${\pi_+^{\mathcal{Y}}}(g)$ is always strictly smaller than $\pi^N_+(g)$, except when $E_{Q^1}[g^1]=E_{Q^2}[g^2]$, as can be deduced from the dual formulations above.

To better understand the reason for this, we directly compute  $\pi_+^{\mathcal{Y}}(g)$ and the optimal $Y:=(y,-y)  \in \mcY$, $y \in \mathbb R$. We know that, for  $i=1,2$, $m^i=E_{Q^i}[g^i]$ is the optimizer of $\rho_{i,+}(g^i)$, that is, there exists $k^i \in K_i$ such that $m^i+k^i \geq g^i $ and $m^i$ is the minimal value for which this is possible. Thus we may rewrite the latter two conditions as
\begin{align}
 (m^1-y)+k^1+y &\geq g^1  \label {m1},\\
  (m^2+y)+k^2-y &\geq g^2 \label{m2}.
\end{align}
In order to find ${\pi_+^{\mathcal{Y}}}(g)$, we need to determine one single number, in place of the two numbers $(m^1-y)$ and $(m^2+y)$ for which \emph{both} equations are satisfied.  By selecting $y \in \mathbb R$ such that $ (m^1-y)=(m^2+y)$ we obtain the amount $y=\frac {m_1-m_2} 2=\frac 1 2 E_{Q^1}[g^1]-\frac 1 2 E_{Q^2}[g^2]$ so that  ${\pi_+^{\mathcal{Y}}}(g)= (m_1-y)=(m_2+y)=\frac {m_1+m_2} 2=\frac 1 2 E_{Q^1}[g^1]+\frac 1 2 E_{Q^2}[g^2]$. Comparing with $\eqref{m1}$ and $\eqref{m2}$ we check indeed that
\begin{align}
\left (\frac 1 2 E_{Q^1}[g^1]+\frac 1 2 E_{Q^2}[g^2]\right )+k^1+ \left (\frac 1 2 E_{Q^1}[g^1]-\frac 1 2 E_{Q^2}[g^2]\right ) &\geq g^1 \label {Y1}\\
\left (\frac 1 2 E_{Q^1}[g^1]+\frac 1 2 E_{Q^2}[g^2]\right )+k^2- \left ( \frac 1 2 E_{Q^1}[g^1]-\frac 1 2 E_{Q^2}[g^2]\right ) &\geq g^2 \label {Y2}
\end{align}
which immediately shows the reason why ${\pi_+^{\mathcal{Y}}}(g)=\frac 1 2 E_{Q^1}[g^1]+\frac 1 2 E_{Q^2}[g^2] $ is sufficient to super-replicate  $g^1$ or $g^2$. Therefore, even if each agent $i$ can invest only in market $K_i$, the cooperation between the two agents allows them to exchange the amount $Y^1=y$ with $Y^2=-y$ and reduce the cost to super-replicate  $g^1$ or $g^2$.\\
Observe that the amount ${\pi_+^{\mathcal{Y}}}(g)$ is to be set apart today, so that the the amount $\pi^N_+(g)-{\pi_+^{\mathcal{Y}}}(g) > 0$ is saved today, while the amounts $(Y^1,Y^2)$ are going to be exchanged only tomorrow. In general cases, it is possible - see the Example in Section \ref{73} Item \ref{73C} - that $(Y^1,Y^2)$ will be exchanged tomorrow \emph{only in some scenarios, while in others no exchanges are required}. In these scenarios, the amount saved  $\pi^N_+(g)-{\pi_+^{\mathcal{Y}}}(g) > 0$ at initial time is not compensated by exchanges at later times, which is even a better feature of this approach.
\item \label{72B}
(Setting of the Sections \ref{section:finiteomega} and  \ref{62}). Take 
\begin{equation*}
\mathcal{Y}=\mathcal{Y}_0=\{Y \in \mathbb (L^{1 }(\Omega,\mcF_1,P))^2 \mid  Y^1 + Y^2 = 0\}.
\end{equation*}
By Proposition \ref{propNCANA}, $\mathbf{NCA(\mathcal{Y})} \iff \mathbf{NA}$. From the above setup, we thus conclude that there exists a Collective Arbitrage. Observe that in order that $(Q^1,Q^2) \in \mathcal M^{\mathcal{Y}}$ it is necessary that $E_{Q^1} [Y^1]+E_{Q^2} [Y^2] \leq 0$ for all $(Y^1,Y^2) \in \mcY $, which implies $Q^1=Q^2$. Thus $\mathcal M^{\mathcal{Y}} \cap \mathcal P^2_e= \emptyset  $ and we can also deduce the existence of a Collective Arbitrage from the Collective FTAP Theorem \ref{th51}.

\item \label{72C} (A case where there exists a Global Arbitrage, $\mathbf{NCA(\mathcal{Y})}$ holds true but $( C^{\mcY})^0_1=\mathcal M ^{\mcY}=\emptyset)$.\\ 
Take  $$\mcY :=\left\{(\lambda X^1_1+\mu X^2_1,-\lambda  X^1_1-\mu  X^2_1), \lambda, \mu \in \mathbb R\right\} \subseteq \mcY _0,$$ which is a closed vector space not fulfilling \eqref{RY}.
Recall that we are assuming that $Q^1 \not= Q^2$. Then we have that $E_{Q^2}\left [\frac{X^1_1}{X^1_0}\right ]\not=1$ and $E_{Q^1}\left [\frac{X^2_1}{X^2_0}\right ]\not=1$.
We will show that in this simple setting, the following condition
\begin{equation}\label{star}
E_{Q^2}\left [\frac{X^1_1}{X^1_0} \right ]E_{Q^1}\left [\frac{X^1_2}{X^2_0} \right ]=1,
\end{equation}
discriminate between $\mathbf{NCA(\mathcal{Y})}$ and $\mathbf{CA(\mathcal{Y})}$. Observe also that in terms of the market parameter, \eqref{star} is equivalent to  $u_1/d_1=u_2/d_2$ and that in the numeric example \eqref{numericEx} $u_1/d_1=u_2/d_2=3$.
One may check that if \eqref{star} is not satisfied then $(C^{\mcY})^0=\{(0,0)\}$. In this case, by Proposition \ref{propositionFTAP3} there exists a Collective Arbitrage. Otherwise if \eqref{star} is satisfied then
\begin{align}
    ( C^{\mcY})^0&=\left \{\alpha \left ( z^1 E_{Q^2}\left [\frac{X^1_1}{X^1_0} \right ] ,z^2 \right ) \mid \alpha \in \mathbb R, \, \alpha \geq 0 \right  \}, \label{CY0}\\
    ( C^{\mcY})^0_1&=\mathcal M ^{\mcY}=\emptyset \notag
\end{align}
for $z^i:=\frac{dQ^i}{dP}$. Observe that from the explicit formulation of $( C^{\mcY})^0$ one immediately deduce that $( C^{\mcY})^0_1=\emptyset$, as the expectations under $P$ of two components $ \left ( z^1 E_{Q^2}\left [\frac{X^1_1}{X^1_0} \right ] ,z^2 \right )$ can not be equal, as $E_{Q^2}\left [\frac{X^1_1}{X^1_0}\right ]\not=1$. Since there exists a positive element in $( C^{\mcY})^0$, by Proposition \ref{prop311} we know that $\mathbf{NCA(\mathcal{Y})}$ holds true. To summarise, we showed an example where $\mcY$ does not fulfill \eqref{RY} and where, if $Q^1 \not= Q^2$ and if \eqref{star} is satisfied, there exists a global Arbitrage, $\mathbf{NCA(\mathcal{Y})}$, $\mathbf{NA_1}$ and $\mathbf{NA_2}$ hold true, 
and $( C^{\mcY})^0_1=\mathcal M ^{\mcY}=\emptyset$.

\begin{remark}
    Indeed, the primal approach would lead to the same conclusion as follows: trying to find a 
$\mathbf{CA}(\mcY)$, agent $1$ forms a portfolio $\alpha X^1+c_1$ with value  $\alpha X^1_0+c_1=0$ at time zero and value  $\alpha X^1_T+c_1+Y^1_T>0$ at time $T=1$  using an allowed exchange 
$Y^1_T=\lambda X^1_T+\mu X^2_T$. On her side, agent $2$ forms a portfolio $\beta X^2+c_2$ with value $\beta X^2_0+c_2=0$ at time zero  and value  $\beta X^2_T+c_2 -Y^1_T>0$ at time $T=1$. 
These requirements lead to that the following two random variables should be nonnegative and at least one strictly positive on a set of non-zero probability:
\begin{align}\label{2RVs}
    \alpha(X^1_T-X^1_0)+\lambda X^1_T+\mu X^2_t&\geq 0,\\\label{2RVsBis}
    \beta (X^2_T-X^2_0)-\lambda X^1_T-\mu X^2_T&\geq 0.
\end{align}
Taking expectation under $Q^1$ in the first inequality and expectation under $Q^2$ in the second inequality leads to
\begin{align*}
    \lambda X^1_0+\mu E_{Q^1}[X^2_T]\geq 0,\\
    -\lambda E_{Q^2}[X^1_T]-\mu X^2_0\geq 0.
\end{align*}
If both $\lambda$ and $\mu$ are zero, we know that an arbitrage is not feasible. Assuming for instance that $\lambda\neq 0$, then $\mu/\lambda$ must be between $X^1_0/E_{Q^1}[X^2_T]$ and $E_{Q^2}[X^1_T]/X^2_0$.
Under the additional condition $X^1_0/E_{Q^1}[X^2_T]=E_{Q^2}[X^1_T]/X^2_0$ which is nothing else than \eqref{star}, the two random variables in  \eqref{2RVs}-\eqref{2RVsBis} have zero expectations under equivalent probabilities to $P$ so that there  are zero and cannot form an arbitrage.
\end{remark}

\item \label{72D}(A case where $\mcY$ does not fulfill \eqref{RY}; $\rho^{\mcY}_+=-\infty$; and where $\rho^{\mcY}_+<\rho^{N}_+$ and $\pi^{\mcY}_+<\pi^{N}_+$).\\
We use the same setting of the previous Item \ref{72C} and suppose that \eqref{star} is satisfied and so $\mathbf{NCA(\mathcal{Y})}$ holds true and $( C^{\mcY})^0_1=\mathcal M ^{\mcY}=\emptyset$. Regarding the pricing-hedging duality, we can not use \eqref{suprho}, but we apply  \eqref{rhodualgeneral} to deduce, using \eqref{CY0}, that
\begin{align}
\mcR(g) &=\inf\left\{\sum_{i=1}^2  m^i  \mid m\in\R^2 ,\sum_{i=1}^2 m^i E[z^i] \geq \sum_{i=1}^2  E[z^ig^i] \text{  }  \forall z \in (C^\mathcal{Y})^0 \right\} \notag \\
&=\inf\left\{  m^1 + m^2  \mid  E_{Q^2}\left [\frac{X^1_1}{X^1_0} \right ]m^1 +m^2 \geq E_{Q^2}\left [\frac{X^1_1}{X^1_0} \right ]E_{Q^1}[g^1]+E_{Q^2}[g^2] \right \} \notag \\
&=\inf\left\{  m^1 + m^2  \mid \beta m^1 +m^2 \geq \beta a + c \right \}=-\infty, \label{83}
\end{align}
for $\beta:=E_{Q^2}\left [\frac{X^1_1}{X^1_0} \right ] \not= 1$, $a:=E_{Q^1}[g^1]$, $c:=E_{Q^2}[g^2]$. To show \eqref{83}, we may assume w.l.o.g. that $\beta<1$ and take, for $\delta \in \R$, $m^1:=a-\delta$, $m^2:=\beta \delta +c$ and compute $m^1+m^2=a+c+\delta (\beta-1) \to -\infty$ if $\delta \to \infty$.   As $\rho^{\mcY}_+=-\infty$, we see that the requirement that $\mcY$ fulfills \eqref{RY} in Proposition \ref{PropertiesRho} is necessary for Item 1 to hold, that is for $\rho^{\mcY}_+(0)=0$. 
As $\rho^{\mcY}_+(g)=-\infty$ and $\rho^{N}_+(g)=E_{Q^1}[g^1]+E_{Q^2}[g^2] \in \R$ for all $g \in L^{0} \times L^{0} $ this is an (extreme) example where $\rho^{\mcY}_+(g)<\rho^{N}_+(g)$ for all $g \in L^{0} \times L^{0}$ (recall that here $\Omega$ is finite). Computing $\pi^{\mcY}_+(g)$ using the formulae \eqref{supgen1} and \eqref{CY0} one obtains $$\pi^{\mcY}_+(g)=\gamma E_{Q^1}[g^1]+(1-\gamma)E_{Q^2}[g^2] \, \text{ for } \, \gamma=\frac{\beta}{1+\beta} \in (0,1) \setminus \left \{\frac{1}{2} \right\},$$ so that $\pi^{\mcY}_+(g)<\pi^{N}_+(g)=\max(E_{Q^1}[g^1],E_{Q^2}[g^2])$ whenever $E_{Q^1}[g^1] \not= E_{Q^2}[g^2]$. 

\item \label{72E}(Relation between $\mathbf{NCA(\mathcal{Y})}$  and $\mathbf{NCA}(\mathcal{Y}+\mathbb {R}^N_0)$).\\
We use the same setting of the previous Item \ref{72C} and  consider the convex cone $\widehat{\mcY}:=\mcY+\mathbb {R}^N_0$ which, differently from $\mcY$, fulfills \eqref{RY}. It can be checked that $(  C^{\widehat{\mcY}})^0_1=(C^{\mcY})^0_1$.
Thus, under \eqref{star}, $(C^{\widehat{\mcY}})^0_1=\emptyset $ which then implies, by the last statement in Proposition \ref{propositionFTAP3}, the existence of a Collective Arbitrage for $\widehat{\mcY}$. This proves that $\mathbf{NCA(\mathcal{Y})}$  and $\mathbf{NCA}(\mathcal{Y}+\mathbb {R}^N_0)$ are not in general equivalent. 
Observe that when there exists a Collective Arbitrage for $\widehat{\mcY}$ we can not apply Proposition \ref{PropertiesRho} Item 1 and conclude that $\rho^{\widehat{\mcY}}_+(0)=0$. Indeed, by Proposition \ref{PropertiesRho} Item 5, we know that  $\rho^{\widehat{\mcY}}_+=\rho^{\mcY}_+=-\infty.$

\begin{remark} From the primal point of view, using $\widehat{\mcY}=\mcY+\mathbb {R}^N_0$, the
    inequalities
\eqref{2RVs}-\eqref{2RVsBis} become
\begin{align}\label{2RVs+}
    \alpha(X^1_T-X^1_0)+\lambda X^1_T+\mu X^2_t+c&\geq 0,\\\label{2RVsBis+}
    \beta (X^2_T-X^2_0)-\lambda X^1_T-\mu X^2_T-c&\geq 0,
\end{align}
for some $c\in \mathbb {R}$.
As before, taking expectations gives
\[
-\lambda X^1_0-\mu E_{Q^1}[X^2_T]\leq c\leq -\lambda E_{Q^2}[X^1_T]-\mu X^2_0.
\]
We cannot have $E_{Q^1}[X^2_T]=X^2_0$ or $E_{Q^2}[X^1_T]=X^1_0$, and by condition \eqref{star}, one of the two inequalities above needs to be strict. Suppose that $E_{Q^1}[X^2_T]> X^2_0$. Then choose $\lambda=0$ and for $\mu$ large enough there exists $c$ such that $-\mu E_{Q^1}[X^2_T]< c< -\mu X^2_0$, allowing for an arbitrage. Therefore, for a market satisfying condition \eqref{star} we conclude, as in the dual approach, that $\mathbf{ NCA(\mcY)}$ {\bf and } $\mathbf{ CA(\hat\mcY)}$ hold.
\end{remark}

\end{enumerate}

\subsection{Two-period Example with Exchanges $Y$ that are  $\mathcal{F}_1-$measurable}\label{73}

We provide in Item \ref{73C} below an example of a two-period market with two assets and two agents, where there exists a (global) Arbitrage, $\mathbf{NCA(\mathcal{Y})}$ holds true  and  it is possible to profit from cooperation, when super-replicating two claims.\\

Let $\mathcal T=\{0,1,2\}$ and assume zero interest rate. Consider, in addition to the riskless asset, two assets with (discounted) price processes $X^1=(X^1_0,X^1_1,X^1_2)$, $X^2=(X^2_0,X^2_1,X^2_2)$ and let $\mathbb F^i=(\mathcal F^i_t)_{t \in \mathcal T} $ be the filtration generated by $X^i$.  We assume that $\mathbb F^1 =\mathbb F^2 $, so we set $\mathcal {F}^1_t=\mathcal{F}^2_t:=\mathcal{F}_t$ for all $t$, and that $\mathcal F_0 $ is the trivial $\sigma$-algebra.   We let $ \mathcal H $ be the class of one dimensional predictable processes. We then consider the following three spaces of terminal gains from trading strategies:\\
$K_1:=\{ H^1_1(X^1_1-X^1_0)+  H^1_2(X^1_2-X^1_1)\mid H^1 \in \mathcal H \} $; \quad $K_2:=\{ H^2_1(X^2_1-X^2_0)+H^2_2(X^2_2-X^2_1) \mid H^2 \in \mathcal H \}$;\\
$K:=\{ (H^1 \cdot X^1)_2+ (H^2 \cdot X^2)_2 ) \mid (H^1,H^2) \in \mathcal H \times \mathcal H \}$.\\ 
We now consider two agents, where agent $i$ is allowed to invest in the single market $K_i$. 
Suppose that $|M_1 \cap \mathcal P_e|=|M_2 \cap \mathcal P_e|=\infty $,  so that  $\mathbf{NA_1}$ and $\mathbf{NA_2}$ hold true and each market is incomplete. 
Each 
agent $i$ may super-replicate any 
contingent claim $f \in L^{1 }(\Omega,\mcF_2,P)$ and the formula $\rho_{i,+}(f)=\sup_{Q \in M_{i}}E_{Q}[f]$ for all  $f \in L^{1 }(\Omega,\mcF_2,P)$ holds, $i=1,2$.\\
Suppose additionally, that $M \cap \mathcal P_e=M_1 \cap M_2 \cap \mathcal P_e = \emptyset $, so that there \textbf{exists a (global) Arbitrage in the market} $K$.\\

To be concrete, we provide a numeric example of a market satisfying such conditions.
Take $|\Omega|=6$ with $P(\omega_j)>0 $ for all $j=1,\dots,6$;  let $\mathcal F_0 $ be the trivial $\sigma$-algebra, $\mathcal {F}^1_1=\mathcal{F}^2_1:=\mathcal{F}_1$ the $\sigma$-algebra generated by the subsets $ \{\omega_1,\omega_2\}$, $ \{\omega_3,\omega_4\}$, $ \{\omega_5,\omega_6\}$, and $\mathcal {F}^1_2=\mathcal{F}^2_2:=\mathcal{F}_2$ be the $\sigma$-algebra of all subsets of $\Omega$, respectively. Let
\begin{align}\label{numeric2}
&X^1_0=16, \text {  } X^1_1=(24,16,8), \text {  } X^1_2=(32,16,24,8,12,6); \\
&X^2_0=12, \text {  } X^2_1=(16,12,8), \text {  } X^2_2=(24,8,16,8,6,12), \notag
\end{align}
with the obvious notations, i.e. $X^1_1(\{\omega_1,\omega_2\})=24$, \, $X^1_2(\{\omega_6\})=6$, etc. In other words the stocks follow the tree in Figure \ref{figtree}.

\begin{figure}\begin{center}
\caption{Tree for the stocks $(X^1,X^2)$ at times $t=0,1,2$.}
    \label{figtree}

\tikzstyle{level 1}=[level distance=2cm, sibling distance=2.5cm,->]
\tikzstyle{level 2}=[level distance=1.5cm, sibling distance=1cm,->]

\tikzstyle{bag} = [text width=1.5em, text centered]
\tikzstyle{end} = []

\begin{tikzpicture}[grow=right, sloped]
\node[bag](c1){$16$}
    child {
        node[bag]{$8$}        
            child {
                node[end, label=right:
                    {$6$}](y16) {}
                edge from parent
                node[above] {}
                node[below]  {$\blue{2/3}$}
            }
            child {
                node[end, label=right:
                    {$12$}](y15) {}
                edge from parent
                node[above] {$\blue{1/3}$}
                node[below]  {}
            }
            edge from parent 
            node[above] {}
            node[below]  {$\blue{(1/2)q}$}
    }
    child {
        node[bag]{$16$}        
            child {
                node[end, label=right:
                    {$8$}](y14) {}
                edge from parent
                node[above] {}
                node[below]  {$\blue{1/2}$}
            }
            child {
                node[end, label=right:
                    {$24$}](y13) {}
                edge from parent
                node[above] {$\blue{1/2}$}
                node[below]  {}
            }
            edge from parent 
            node[above] {$\blue{1-q}$}
            node[below]  {}
    }
    child {
        node[bag] {$24$}        
        child {
                node[end, label=right:
                    {$16$}] (y12){}
                edge from parent
                node[above] {}
                node[below]  {$\blue{1/2}$}
            }
            child {
                node[end, label=right:
                    {$32$}](y11) {}
                edge from parent
                node[above] {$\blue{1/2}$}
                node[below]  {}
            }
        edge from parent         
            node[above] {$\blue{(1/2)q}$}
            node[below]  {}
    };

\node[bag](y21) at ([xshift=1.1cm]y11) {$\red{\omega_1}$};
\node[bag](y22) at ([xshift=1.1cm]y12) {$\red{\omega_2}$};
\node[bag](y23) at ([xshift=1.1cm]y13) {$\red{\omega_3}$};
\node[bag](y24) at ([xshift=1.1cm]y14) {$\red{\omega_4}$};
\node[bag](y25) at ([xshift=1.1cm]y15) {$\red{\omega_5}$};
\node[bag](y26) at ([xshift=1.1cm]y16) {$\red{\omega_6}$};
\node[bag](c2) at ([xshift=6cm]c1){$12$}
    child {
        node[bag]{$8$}        
            child {
                node[end, label=right:
                    {$12$}](z16) {}
                edge from parent
                node[above] {}
                node[below]  {$\blue{1/3}$}
            }
            child {
                node[end, label=right:
                    {$6$}](z15) {}
                edge from parent
                node[above] {$\blue{2/3}$}
                node[below]  {}
            }
            edge from parent 
            node[above] {}
            node[below]  {$\blue{(1/2)p}$}
    }
    child {
        node[bag]{$12$}        
            child {
                node[end, label=right:
                    {$8$}](z14) {}
                edge from parent
                node[above] {}
                node[below]  {$\blue{1/2}$}
            }
            child {
                node[end, label=right:
                    {$16$}](z13) {}
                edge from parent
                node[above] {$\blue{1/2}$}
                node[below]  {}
            }
            edge from parent 
            node[above] {$\blue{1-p}$}
            node[below]  {}
    }
    child {
        node[bag] {$16$}        
        child {
                node[end, label=right:
                    {$8$}] (z12){}
                edge from parent
                node[above] {}
                node[below]  {$\blue{1/2}$}
            }
            child {
                node[end, label=right:
                    {$24$}](z11) {}
                edge from parent
                node[above] {$\blue{1/2}$}
                node[below]  {}
            }
        edge from parent         
            node[above] {$\blue{(1/2)p}$}
            node[below]  {}
    };

\node[bag](z21) at ([xshift=1.1cm]z11) {$\red{\omega_1}$};
\node[bag](z22) at ([xshift=1.1cm]z12) {$\red{\omega_2}$};
\node[bag](z23) at ([xshift=1.1cm]z13) {$\red{\omega_3}$};
\node[bag](z24) at ([xshift=1.1cm]z14) {$\red{\omega_4}$};
\node[bag](z25) at ([xshift=1.1cm]z15) {$\red{\omega_5}$};
\node[bag](z26) at ([xshift=1.1cm]z16) {$\red{\omega_6}$};

\end{tikzpicture}
\end{center}
\end{figure}
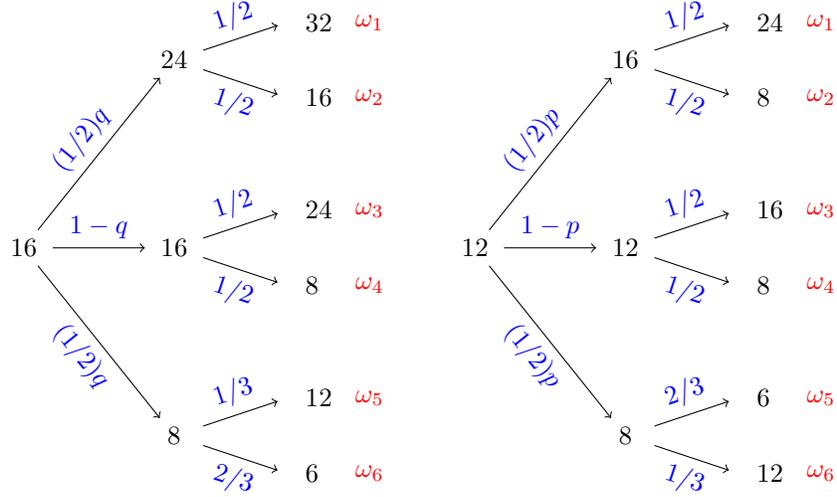
In this case we have
\begin{align*}
M_1 \cap \mathcal P_e&=\left \{ \left ( \frac 1 4 q, \frac 1 4 q, \frac 1 2 (1-q) , \frac 1 2 (1-q), \frac 1 6 q, \frac 2 6 q \right ) \mid 0 < q < 1 \right \}; \\
M_2 \cap \mathcal P_e&=\left \{ \left ( \frac 1 4 p, \frac 1 4 p, \frac 1 2 (1-p) , \frac 1 2 (1-p), \frac 2 6 p, \frac 1 6 p \right ) \mid 0 < p < 1 \right \}.
\end{align*}
Then $\mathbf{NA_1}$ and $\mathbf{NA_2}$ hold true and each single market $K_i$ is incomplete. Moreover, $M \cap \mathcal P_e=M_1 \cap M_2 \cap \mathcal P_e = \emptyset $, so that there \textbf{exists a (global) Arbitrage in the market} $K$.\\

Now we consider two agents, where agent $i$ is allowed to invest in the single market $K_i$. In this setup, we will make several choices for the set $\mcY$ describing different ways the agents may cooperate.

\begin{enumerate}
    \item \label{73A}
(Setting of the Sections \ref{section:finiteomega} and \ref{61}) Take
\begin{equation*}
\mathcal{Y}=\{Y \in (L^{0}(\Omega,\mcF_0,P))^2 = \mathbb R^2  \mid  Y^1 + Y^2 = 0\}.
\end{equation*}
By Theorem \ref{THOne}, $\mathbf{NCA(\mathcal{Y})} \iff (\mathbf{NA_1} $ and $\mathbf{NA_2}$). From the above setup we thus conclude that $\mathbf{NCA(\mathcal{Y})}$ holds true. Observe that  $\mathcal M^{\mathcal{Y}} \cap \mathcal P^2_e:=\left \{ (Q^1,Q^2) \mid  Q^i \in M_i \cap \mathcal P_e \right \}, \, i=1,2 \}$ is not empty and thus we may also deduce that $\mathbf{NCA(\mathcal{Y})}$ holds true from the Collective FTAP Theorem \ref{TH111}.

\item \label{73B} (Setting of the Sections \ref{section:finiteomega} and \ref{62}). Take now 
\begin{equation*}
\mathcal{Y}=\mathcal{Y}_0=\{Y \in \mathbb (L^{1 }(\Omega,\mcF_2,P))^2 \mid  Y^1 + Y^2 = 0\}.
\end{equation*}
By Proposition \ref{propNCANA}, $\mathbf{NCA(\mathcal{Y})} \iff \mathbf{NA}$. From the above setup we thus conclude that there exists a Collective Arbitrage. Observe that in order that $(Q^1,Q^2) \in \mathcal M^{\mathcal{Y}}$ it is necessary that $E_{Q^1} [Y^1]+E_{Q^2} [Y^2] \leq 0$ for all $(Y^1,Y^2) \in \mcY $, which implies $Q^1=Q^2$. From the assumption $M_1 \cap M_2 \cap \mathcal P_e = \emptyset $ we deduce $\mathcal M^{\mathcal{Y}} \cap \mathcal P^2_e= \emptyset  $ and thus the existence of a Collective Arbitrage also follows from the Collective FTAP Theorem \ref{th51}. 

\item \label{73C}
 (Setting of the Sections \ref{section:finiteomega} and \ref{63}).Take now 
\begin{equation*}
\mathcal{Y}=\{Y \in \mathbb (L^{1 }(\Omega,\mcF_1,P))^2 \mid  Y^1 + Y^2 = 0\}.
\end{equation*}
Observe that in order that $(Q^1,Q^2) \in \mathcal M^{\mathcal{Y}}$ it is necessary that $E_{Q^1} [Y^1]+E_{Q^2} [Y^2] \leq 0$ for all $(Y^1,Y^2) \in \mcY $, which implies $Q^1=Q^2$ on $(\Omega, \mathcal F_1$). One may verify thus that
\begin{equation*}
\mathcal M^{\mathcal{Y}} \cap \mathcal P^2_e= \left \{ (Q^1,Q^2) \mid Q^i \in M_i \cap \mathcal P_e, \, i=1,2,  \text { and } Q^1=Q^2 \text { on } (\Omega, \mathcal F_1)  \right \},
\end{equation*}
to be compared with \eqref{eqCC}. Depending on the selection of the price processes $X^1 , X^2$, one obtains either that $\mathcal M^{\mathcal{Y}} \cap \mathcal P^2_e \not=\emptyset$ or $\mathcal M^{\mathcal{Y}} \cap \mathcal P^2_e=\emptyset $.
Assume that $\mathcal M^{\mathcal{Y}} \cap \mathcal P^2_e \not=\emptyset$, which by Theorem \ref{616} or Theorem \ref{th51}, is equivalent to $\mathbf{NCA(\mathcal{Y})}$. We prove below that this assumption holds in the numeric  example \eqref{numeric2}.

Thus, from the pricing-hedging duality in Theorem \ref{propsupermulti} and \eqref{RR} we deduce for any $g=(g^1,g^2) \in (L^{1 }(\Omega,\mcF_2,P))^2$ that
\begin{align*}
\mcR(g)&=\sup_{(Q^1,Q^2) \in \mathcal M^{\mathcal{Y}}}  {\sum_{i=1}^2  {E_{Q^i } [g^i]}} \leq \sup_{(Q^1,Q^2) \in  M_1 \times  M_2} {\sum_{i=1}^2  {E_{Q^i } [g^i]}}  \\
&=\sup_{Q^1 \in  M_1 } E_{Q^1}[g^1]+ \sup_{Q^2 \in  M_2} E_{Q^2}[g^2]=\ro (g),
\end{align*}
and the inequality is in general strict whenever $\mathcal M^{\mathcal{Y}}$ is strictly contained in $ M_1 \times  M_2$. When $\mcR(g)<\ro (g)$,  cooperation thus allows the agents to save money when super-replicating the claims $(g^1,g^2)$. In this case, as $\mcY$ fulfills \eqref{RY}, we also deduce from \eqref{rhoNpi} and \eqref{PP} that $\pi^{\mcY}_+(g)<\pi^{N}_+(g)$. 
 
Coming back to the numeric example \eqref{numeric2}, one can verify that
{\small\begin{equation*}
\mathcal M^{\mathcal{Y}} \cap \mathcal P^2_e= \left \{ \left ( \frac 1 4 q, \frac 1 4 q, \frac 1 2 (1-q) , \frac 1 2 (1-q), \frac 1 6 q, \frac 2 6 q \right ) , \left (\frac 1 4 q, \frac 1 4 q, \frac 1 2 (1-q) , \frac 1 2 (1-q), \frac 2 6 q, \frac 1 6 q \right ) \mid   0 < q < 1 \right \}.
\end{equation*}}
Thus $\mathcal M^{\mathcal{Y}} \cap \mathcal P^2_e$ is not empty  and strictly contained in $( M_1 \times  M_2) \cap \mathcal P^2_e$. If we select the two contingent claims $$g^1=(26,18,24,20,12,9) \quad \text{ and }  \quad g^2=(12,8,6,6,24,18), $$ it is easy to verify that
\begin{align*}
    \rho_{1,+}(g^1)&=\sup_{Q \in M_{1}}E_{Q}[g^1]=22,\\ \quad \rho_{2,+}(g^2)&=\sup_{Q \in M_{2}}E_{Q}[g^2]=16,\\ \quad \ro (g)&=\rho_{1,+}(g^1)+\rho_{2,+}(g^2)=38,
\end{align*}

where the supremum in computing  $\rho_{1,+}(g^1)$ (resp. $\rho_{2,+}(g^2)$) is obtained at $q=0$ (resp. $p=1$). Thus  $\ro (g)=38 $ is the cost to super-replicate both claims $(g^1,g^2)$ without cooperation.  While the cost to super-replicate both claims $(g^1,g^2)$ with cooperation is assigned by:
$$\mcR(g)=\sup_{(Q^1,Q^2) \in \mathcal M^{\mathcal{Y}}}  {\sum_{i=1}^2  {E_{Q^i } [g^i]}}=32,$$ with the supremum over $\mathcal M^{\mathcal{Y}}$ attained at $q=1$.  For $Y^1=-Y^2$, one verifies that one exchange $ (Y^1,Y^2) \in \mathbb (L^{1 }(\Omega,\mcF_1,P))^2$  that allows to obtain $\mcR(g)=32 $ is given by
$$ Y^1=0 \text { on } \{\omega_1,\omega_2\}; \quad   Y^1=6 \text { on } \{\omega_3,\omega_4\}; \quad Y^1=0 \text { on } \{\omega_5,\omega_6\}.$$
Thus no exchanges are required at the nodes $\{\omega_1,\omega_2\}$ and $\{\omega_5,\omega_6\}$ while at the node $\{\omega_3,\omega_4\}$ agent $1$ receives $6$ units of money from agent $2$. In order to super-replicate both claims $(g^1,g^2)$, the cooperation between the two agents allows the two agents to save, at initial time, 6 units of money and requires the two agents to exchange 6 units of money, at later time, only at one 
node. Also observe that $E_{\hat Q^i}[Y^i]=0$ with this choice of $(Y^1,Y^2)$.

\end{enumerate}

{\renewcommand{\arraystretch}{2}
\begin{center}
\begin{table}
\begin{tabular}{||c | c| c| c |c | c| c| c ||} 
 \hline
 $\mathbf{NA}$ & $\mathbf{NCA}(\mathcal{Y})$ & $\mathbf{NCA}(\widehat{\mathcal{Y}})$ &  $\mathbf{NA}_i\,\forall i$ & $M_\mathcal{Y}\neq \emptyset$ & $\pi_+^\mathcal{Y}<\pi_+^N$ & $\mcR<\rho_+^N$ & Example\\  
 \hline\hline
\xmark & \cmark & \cmark &  \cmark & \cmark & \cmark & \xmark & Ex. \ref{toyex} Item \ref{72A}\\   
 \hline
\xmark & \cmark & \cmark &  \cmark & \cmark & \cmark & \cmark & Ex. \ref{73} Item \ref{73C} \\
 \hline
 \xmark & \cmark & \xmark &  \cmark & \xmark & \cmark & \cmark & Ex. \ref{toyex} Items \ref{72C} - \ref{72E} \\
 \hline
 \end{tabular}
\caption{Recap on some examples. We set $\widehat{\mcY}=\mcY+\R^N_0$.}
\label{tablerecap}
 \end{table}
\end{center}}

\section{Appendix}

\subsection{Finitely Generated Convex Cones}\label{app3}
Let $V$ be a topological vector space. By definition, a convex cone $S \subseteq V$ is finitely generated if 
\begin{equation*}
    S=\left \{ \sum _{i=1}^M \lambda _i v_i \mid \lambda _i \geq 0 \text{ for each } i \right \}
\end{equation*}
for some finite number of vectors $v_1, \dots , v_M$ in $V$. Any finite dimensional vector space $\mathbb K \subseteq V$ is given by $$ \mathbb K = \left \{ \sum _{i=1}^N \alpha_i k_i \mid \alpha_i \in \mathbb R \right \} $$ for some finite number of vectors $k_1, \dots , k_N \in V$. Then the vector space $\mathbb K$ is the finitely generated convex cone generated by the vectors $k_1, \dots , k_N, -k_1 , \dots , -k_N$, as $$ \mathbb K = \left \{ \sum _{i=1}^N \alpha_i k_i + \sum _{i=1}^N \beta_i (-k_i) \mid \alpha_i \geq 0, \beta_i \geq 0 \text{ for each } i \right \}. $$
Clearly, the sum of two (or a finite number of) finitely generated convex cones is a finitely generated convex cone. Indeed if
$ S^1=\left \{ \sum _{i=1}^{M_1} \lambda^1 _i v^1_i \mid \lambda^1 _i \geq 0 \text{ for each } i \right \}$ and $ S^2=\left \{ \sum _{j=1}^{M_2} \lambda^2 _j v^2_j \mid \lambda^2 _j \geq 0 \text{ for each } j \right \}$, then $$S^1+S^2=\left \{ \sum _{i=1}^{M_1} \lambda^1 _i v^1_i + \sum _{j=1}^{M_2} \lambda^2 _j v^2_j \mid \lambda^1 _i \geq 0, \lambda^2 _j \geq 0 \text{ for each } i \text { and } j \right \}. $$
\subsection{Properties of Full Exchanges}
\begin{lemma}
    \label{lemma:exaust} 
    Suppose that $\mathbb F^i=\mathbb F^j=\mathbb F$ for every $i,j=1,\dots,N$.
    Let $k\in K$ (see Definition \ref{NAclassic}). Then there exist $k^i\in K_i, i=1,\dots,N$ (see \eqref{def:Ki}) such that $k=\sum_{i=1}^Nk^i.$
\end{lemma}
\begin{proof}
Let $k=(H\cdot X)_T$, where $H=(H^1,\dots,H^J)$. Then
$$k=\sum_{a\in (1)}(H^a\cdot X^a)_T+\sum_{a\in (2)\setminus (1)}(H^a\cdot X^a)_T+\sum_{a\in (3)\setminus ((1)\cup (2))}(H^a\cdot X^a)_T+\dots$$
where summations over empty sets are set to zero. Since $\mathbb F^i=\mathbb F^j=\mathbb F$ for every $i,j$,  $k^1:=\sum_{a\in (1)}(H^a\cdot X^a)_T\in K_1$, $k^2:=\sum_{a\in (2)\setminus (1)}(H^a\cdot X^a)_T\in K_2$ and so on, we have the desired decomposition.
\end{proof}
\begin{lemma}
\label{auxNCAimplNA}
  Suppose that $\mathbb F^i=\mathbb F^j=\mathbb F$ for every $i,j=1,\dots,N$,  and that $\mathcal{F}_T=\mathcal{F}$. Then for $K$ introduced in Definition \ref{NAclassic} and $\mathcal{Y}_0$ as in \eqref{Y00} we have 
  \begin{equation}
       \label{NCAandNAfullgroup}
       {\sf X}_{i=1}^NK_i+\mathcal{Y}_0={\sf X}_{i=1}^NK+\mathcal{Y}_0.
       \end{equation}
\end{lemma}
\begin{proof}
    The inclusion $\subseteq $ is straightforward  since $K_i\subseteq K$ for every $i$ by assumption on the filtrations. To show $\supseteq$ take $k^1,\dots, k^N\in K$, so that $(k^1,\dots,k^N)\in {\sf X}_{i=1}^NK$, and $Y\in\mcY_0$. By applying Lemma \ref{lemma:exaust} to every $k^j,j=1,\dots, N$,  there exist $k^j_i\in K_i,i=1,\dots, N$ such that $k^j=\sum_{i=1}^Nk^j_i$.    
    Defining $\widehat{k}^i:=\sum_{j=1}^Nk^j_i$, we have $\widehat{k}^i\in K_i$ by assumption on the filtrations. Setting  $\widehat{Y}^i:=k^i-\widehat{k}^i$, one can verify that $\sum_{j=1}^N\widehat{Y}^j= \sum_{j=1}^Nk^j-\sum_{j=1}^N\widehat{k}^j=0$, i.e. $\widehat{Y}\in\mcY_0$. This implies that $k=\widehat{k}+\widehat{Y}$, and $k+Y=\widehat{k}+(\widehat{Y}+Y)\in{\sf X}_{i=1}^NK_i+\mathcal{Y}_0$.
\end{proof}

\subsection{The Multi-dimensional Kreps-Yan Theorem}

We assume that $\mathbf{P}:=(P^1, \dots , P^N)$, with $P^i  \in \mathcal{P}_e$ for all  $i$, and so all inequalities among random variables are meant to hold $P$-a.s..
Recall also the notation \ref{productnotation} for $L^{1 }(\Omega, \mathbf{F}_T,\mathbf{P})$. 

\begin{theorem}[Multi-dimensional Kreps-Yan Theorem]\label{KY}
Suppose that $G \subseteq L^{1 }(\Omega, \mathbf{F}_T,\mathbf{P})$ is a closed convex cone satisfying $-L^{1 }_+(\Omega, \mathbf{F}_T,\mathbf{P}) \subseteq G$ and $G \cap L^{1 }_+(\Omega, \mathbf{F}_T,\mathbf{P})=\{0\}$. Then there exists  $z \in L^{\infty }(\Omega, \mathbf{F}_T,\mathbf{P})$ satisfying $z^i>0 $ for all $i$ and  $ \sum_{i=1}^N E_{P^i}[z^i g^i ] \leq 0$ for all $g \in G $. 
\end{theorem}

\begin{lemma}\label{L2}
If $G \subseteq L^{1 }(\Omega, \mathbf{F}_T,\mathbf{P})$  then the set 
\begin{equation*}
\mathcal Z := \left \{ z \in L^{\infty }(\Omega, \mathbf{F}_T,\mathbf{P}) \mid 0 \leq z^i \leq 1,  \text { for all } i, \text { and } \text {  }  \sum_{i=1}^N E_{P^i}[z^i g^i ] \leq 0 \text { for all } g \in G \right \}
\end{equation*}
is countably convex.
\end{lemma}

\begin{proof}
Take a sequence $ z_n \in \mathcal Z$ and set $z_\star := \sum_{n=1}^{\infty} \alpha_n z_n $, for $\alpha_n \geq 0$ for each $n$ and  $\sum_{n=1}^{\infty} \alpha_n =1$.
Clearly $0 \leq z_{\star}^i \leq 1$ for all $i$. Take any $g \in G$. Then we need only to prove that $\sum_{i=1}^N E_{P^i}[z_{\star}^i g^i ] \leq 0$. Set $g_{m}^i =\sum_{n=1}^{m}  \alpha_n z^i_n g^i$. Note that $g_{m}^i  \rightarrow z_{\star}^i g^i $  a.s., for each $i$, as $m \rightarrow \infty $, and that $|g_{m}^i |=  | \sum_{n=1}^{m}  \alpha_n z^i_n g^i | \leq   \sum_{n=1}^{m}  \alpha_n  |g^i | \leq | g^i | \in L^{1 }(\Omega, \mathcal{F}_T^{i},P^i)$. Thus by dominated convergence,  $g_{m}^i  \rightarrow z_{\star}^i g^i $ in $L^{1 }(\Omega, \mathcal{F}_T^{i},P^i)$, for each $i$, as $m \rightarrow \infty $. From $\sum_{i=1}^N E_{P^i}[z_{n}^i g^i ] \leq 0$ for each $n$ we obtain
\begin{equation*}
\sum_{i=1}^N E_{P^i}[z_{\star}^i g^i ]=\sum_{i=1}^N \lim_m E_{P^i}[g_{m}^i]=\lim_m \sum_{n=1}^m \alpha_n \sum_{i=1}^N E_{P^i}[z_{n}^i g^i ] \leq 0.
\end{equation*}
\end{proof}

\begin{lemma}\label{L1}
Under the assumption of Theorem \ref{KY}, for all $j=1, \dots , N$ and for all $A \in {\mathcal F}^j_T$ such that $P(A)>0$,  there exists  $z \in \mathcal Z $ such that  $E_{P^j}[z^j 1_A]>0 $.
\end{lemma}

\begin{proof}
Fix  $j \in \{1, \dots , N\}$ and $A \in {\mathcal F}^j_T $ such that $P(A)>0$ and let $e^j $ be the $j$th vector in the canonical base in $\mathbb R^N$. Then the convex compact set $\{1_A e^j \}$ is disjoint from  the closed convex cone $G$ and the Hahn-Banach Theorem guarantees the existence of $Z \in L^{\infty }(\Omega, \mathbf{F}_T,\mathbf{P})$ such that 
\begin{equation*} 
\sup_{g \in G} \sum_{i=1}^{N} E_{P^i}[Z^i g^i]< \sum_{i=1}^{N} E_{P^i}[Z^i 1_A e^j]=E_{P^j}[Z^j 1_A ].
\end{equation*} 
Since $G$ is a cone, we deduce $\sup_{g \in G} \sum_{i=1}^{N} E_{P^i}[Z^i g^i]=0<E_{P^j}[Z^j 1_A ]$. To show that $Z^i \geq 0 $, for all $i$, suppose by contradiction that there exists $h \in \{1,\dots,N\}$ such that $P(\{ Z^h<0 \})>0$. Take $f^h:=-1_{ \{Z^h < 0 \} }$.  Then $E_{P^h}[Z^h f^h ]>0$.
However, as $-L^{1 }_+(\Omega, \mathbf{F}_T,\mathbf{P}) \subseteq G$, then  $f^he^h \in G$ and, from $\sup_{g \in G} \sum_{i=1}^{N} E_{P^i}[Z^i g^i]=0$, we obtain $E_{P^h}[Z^h f^h ] \leq 0$, that is a contradiction. Define now  $z^i:= \frac {Z^i}  {\sum_{i=1}^{N} || Z^i ||_{\infty} }$ for all $i$. Then $z$ satisfies all the required conditions.
\end{proof}

\begin{proof} 
of the Multi-dimensional Kreps-Yan Theorem.\\
Lemma \ref{L1} guarantees that the set $\mathcal Z$  is nonempty. Set $c:=\sup _{z \in \mathcal Z}  \sum _{i=1} ^ N P(\{z^i >0 \} ) \leq N$. Then there exists $z_n \in \mathcal{ Z} $ such that $\sum _{i=1} ^ N P(\{z_n^i >0 \}) \uparrow c$. Lemma \ref{L2}  guarantees that $z_*:=\sum_{n=1}^{\infty} \frac {1} {2^n} z_n \in  \mathcal Z$. 
By definition of $z_*$, $P(\{z_{*}^i >0\} ) \geq P( \{z_n^i >0 \} )$, for all $i$ and all $n$ and thus
\begin{equation*}
c \geq  \sum _{i=1} ^ N P(\{z_{*}^i >0 \}) \geq  \sum _{i=1} ^ N P(\{z_{n}^i >0 \} ) \uparrow c .
\end{equation*} 
Then necessarily  $c= \sum _{i=1} ^ N P(\{z_{*}^i >0\} )$. Observe that such $z_*$ satisfies all the conditions in the thesis of the theorem, except for the requirements $z_{*}^i>0 $ for all $i$. By contradiction, suppose that there exists $j$ such that $P(\{z_{*}^j =0 \} )>0$ and set $A:= \{ z_{*}^j =0 \}$. Lemma \ref{L1} guarantees the existence of an element $z \in \mathcal Z $ such that $E_{P^j}[z^j 1_A]>0.$ Then $P(\{z^j >0\} \cap \{ z_{*}^j  = 0 \} ) >0$. Since $ \frac 1 2 (z + z_*) \in \mathcal Z $, we obtain
$c  \geq  \sum _{i=1} ^ N P \left ( \{ \frac 1 2 (z^i + z_{\star}^i ) >0 \} \right  )  =\sum _{i=1} ^ N P \left (\{ z_{\star}^i  > 0 \} \cup [ \{ z^i >   0 \} \cap \{ z_{\star}^i  = 0 \} ] \right ) \\
= \sum _{i=1} ^ N [ P  ( \{ z_{\star}^i  > 0 \} ) + P(\{z^i >0\} \cap \{ z_{\star}^i  = 0 \} )  ] =c + \sum _{i=1} ^ N P(\{z^i >0\} \cap \{ z_{\star}^i  = 0 \} )  > c $, which is a contradiction.
\end{proof}

\subsection{Review of some known results}
Let $(\Omega, \mathcal F, Q)$ be a probability space and $\mathcal G \subseteq \mathcal F$ a sub-sigma algebra.  If $X^- \in L^1(\Omega, \mathcal F, Q)$ then the conditional expectation $E_Q[X | \mathcal G]:=E_Q[X^+ | \mathcal G]-E_Q[X^- | \mathcal G]$ is well defined as an extended  $\mathcal G$-measurable random variable satisfying $(E_Q[X | \mathcal G])^- \in L^1(\Omega, \mathcal G, Q)$.
\begin{lemma}\label{Lemma1App}
Let $(\Omega, \mathcal F, Q)$ be a probability space and $\mathcal G \subseteq \mathcal F$ a sub-sigma algebra. Suppose that $X \in L^1(\Omega, \mathcal F, Q)$, $H \in L^0(\Omega, \mathcal G, Q)$, $Y \in L^0(\Omega, \mathcal G, Q)$. If $(HX+Y)^- \in L^1(\Omega, \mathcal F, Q)$ then $$ E_Q[HX+Y | \mathcal G]=HE_Q[X | \mathcal G]+Y.$$
\end{lemma}

\begin{proof}
Set $A_n:=\{|H|+|Y| \leq n \} \in \mathcal G$, observe that  $1_{A_n} \uparrow 1_{\Omega}$, as $n \uparrow \infty$, and compute
\begin{equation*}
1_{A_n}  E_Q[HX+Y | \mathcal G]= E_Q[1_{A_n}HX+1_{A_n}Y | \mathcal G]=E_Q[1_{A_n}HX | \mathcal G]+E_Q[1_{A_n}Y | \mathcal G]=1_{A_n}(HE_Q[X | \mathcal G]+Y).
\end{equation*}
\end{proof}

We present a well knows result (see, for example, \cite{FollmerSchied2} 
Theorem 5.14 condition c)) and we provide a simple proof by induction. 
\begin{lemma}\label{Ek=0}
Let $\mathcal T =\{0, \dots , T \}$, $T \geq 1$, $(\Omega, \mathcal F,(\mathcal F_t)_{t\in\mathcal{T}}, Q)$ be a filtered probability space, $X$ be a $d$-dimensional $((\mathcal F_t)_{t\in\mathcal{T}}, Q)$-martingale and set 
\begin{equation*}
K:=\{(H\cdot X)_T \mid H \text { }  \text {d-dimensional stochastic predictable processes }  \}.
\end{equation*}
If  $k \in K$ satisfies $k^- \in L^{1}(\Omega, \mathcal F_T, Q)$, then
\begin{equation*}
 k \in L^{1}(\Omega, \mathcal F_T, Q) \text { and } E_Q[k]=0.
\end{equation*}
\end{lemma}

\begin{proof}
Take first $T=1$. Let $k=H(X_1-X_0) \in K:=\{H_1(X_1-X_0) \mid H_1 \in (L^{0}(\Omega, \mathcal F_0, P) )^d \}$. Since $k^- \in L^{1}(\Omega, \mathcal F_1, Q)$ then $(H(X_1-X_0))^-\in L^{1}(\Omega, \mathcal F_1, Q)$ and we may apply the Lemma  \ref{Lemma1App} and conclude:
\begin{equation*}
E_Q[k]=E_Q[H(X_1-X_0)]=E_Q [E_Q[H(X_1-X_0) | \mathcal F_0]=E_Q [HE_Q[X_1-X_0| \mathcal F_0]]=0.
\end{equation*}
Suppose that for $T-1$ the statement is true. At time $T$ we have $$K:=\left \{ \sum _{t=1}^{T} \widehat{H}_t(X_t-X_{t-1}) \mid  \widehat{H} \text { d-dimensional predictable processes} \right \}$$
and let $k=(H_T(X_T-X_{T-1})+ \sum _{t=1}^{T-1} H_t(X_t-X_{t-1}) ) \in K$ satisfy $k^- \in L^{1}(\Omega, \mathcal F_T, Q) $. Applying Lemma \ref{Lemma1App} we deduce
\begin{align*}
E_Q[ k | \mathcal F_{T-1}]
&=E_Q [ H_T(X_T-X_{T-1})+ \sum _{t=1}^{T-1} H_t(X_t-X_{t-1}) | \mathcal F_{T-1} ] \\
&=H_TE_Q[ (X_T-X_{T-1})| \mathcal F_{T-1}]+ \sum _{t=1}^{T-1} H_t(X_t-X_{t-1}) = \sum _{t=1}^{T-1} H_t(X_t-X_{t-1}).
\end{align*}
Since $(E_Q[ k | \mathcal F_{T-1}])^- \in L^{1}(\Omega, \mathcal F_{T-1}, Q)$, we also have that $\left (\sum _{t=1}^{T-1} H_t(X_t-X_{t-1}) \right )^- \in L^{1}(\Omega, \mathcal F_{T-1}, Q) $ and by the induction hypothesis, $\sum _{t=1}^{T-1} H_t(X_t-X_{t-1}) \in L^{1}(\Omega, \mathcal F_{T-1}, Q) $ and $E_Q \left [ \sum _{t=1}^{T-1} H_t(X_t-X_{t-1}) ) \right ]=0$. We can now conclude that

\begin{equation*}
E_Q[k]=E_Q [E_Q[k | \mathcal F_{T-1}]=E_Q \left [\sum _{t=1}^{T-1} H_t(X_t-X_{t-1})| \mathcal F_{T-1} \right ]=0.
\end{equation*}
\end{proof}
\begin{corollary} \label{CoMMM} Recall the definition of $M_i$ given in \eqref{MMM}, and that $X$ is integrable under $P$.
\begin{equation*}
\label{MartingaleMeasures1}
\mi=\widehat M_i :=\left\{ Q \in \mathcal{P}_{ac} \mid \frac {dQ} {dP} \in L^{\infty }(\Omega, \mathcal{F}_T^{i},P)  \text { and } E_Q[k]= 0  \text {  } \forall k \in \ki \cap L^{1 }(\Omega, \mathcal{F}_T^{i},P) \right\}.
\end{equation*}    
\end{corollary}
\begin{proof}   
To show that $\widehat M_i \subseteq \mi $, take for any $t=1, \dots , T,$ any $A \in \mathcal F^i_{t-1} $ and any $h \in (i)$,  $k_i:=1_A(X^h_t-X^h_{t-1}) \in K_i $. As $X^h$ is integrable under $P$, $k_i \in K_i  \cap L^{1}(\Omega, \mathcal F^i_T, P) $. If $Q \in \widehat M_i$ then $E_Q[k_i]=0$, which implies that $Q \in \mi$.
To show that $ \mi  \subseteq \widehat M_i  $, take $Q \in M_i$ and $k_i \in K_i  \cap L^{1}(\Omega, \mathcal F^i_T, P) $. As $\frac {dQ} {dP} \in L^{\infty }(\Omega, \mathcal{F}_T^{i},P)$, $k_i \in  K_i \cap L^{1}(\Omega, \mathcal F^i_T, Q)$ and Lemma \ref{Ek=0} implies $E_Q[k_i]=0$.
\end{proof}

{
\bibliographystyle{abbrv}
\bibliography{BibAll}
}

\end{document}